\documentclass[10pt, conference, letterpaper]{IEEEtran}

\usepackage[utf8]{inputenc}
\usepackage{hyperref}
\title{Rate Allocation and Content\\
Placement in Cache Networks}

\author{\IEEEauthorblockN{Khashayar Kamran, Armin Moharrer, Stratis Ioannidis, and Edmund Yeh}
\IEEEauthorblockA{\textit{Electrical and Computer Engineering, Northeastern University, Boston, MA, USA} \\
\{kamrank, amoharrer, ioannidis, eyeh\}@ece.neu.edu}
}

\newcommand{\itemcat}{\ensuremath{\mathcal{I}}}
\newcommand{\requestset}{\ensuremath{\mathcal{N}}}
\newcommand{\requestindex}{\ensuremath{n}}
\newcommand{\lbsbconstset}{\ensuremath{\mathcal{J}}}
\newcommand{\lbsbconstindex}{\ensuremath{j}}

\newcommand{\firstsecond}[3] {\ifthenelse{\equal{#1}{1}}{#2}{}\ifthenelse{\equal{#1}{2}}{#3}{}}

\newcommand{\fullversion}[2]{#2}

\usepackage{booktabs} 
\usepackage{graphicx}
\usepackage{amsmath}
\usepackage{amssymb}
\usepackage{amsthm}
\usepackage{array}
\usepackage{bm}
\usepackage{comment}
\usepackage{mdwmath}
\usepackage{epsfig}
\usepackage{float}
\usepackage{color}
\usepackage[all,color]{xy}
\usepackage[lofdepth,lotdepth]{subfig}
\usepackage{ragged2e}
\usepackage{subfig}
\usepackage{grffile}
\usepackage{cleveref}
\usepackage{tabulary}
\usepackage{booktabs}
\usepackage[]{footmisc}
\usepackage{url}
\usepackage{cite}
\usepackage[font=small]{caption}

\theoremstyle{plain} 
\newtheorem{thm}{Theorem}

\newtheorem{lem}{Lemma}
\newtheorem{prop}{Proposition}
\newtheorem{corollary}{Corollary}

\theoremstyle{definition}
\newtheorem{defn}{Definition}
\newtheorem{assm}{Assumption}

\usepackage[skip=0pt]{caption}

\DeclareCaptionType{mycapequ}[][List of equations]
\usepackage{algorithm}
\usepackage{algpseudocode}

\usepackage{makecell}

\begin{document}

\maketitle

\begin{abstract}
We introduce the problem of optimal congestion control in cache networks, whereby \emph{both} rate allocations and content placements are optimized \emph{jointly}. We formulate this as a maximization problem with non-convex constraints, and propose  solving this problem via (a)  a Lagrangian barrier algorithm and (b) a convex relaxation. We prove different optimality guarantees for each of these two algorithms; our proofs  exploit the fact that the non-convex constraints of our problem involve DR-submodular functions. 
\end{abstract}

\fullversion{}{}

\section{Introduction}\label{sec:intro}
Traffic engineering and congestion control have played a crucial role in the stability and scalability of communication networks since the early days of the Internet. 
They have been  extremely active  research areas since the seminal work by Kelly et al.  \cite{kelly1998rate},  who studied  optimal rate control subject to link capacity constraints. Formally, given a network $G(\mathcal{V},\mathcal{E})$ with nodes $v\in \mathcal{V}$, links $e\in \mathcal{E}$, and flows $\requestindex \in \requestset$,  Kelly et al.~\cite{kelly1998rate} studied the following convex optimization problem:
\begin{subequations}
\label{prob:kellystyle}
\begin{align} \small
    \max_{\boldsymbol{\lambda}} \quad & \textstyle\sum_{\requestindex \in \requestset} U_{\requestindex}(\lambda_\requestindex) \label{eq:kellyobj}\\
    \text{s.t.} \quad &\rho_e(\boldsymbol{\lambda}) \leq C_e, \quad \forall e\in \mathcal{E},\label{eq:kellyconstr}
\end{align}
\end{subequations}
where $\boldsymbol{\lambda}=[\lambda_\requestindex]_{\requestindex \in \requestset}\in \mathbb{R}_+^{|\requestindex|}$ is the vector of rate allocations $\lambda_\requestindex$, $\requestindex \in \requestset$ across flows, $\rho_e:\mathbb{R}^{|\requestindex|}\to \mathbb{R}_+$, $C_e\in \mathbb{R}_+$ are the loads and  capacities of links $e\in \mathcal{E}$, respectively, and $U_\requestindex:\mathbb{R}_+\to \mathbb{R}$, $\requestindex\in\requestset$, are concave utility functions of  rates. 
Motivated by Kelly et al. \cite{kelly1998rate}, distributed congestion control algorithms solving Prob.~\eqref{prob:kellystyle} are now both numerous and classic \fullversion{\cite{srikant2012mathematics,warland2000fairwindow, low1999optimization}.}{ \cite{srikant2012mathematics,kelly2014stochastic,warland2000fairwindow, low1999optimization,  chiang2007layering}.}

In this work, we revisit this problem in the context of \emph{cache networks}  \cite{rosensweig2010approximate,fofack2012analysis,ioannidis2016addaptive}. Motivated by technologies such as software defined networks \fullversion{\cite{shenker2011futuresdn}}{ \cite{shenker2011futuresdn,kreutz2015sdnsurvey}} and network function virtualization \cite{han2015nfv},  nodes in cache networks are no longer merely static routers. Instead, they are entities capable of storing data, performing computations, and making decisions.  Nodes can thus  fetch user-requested content \fullversion{\cite{zhang2010ndn}}{\cite{jacobson2009ccn, zhang2010ndn}}, or  perform user-specified computation tasks \cite{deco2019kamran,tulino2018comp}, instead of simply maintaining point to point communication sessions. In turn, such functionalities can address the ever increasing interest in running data-intensive applications in large-scale networks, such as machine learning at the edge \cite{dong2018iotinedge}, IoT-enabled health care \cite{mahmud2018healthcare}, and scientific data-intensive computation \fullversion{\cite{ekanayake2008mapreduce}}{\cite{ekanayake2008mapreduce, makhlouf2019iaas}}.

\begin{figure} [!t]
  \centerline{\includegraphics[width = \columnwidth]{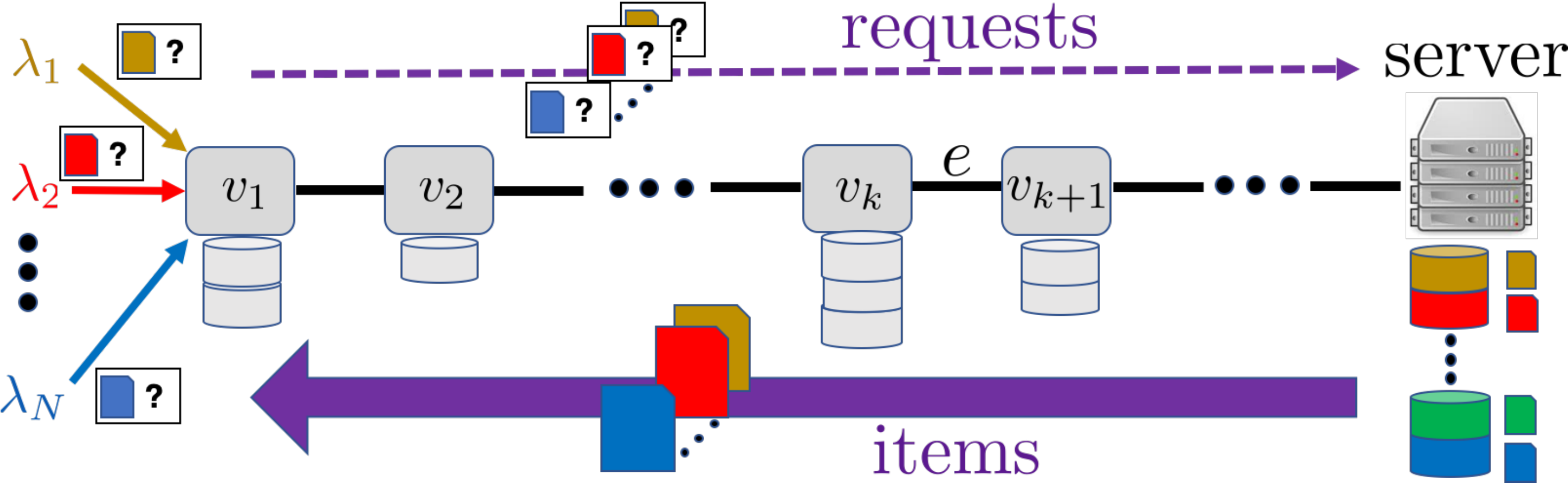}}
  \vspace{1mm}
  \caption{ Example of a cache network.
  $N$ distinct flows of requests  enter the network at node $v_1$. Each flow contains requests for an item $i$ in a catalog $\itemcat$. Requests are forwarded towards the designated server that stores all items in $\itemcat$. Upon reaching it, responses carrying the requested items follow the reverse path towards node $v_1$. However, all intermediate nodes have caches that can be used to store items in $\mathcal{I}$. Thus, requests need not traverse the entire path, but can be satisfied upon the first hit. Hence, the load on  link $e=(v_{k+1},v_k)$ is a function of both rates $\boldsymbol{\lambda}=[\lambda_\requestindex]_{\requestindex=1}^N$  \emph{as well as } cache allocation decisions made by the intermediate nodes $v_1$, \ldots,$v_{k}$, to name a few.} 
  \label{fig:exampleofcache}
\end{figure}

Congestion control in such networks is fundamentally different from the classic setting. When  nodes can store user-requested content or provide server functionalities, network design amounts to determining not only the rate allocations per flow but also the \emph{location} of offered network services. Put differently, to attain optimality, \emph{cache allocation decisions need to be optimized jointly with rate allocation decisions}. In turn, this necessitates the development of novel congestion control algorithms that take cache allocation into account. 

To make this point clear, we illustrate the effect of cache allocations on congestion in a cache network shown in Fig.~\ref{fig:exampleofcache}. Requests for content items in a catalog $\itemcat$ arrive over $N$ flows on a node on the  left of a path network. They are subsequently forwarded towards a designated server on the right, that stores all items in $\itemcat$. Upon reaching the server, responses carrying the requested items are sent back over the reverse path. Assuming that request traffic is negligible, the traffic load on an edge is determined by the item (i.e., response) traffic flowing through it. However, if intermediate nodes are equipped with caches that can store some of the items in $\itemcat$, as in Fig.~\ref{fig:exampleofcache}, requests need not be propagated all the way to the designated server. As a result, the load $\rho_e$ caused by items traversing edge $e=(v_{k+1},v_{k)}$ depends not only on the rate vector $\boldsymbol{\lambda}=[\lambda_\requestindex]_{\requestindex \in \requestset}\in \mathbb{R}_+^{|\requestindex|}$, \emph{but also on the cache allocation decisions made at all nodes $v_1, \ldots, v_k$  preceding $e$ in the path.} For example, the load $\rho_e$ is zero if all items in $\itemcat$ are stored in nodes $v_1,\ldots,v_{k}$.

Formally,  in cache networks, Problem \eqref{prob:kellystyle} becomes:
\begin{subequations}
\label{prob:cachestyle}
\begin{align} \small
    \max_{\boldsymbol{\lambda},\boldsymbol{x}} \quad & \textstyle\sum_{\requestindex \in \requestset} U_{\requestindex}(\lambda_\requestindex)  \label{eq:cacheobj}\displaybreak[0]\\
    \text{s.t.} \quad &\rho_e(\boldsymbol{\lambda},\boldsymbol{x}) \leq C_e, \quad \forall e \in \mathcal{E}, \label{eq:cacheload}\displaybreak[0]\\
    & \textstyle\sum_{i\in \itemcat} x_{vi} \leq c_v,\quad \forall v \in \mathcal{V}, \label{eq:cachestorage}
\end{align}
\end{subequations}
where $\boldsymbol{x}=[x_{vi}]_{v\in V, i\in\itemcat}\in \{0,1\}^{|\mathcal{V}| |\itemcat|}$ is the vector of cache allocation decisions $x_{vi}\in\{0,1\}$, indicating if node $v\in V$ stores $i\in \itemcat$, $c_v\in \mathbb{N}$ is the storage capacity of node $v\in V$, 
and $\boldsymbol{\lambda}$, $\rho_e$, $C_e$, $U_\requestindex$ are respectively the rate allocation vector, loads, link capacities, and utilities, as in Eq.~ \eqref{prob:kellystyle}.
Crucially, the load $\rho_e$ on  links $e\in \mathcal{E}$ is a function of 
\emph{both} the allocated rates \emph{and} cache decisions. As a result, constraints \eqref{eq:cacheload} define a  a non-convex set.  This is not just due to the combinatorial nature of cache allocation decisions $\boldsymbol{x}$: even if $x_{vi}$ are relaxed to real values in $[0,1]$, which corresponds to making probabilistic cache allocation decisions, the resulting constraint \eqref{eq:cacheload} is \emph{still not convex}, and Problem \eqref{prob:cachestyle} \emph{cannot} be solved via standard convex optimization techniques. This is a significant departure from  Problem~\eqref{prob:kellystyle}, in which constraints \eqref{eq:kellyconstr} are linear.

In spite of the challenges posed by the lack of convexity,  \emph{we propose algorithms  solving Problem~\ref{prob:cachestyle} with provable approximation guarantees.}
Specifically:
\begin{enumerate}
    \item We provide a unified optimization formulation for congestion control in cache networks, through joint probabilistic content placement and rate control. To the best of our knowledge, we are the first to study this class of non-convex problems and develop algorithms with approximation guarantees.
    \item We propose two algorithms, each yielding different approximation guarantees. The first
   is a Lagrangian barrier method; the second is a convex relaxation. In both cases, we exploit the fact that constraints \eqref{eq:cacheload} can be expressed in terms of DR-submodular functions \cite{bian2017guaranteed}. Both algorithms and their corresponding analysis are novel and of independent interest, as they may be applicable for attacking problems with DR-submodular constraints beyond the cache network setting we consider here.
    \item Finally, we implement both methods and compare them experimentally to  greedy algorithms over several  real-world and synthetic topologies, observing an improvement in aggregate utility by as much as 5.43$\times$.
\end{enumerate}

The remainder of this paper is structured as follows. We review related work in Section~\ref{sec:rel}. Our network model and problem formulation are discussed in Section~\ref{sec:model}. In Section~\ref{sec:main}, we describe our two different methods for solving the optimal congestion control problem, as well as our performance guarantees. Finally, we present our evaluations in Section~\ref{sec:simul}, and we conclude in Section~\ref{sec:conc}.
\section{Related Work}\label{sec:rel}
\noindent \textbf{Network Cache Optimization.}  Studies on optimal in-network cache allocation are numerous,  roughly split into the \emph{offline} and \emph{online} solutions.
 Several papers study centralized, offline cache optimization in a network modeled as a bipartite graph \fullversion{\cite{approxplacement2008baev,tight2006lisa}}{ \cite{approxplacement2008baev,competetive1995bartal,tight2006lisa}}. Shanmugam et al. \cite{shanmugam2013femto} consider a femto-cell network,  where content is placed in caches to reduce the cost of fetching data from a base station. They do not consider congestion, and study routing costs that are linear in the traffic per link. 
 Mahdian et al. \cite{mahdian2019kelly} model every link with an M/M/1 queue, and consider objectives that are (non-linear) functions of the queue sizes. Similarly, Li and Ioannidis \cite{universally2020Li} model every link with an M/M/1c queue to capture the consolidation of identical responses before being forwarded downstream. The same problem was also studied, albeit in a different model, by Dehghan et al.~\cite{dehghan2015complexity}.
Online cache allocation algorithms exist, e.g.,  for maximizing throughput \cite{vip2014yeh, deco2019kamran}, or minimizing delay \cite{mahdian2018mindelay}. Ioannidis and Yeh \cite{ioannidis2016addaptive} study a similar problem as Shanmugam et al. \cite{shanmugam2013femto} for networks with arbitrary topology and linear link costs, seeking cache allocations that minimize routing costs across multiple hops. \fullversion{}{The same authors extend this work to jointly optimizing cache and routing decisions \cite{ioannidis2017routing}.}

Although we too consider offline algorithms, we depart substantially from prior work. First, all mentioned papers assume the input request rates are fixed, whereas we consider joint  cache allocation \emph{and} rate control. Prior works on  allocation  minimizing  costs \fullversion{\cite{shanmugam2013femto,ioannidis2016addaptive,mahdian2019kelly,universally2020Li,dehghan2015complexity}}{\cite{shanmugam2013femto,ioannidis2016addaptive, ioannidis2017routing,mahdian2019kelly,universally2020Li,dehghan2015complexity}} cast the problem as a sub-modular maximization problem subject to matroid constraints, for which a $(1-1/e)$-approximate solution can be constructed in polynomial time. 
Instead, akin to Kelly et al.~\cite{kelly1998rate}, we treat loads on links as \emph{constraints} rather than part of the objective. Hence, we cannot directly leverage sub-modular maximization techniques and need to design altogether new algorithms. Moreover, 
works which consider congestion \cite{mahdian2019kelly,universally2020Li,dehghan2015complexity} assume that the system is stable when all caches are empty; in fact, finding a cache allocation under which the system is stable is left open. We partially resolve this, jointly finding a rate and an allocation that ensure stability.

\fullversion{}{
Similar to this paper, many consider fixed routing for requests. Although one can incorporate routing decisions into the problem, not doing that does not make the optimal cache allocation problem trivial, and the problem is still NP-complete
\cite{shanmugam2013femto, ioannidis2016addaptive, mahdian2019kelly, universally2020Li, mahdian2018mindelay, dehghan2015complexity, complexity2016tassiulas, optimalcache2013wang}. }

\noindent \textbf{TTL caches}.  Time-to-Live (TTL) caches  providing an elegant general framework for analyzing cache replacement policies. In TTL caches, a timer is assigned to each content, and an eviction occurs upon timer expiration. Multiple studies analyzed TTL caches as approximations to popular cache eviction policies \fullversion{(see \cite{hierarchical2002che, versatile2012fricker,unified2014martina, fofack2012analysis})}{(see \cite{hierarchical2002che, versatile2012fricker,unified2014martina, check2013bianchi, fofack2012analysis, exactttl2014berger})}. TTL cache optimization includes  maximizing the cache hit rate \cite{ferra2016tailed} and the aggregate utility of cache hits \cite{dehghan2019utilitymaxcache,panigrahy2017ttledge,panigraphy2017ttlhitrate}.  In contrast to our approach, however, works on TTL caching do not provide a solution for joint cache and rate allocation, and do not guarantee network stability.  Furthermore, they focus on the utility of cache hits, whereas we consider rate utility, similar to Kelly \cite{kelly1998rate}.

\noindent \textbf{Rate admission control in cache networks.} Various  methods have been proposed for rate admission control in Content-Centric Networking (CCN) \cite{jacobson2009ccn} and Named Data Networking (NDN) \cite{zhang2010ndn} architectures, primarily using congestion feedback from the network for rate control \cite{rozhnova2012hopbyhop, ren2016survey, milad2016mircc, badov2014conaware, tanaka2016concon}. In contrast to our work, none of these come with optimality guarantees. Closer to us, Carofiglio et al. \cite{giovani2013optimal}  fix a cache allocation  and maximize rate utility via rate control;   we depart by jointly optimizing rate and cache allocations.

\fullversion{}{\noindent \textbf{Non-Convex Optimization Techniques.} Our analysis employs the Lagrangian barrier algorithm proposed by Conn et al. \cite{conn1997globally} along with a trust-region algorithm proposed by the same authors \cite{trustRegion}. This is a modified version of the barrier method \cite{fiacco1967} which explicitly deals with numerical difficulties arising from maximizing an unconstrained barrier function. For general, non-convex optimization problems, the aforementioned Lagrangian barrier and the trust-region algorithms come with no optimality guarantees. One of our main technical contributions is to provide such guarantees for our problem by exploiting the fact that the non-convex constraints involve  DR-submodular functions \cite{bian2017guaranteed}. }

\noindent \textbf{DR-submodular optimization.} Since their introduction by Bian et al. \cite{bian2017guaranteed}, DR-submodular functions have received much attention \fullversion{\cite{bian2017dr,crawford2019oracle,iyer2013submodularconst}}{\cite{bian2017dr,crawford2019oracle,iyer2013submodularconst, mokhtari2018stochastic,hassani2019stochastic,mokhtari2018decentralized,pmlr-v97-bian19a}} as examples of  functions which can be maximized with performance guarantees, in spite of the fact they are not convex. 
Bian et al. \cite{bian2017guaranteed} 
propose a constant-factor approximation algorithm for (a) maximizing  monotone DR-submodular functions subject to down-closed convex constraints and (b) maximizing non-monotone DR-submodular functions subject to box constraints. In a follow-up paper, Bian et al. \cite{bian2017dr} provide a constant-factor algorithm for  maximizing non-monotone continuous DR-submodular functions under general down-closed convex constraints. These works, however, do not consider DR-submodular functions  in constraints, rather than in the objective.
In a combinatorial setting, Crawford et al. \cite{crawford2019oracle} and Iyer et al. \cite{iyer2013submodularconst} provide approximate greedy algorithms for minimizing a submodular function subject to a single threshold constraint involving a submodular function.  None of the above solutions, however, is applicable to our problem, which involves maximizing a concave  function subject to \emph{multiple} DR-submodular constraints. To the best of our knowledge, we are the first to study this problem and provide solutions with optimality guarantees.

\section{System Model}\label{sec:model}

We consider a network of caches, each capable of storing  items such that the number of stored items cannot exceed a finite cache capacity. Requests are routed over fixed (and given) paths and are satisfied upon hitting the first cache that contains the requested item. Our goal is to determine the (a) items to be stored at each cache as well as (b) the request rates, so that the aggregate utility is maximized, subject to both bandwidth and storage capacity constraints in the network.  

\subsection{Network Model}
\noindent\textbf{Caches and Items.} Following \cite{ioannidis2016addaptive,mahdian2019kelly}, we represent a network by a directed graph $G(\mathcal{V}, \mathcal{E}).$ We assume $G(\mathcal{V}, \mathcal{E})$ is \textit{symmetric}, i.e., $(b,a) \in \mathcal{E}$ implies that $(a,b) \in \mathcal{E}$. 
There exists a catalog $\itemcat$ of items (e.g., files, or file chunks) of equal size which network users can request. Each node is associated with a cache that can store a finite number of items. We describe cache contents via indicator variables:
$x_{vi} \in \{0, 1\} \quad \text{for}~v \in \mathcal{V},~i \in \itemcat,$
 where $x_{vi}=1$ indicates that node $v$ stores item $i\in \itemcat$. The total number of items that a node 
$v \in \mathcal{V}$ can store is bounded by its \emph{node capacity} $c_v\in \mathbb{N}$ (measured in number of items). More precisely,
\begin{equation} \label{equ:cachecv} 
    \textstyle \sum_{i\in \itemcat} x_{vi} \leq c_v \quad \text{for all}~v \in \mathcal{V}.
\end{equation} 
We associate each item $i$ with a fixed set of \textit{designated servers} $S_i \subseteq \mathcal{V}$, that permanently store $i$; equivalently, $x_{vi} = 1$, for all $v \in S_i$. As we discuss below, these act as ``caches of last resort'', and ensure that all items can be eventually retrieved.

\noindent\textbf{Content Requests.} \label{subsec:request}
Item requests are routed over the network toward the designated servers. We denote by $\requestset$ the set of all requests. A request is determined by the item requested and  the path that the request follows. Formally, a request $\requestindex$ is a pair $(i, p)$ where $i \in \itemcat$ is the requested item and $p\subseteq \mathcal{V}$ is the path to be traversed to serve this request.   Path $p$ of length $|p| = K$ is a sequence of nodes $\{p_1, p_2, \dots, p_K \}\subseteq \mathcal{V}$,
where  $(p_j, p_{j+1}) \in \mathcal{E}$ for all $j \in [K-1] \triangleq \{1,2,\dots, K-1\}$.  An incoming request $(i, p)$ is routed over the graph $G$ and follows the path $p$, until it reaches a node that stores item $i$. At that point, a \emph{response message} is generated, which carries the requested item. The response is propagated over $p$ in the reverse direction, i.e., from the node that stores the item, back to the first node in $p$, from which the request was generated. Following \cite{ioannidis2016addaptive,mahdian2019kelly}, we say that a request $\requestindex=(i,p)$ is \textit{well-routed} if:
    	(a) the path $p$ is simple, i.e., it contains no loops, (b) the last node in the path is a designated server  for $i$, i.e., $p_K \in S_i$, and
        (c) no other node in the path is a designated server node for $i$, i.e.,  $p_k \notin S_i$, for $k \in [K-1]$.
Without loss of generality, we assume that all requests in $\requestset$ are well-routed;  
 note that every well-routed request eventually encounters a node  that contains the item requested.

\label{sect:congestioncontrol}

\noindent\textbf{Bandwidth Capacities.} For each link $(a,b) \in \mathcal{E}$ there exists a positive and finite \emph{link capacity} $C_{ab} > 0$ (measured in items/sec) indicating the bandwidth available on  $(a,b)$. We denote the vector of link capacities by $\boldsymbol{C}\in {\mathbb{R}_+}^{|\mathcal{E}|}$. We consider  two means of controlling the rate of item transmission on a link, thereby preventing congestion in the network: (a) via the \emph{cache allocation} strategy, i.e., by storing the requested item on a node along the path, which eliminates the flow of item on upstream links, 
 and (b) via the \emph{rate allocation} strategy, i.e., by controlling the  rate with which requests enter the network. We describe each one in detail  below.
 
 
 

\noindent \textbf{Cache Allocation Strategy.} We adopt a probabilistic cache allocation strategy. That is, we partition time into periods of equal length $T > 0$. At the beginning of the $t$-th time period, each node $v \in \mathcal{V}$ stores an item $i \in \itemcat$ independently of other nodes and other time periods with probability $y_{vi} \in [0,1]$, i.e., $y_{vi} = \mathrm{P} \{ x_{vi}(t) = 1 \} = \mathbb{E} [ x_{vi}(t) ]$, for all $t>0$, where $x_{vi}(t) = 1$ indicates that node $v$ stores item $i$ at the $t$-th  time period. We denote by $\boldsymbol{Y} = [y_{vi}]_{v \in \mathcal{V}, i \in \itemcat} \in [0,1]^{|\mathcal{V}| |\itemcat|}$ the \emph{cache allocation strategy} vector, satisfying constraints: 
\begin{equation}\label{equ:avgcv} 
   \textstyle \sum_{i\in\itemcat} y_{vi}  \leq c_v \quad \text{for all}~v \in \mathcal{V}.
\end{equation}
Although condition \eqref{equ:avgcv} implies that cache capacity constraints are satisfied in \emph{expectation}, it is necessary and sufficient for the existence of a probabilistic content placement  (i.e., a mapping of items to caches) that satisfies capacity constraints \eqref{equ:cachecv} \emph{exactly} (see, e.g., \cite{geocaching2015, ioannidis2016addaptive}).  We present this  probabilistic placement in detail in Appendix~\ref{append:probcachealg}.
\fullversion{}{In this mapping, (a) node $v$ stores at most $c_v$ items, and (b) the marginal probability of storing $i$ is $y_{vi}$. Given a global cache allocation strategy $\boldsymbol{Y}$ that satisfies \eqref{equ:avgcv}, this mapping can be used to randomly place items at caches  at the beginning of each time period. We therefore treat $\boldsymbol{Y}$ as the parameter to optimize over in the remainder of the paper. }

\noindent \textbf{Rate Allocation Strategy.} 
Our second knob for controlling congestion is classic rate allocation, as in \cite{kelly1998rate, giovani2013optimal}. That is, we control the input rate of requests so that the final requests injected into the network have a rate equal or smaller than original rates. We refer to the original exogenous arrival rate of a requests $\requestindex=(i, p) \in \requestset$ as the \emph{demand rate}, and denote it by $\bar{\lambda}_{\requestindex} > 0$ (in requests per second). We denote the vector of demand rates by $\bar{\boldsymbol{\lambda}} = [\bar{\lambda}_{\requestindex}]_ {\requestindex\in \requestset}$. We also denote the admitted input rate of requests into the network by $\lambda_{\requestindex}$, where
\begin{align}\label{equ:bound} \small
    \lambda_{\requestindex} \leq \bar{\lambda}_{\requestindex}, \quad \text{for all}~ \requestindex \in \requestset.
\end{align}
We refer to the vector $\boldsymbol{\lambda} = [\lambda_{\requestindex}]_{\requestindex \in \requestset}\in \mathbb{R}_+^{|\requestset|}$ as the \emph{rate allocation strategy}. We make the following assumptions on  requests admitted into the network: 
(a) the request process is stationary and ergodic,
    (b) a corresponding response message is eventually created for every admitted request, 
    (c)  the network is stable if, for all $(b,a) \in \mathcal{E}$, the following holds:
    \begin{equation}\label{equ:cab} \small
   \rho_{(b,a)}(\boldsymbol{\lambda},\boldsymbol{Y}) =\textstyle\sum_{(i,p): (a,b) \in p}\lambda_{(i,p)} \prod_{v=p_1}^{a} (1-y_{vi}) \leq C_{ba}.
    \end{equation}
    
Using the probabilistic cache allocation scheme, and the fact that the admitted request process is stationary and ergodic, $\rho_{(b,a)}(\boldsymbol{\lambda},\boldsymbol{Y})$ is the expected rate of requests passing through link $(a,b)$. In particular, $\prod_{v=p_1}^{a} (1-y_{vi})$ is the fraction of admitted rate $\lambda_{(i,p)}$ which is forwarded on link $(a,b)$. Since we have assumed that for each request a response message is generated, and comes back on the reverse path, the condition in \eqref{equ:cab} ensures that the rate of items transmitted on link $(b,a)$ is less than or equal to the link capacity $C_{ba}$. If the traffic rate on a link is greater than the link capacity, the network becomes unstable. In order for \eqref{equ:cab} to ensure stability, similar to \cite{vip2014yeh, deco2019kamran}, in effect we assume that the size of requests are negligible compared to the size of requested items, and the load primarily consists of the downstream traffic of items. Note that the load on edge $(b,a)$ depends both the rate \emph{and} the cache allocation strategy, while constraints \eqref{equ:cab} are non-convex.

\noindent\textbf{System Utility.} Consistent with Kelly et al.~\cite{kelly1998rate},  each request class $\requestindex \in \requestset$ is associated with a \emph{utility function} $U_{\requestindex} : \mathbb{R}_+ \rightarrow \mathbb{R}$ of the admitted rate $\lambda_n$. The network utility is then the social welfare, i.e., the sum of all request utilities in the network:
\begin{equation}\label{equ:utility} \small
    U(\boldsymbol{\lambda}) = \textstyle\sum_{\requestindex \in \requestset} U_{\requestindex} (\lambda_{\requestindex}).
\end{equation} 
We assume that each function $U_{\requestindex}$ is twice continuously differentiable, non-decreasing, and concave for all $\requestindex \in \requestset$. Our goal is to determine a rate allocation strategy $\boldsymbol{\lambda} = [\lambda_{\requestindex}]_{\requestindex \in \requestset}\in \mathbb{R}_+^{|\requestset|}$ and a cache allocation strategy $\boldsymbol{Y} = [y_{vi}]_{v \in \mathcal{V}, i \in \itemcat} \in [0,1]^{|\mathcal{V}| |\itemcat|}$ that jointly maximize \eqref{equ:utility}, subject to the constraints \eqref{equ:avgcv}, \eqref{equ:bound}, and \eqref{equ:cab}. For technical reasons, we first transform this problem into an equivalent problem via a change of variables.

\subsection{Problem Formulation}\label{sec:probform}
\noindent\textbf{Change of variables.} Let the \emph{residual rate} per request be $r_{\requestindex} \triangleq \bar{\lambda}_{\requestindex} - \lambda_{\requestindex},$ for $\requestindex \in \requestset.$
Given the \emph{rate residual strategy}  $\boldsymbol{R} \triangleq [r_{\requestindex}]_{\requestindex \in \requestset}\in \mathbb{R}_+^{|\requestset|}$, we rewrite the  utility as
\begin{equation} \label{equ:newutility} \small
    F(\boldsymbol{R}) \triangleq U(\bar{\boldsymbol{\lambda}}-\boldsymbol{R}) = \textstyle\sum_{\requestindex \in \requestset} U_{\requestindex} (\bar{\lambda}_{\requestindex} - r_{\requestindex}).
\end{equation}
Under this change of variables, we state our problem as:

\vspace{3mm}
{\hspace*{\stretch{1}} \textsc{UtilityMax} \hspace*{\stretch{1}} }
   \begin{subequations}\label{cachingproblem1}
    \begin{align}
    \textrm{maximize} \quad &F(\boldsymbol{R}) \displaybreak[0]\\
    \textrm{subject to} \quad & (\boldsymbol{Y},\boldsymbol{R}) \in \mathcal{D}_1,
\end{align}
\end{subequations}
where $\mathcal{D}_1$ is the set of points $(\boldsymbol{Y},\boldsymbol{R}) \in \mathbb{R}^{|\mathcal{V}|  |\itemcat|} \times \mathbb{R}^{|\requestset|}$ satisfying the following constraints\footnote{
W.l.o.g., we implicitly set 
$y_{vi} = 1 ,$ for all $v \in \mathcal{S}_i, i \in \itemcat$ and do include these constraints in \eqref{equ:D}.}:
\begin{subequations}\label{equ:D}
\begin{align} \small
    &g_{ba}(\boldsymbol{Y}, \boldsymbol{R}) \geq \textstyle\sum_{(i,p) : (a,b) \in p} \bar{\lambda}_{(i,p)} - C_{ba}, \quad \forall (b,a) \in \mathcal{E} \label{const:link}\displaybreak[0]\\
    &g_{v}(\boldsymbol{Y}) \leq c_v, \quad \quad \quad \quad \quad \quad \quad \quad \quad \quad \quad \quad \quad \forall v \in \mathcal{V} \label{const:cache}\displaybreak[0]\\
    &0 \leq y_{vi} \leq 1, \quad \quad \quad \quad \quad \quad \quad \quad \quad \quad \forall v \in \mathcal{V}, i \in \itemcat \label{const:boxfory}\displaybreak[0]\\
    &0 \leq r_{\requestindex} \leq \bar{\lambda}_{\requestindex}, \quad \quad \quad \quad \quad \quad \quad \quad \quad \quad \quad \quad \forall \requestindex \in \requestset, \label{const:boxforlambda}
\end{align}
\end{subequations}
 where, for $(b,a)\in\mathcal{E}$ and $v\in\mathcal{V}$, $g_{v}(\boldsymbol{Y}) \triangleq \sum_{i \in \itemcat} y_{vi}$, and 
 \begin{align}g_{ba}(\boldsymbol{Y},\boldsymbol{R}) \triangleq \!\!\!\!\!\!\!\sum_{(i,p) : (a,b) \in p}\!\!\! \!\!\!\!\!\bar{\lambda}_{(i,p)}\! -\! (\bar{\lambda}_{(i,p)} \!-\! r_{(i,p)})
 \!\!   \prod_{v=p_1}^{a}\! (1\!-\!y_{vi}).\!\!\!\!\label{eq:gba}\end{align}

An important consequence of this change of variables is the following lemma.
\begin{lem} \label{lem:monotonedr}
For all $(b,a) \in \mathcal{E}$, functions $g_{ba}:\mathbb{R}^{|\mathcal{V}|  |\itemcat|} \times \mathbb{R}^{|\requestset|}  \to \mathbb{R} $, are monotone DR-submodular.
\end{lem}

\begin{proof} Please see Appendix~\ref{append:proofoflemmonotonedr}.
\end{proof}

\fullversion{}{We use this in Section \ref{sec:main} to provide algorithms with optimality guarantees for \textsc{UtilityMax}. For this reason, we briefly review  DR-submodular functions below.  

\subsection{DR-submodular Functions} \label{subsec:drsub} Bian et al. \cite{bian2017guaranteed} define a DR-submodular function as follows:
\begin{defn} \label{def:drproperty}
 Suppose $\mathcal{X}$ is a subset of $\mathbb{R}^n$. A function $f : \mathcal{X} \rightarrow \mathbb{R}$ is DR-submodular if for all $ \boldsymbol{a} \leq \boldsymbol{b} \in \mathcal{X}$, $ i \in [n]$, and $k \in \mathbb{R}_+$, such that $(k  e_i + \boldsymbol{a})$ and $(k  e_i + \boldsymbol{b})$ are still in $\mathcal{X}$, the following inequality holds:
\begin{equation*}
    f(k e_i + \boldsymbol{a}) - f(\boldsymbol{a}) \geq f(k e_i + \boldsymbol{b}) - f(\boldsymbol{b})
\end{equation*}
\end{defn}
Intuitively, a DR-submodular function $f$ is concave coordinate-wise along any non-negative or non-positive direction.
DR-submodular functions arise in a variety of different settings (see Bian et al.~\cite{bian2017guaranteed}), and in some sense satisfy a weakened notion of concavity. They can also be defined in alternative ways that parallel the zero-th, first, and second order conditions for concavity (see \cite{boyd2004convex}). For example, for $\mathcal{X} \subseteq \mathbb{R}^n$,
a function $f: \mathcal{X} \rightarrow \mathbb{R}$ is DR-submodular \emph{iff} for all $\boldsymbol{x}, \boldsymbol{y} \in \mathcal{X}$,
\begin{equation*}
    f(\boldsymbol{x}) + f(\boldsymbol{y}) \geq f(\boldsymbol{x} \vee \boldsymbol{y}) + f( \boldsymbol{x} \wedge \boldsymbol{y}),
\end{equation*}
where $\vee$ and $\wedge$ are coordinate-wise maximum and
minimum operations, respectively. 
A list of such conditions of DR-submodular functions is summarized in Table \ref{table:drprop}. Each one of these properties serves as a necessary and sufficient condition for a function to be DR-submodular \cite{bian2017guaranteed}. 

\begin{table}[!t]  
\centering
\caption{Properties of DR-submodular functions}
 \begin{tabular}{|c | c |} 
 \hline
 Properties & DR-submodular $f(.)$, $\forall \boldsymbol{x}, \boldsymbol{y} \in \mathcal{X}$ \\ [1ex] 
 \hline\hline 
  0'th order & $f(\boldsymbol{x}) + f(\boldsymbol{y}) \geq f(\boldsymbol{x} \vee \boldsymbol{y}) + f( \boldsymbol{x} \wedge \boldsymbol{y})$, \\
 & and $f(.)$ is coordinate-wise concave. \\
 \hline
  1'st order & Definition \ref{def:drproperty}  \\
 \hline
 2'nd order &  $\frac{\partial^2 f(\boldsymbol{x})}{\partial x_i \partial x_j} \leq 0,~ \forall i,j \in [n]$ \\ 
 \hline
\end{tabular}
\label{table:drprop}
\end{table}

The following lemma, proved in Appendix~\ref{append:proofoflemmonotonedr}, indicates how DR-submodularity arises in the constraints of \textsc{UtilityMax}:
\begin{lem} \label{lem:monotonedr}
for all $(b,a) \in \mathcal{E}$, functions $g_{ba}:\mathbb{R}^{|\mathcal{V}|  |\itemcat|} \times \mathbb{R}^{|\requestset|}  \to \mathbb{R} $, given by \eqref{eq:gba}, are monotone DR-submodular.
\end{lem}
For more information on DR-submodular functions, we refer the interested reader to \cite{bian2017guaranteed,bian2017dr}.
 }
\section{Cache and Rate Allocation}\label{sec:main}
The constraint set $\mathcal{D}_1$ in Problem~\eqref{cachingproblem1} is not convex. Therefore, there is in general no efficient way to find the global optimum.
\fullversion{Here, we propose two algorithms that come with (different) optimality guarantees. Both algorithms exploit the fact that the functions $g_{ba}(\cdot)$ are DR-submodular functions.}{Constrained optimization techniques can be used to find a  \emph{Karush-Kuhn-Tucker (KKT) point} (i.e., a point in which KKT necessary conditions for optimality hold) under mild conditions. In general, there is no guarantee on the value of the objective at a KKT point compared to the global optimum. However, as one of our major contributions, we  provide optimality guarantees for the objective value of \eqref{cachingproblem1} at a KKT point. In particular, we propose two algorithms that come with (different) optimality guarantees. Both algorithms exploit the fact that the functions $g_{ba}(\cdot)$ in \eqref{const:link}--which are the cause of non-convexity--are DR-submodular functions.} 

In our first approach, described in Section~\ref{subsec:LagrangianBarrierForUtilityMax}, we solve Problem~\eqref{cachingproblem1} using a \emph{Lagrangian barrier algorithm} \cite{conn1997globally}. We show that this converges to
a Karush-Kuhn-Tucker (KKT) point (i.e., a point at which KKT necessary conditions for optimality hold) under mild assumptions. Crucially, and in contrast to general non-convex problems~\cite{conn1997globally}, we provide guarantees on the objective value at such KKT points. In particular, we show that the ratio of the objective value at a KKT point to the global optimum value approaches 1, asymptotically, under an appropriate proportional scaling of capacities and demand. 
In Section~\ref{subsec:concaverelax}, we provide an alternative solution via \emph{convex relaxation} of the constraint set $\mathcal{D}_1$. This turns our problem into a convex optimization problem for which efficient algorithms exist. We show that the solution obtained by solving the convex problem is feasible, and its objective value is  bounded from below by the optimal value of another instance of Problem~\eqref{cachingproblem1} with tighter constraints. 



\subsection{Lagrangian Barrier Algorithm for \textsc{UtilityMax}} \label{subsec:LagrangianBarrierForUtilityMax}
Problem~\eqref{cachingproblem1} is a maximization problem subject to the inequality constraints \eqref{const:link}, \eqref{const:cache}, and the simple box constraints on the variables \eqref{const:boxfory} and \eqref{const:boxforlambda}. Due to this structure, we propose to use the \emph{Lagrangian Barrier with Simple Bounds} (LBSB) Algorithm, introduced by Conn et al. \cite{conn1997globally}. 

\noindent \textbf{Algorithm description.}  LBSB defines the Lagrangian barrier function $\Psi(\boldsymbol{Y}, \boldsymbol{R}, \boldsymbol{\mu},\boldsymbol{\gamma}, \boldsymbol{s}) $, given by:
\begin{align} \small
    &F(\boldsymbol{R}) + \sum_{(b,a) \in \mathcal{E}} \mu_{ba} s_{ba} \log (g_{ba}(\boldsymbol{Y},\boldsymbol{R}) - \sum_{(i,p) : (a,b) \in p} \bar{\lambda}_{(i,p)} \nonumber \\
     &+ C_{ba} + s_{ba} ) + \sum_{v \in \mathcal{V}} \gamma_{v} s_{v} \log (c_v - g_{v}(\boldsymbol{Y}) + s_{v}), \label{equ:psi}
\end{align}
where the elements of vectors $\boldsymbol{\mu} \triangleq [\mu_{ba}]_{(b,a)\in \mathcal{E}} \in \mathbb{R}_+^{|\mathcal{E}|}$ and $\boldsymbol{\gamma} \triangleq [\gamma_v]_{v \in \mathcal{V}} \in \mathbb{R}_+^{|\mathcal{V}|}$ are the positive \emph{Lagrange multiplier estimates} corresponding to \eqref{const:link} and \eqref{const:cache} respectively, and the vector $\boldsymbol{s} \in \mathbb{R}_+^{|\mathcal{E}| + |\mathcal{V}| }$ consists of the positive values $[s_{ba}] _{(b,a) \in \mathcal{E}}$  and $[s_v]_{v \in \mathcal{V}} $ called \emph{shifts} \cite{conn1997globally}. Intuitively, the Lagrangian barrier function in \eqref{equ:psi} penalizes the infeasibility of the link and cache constraints, and the shifts allow the constraints to be violated to some extent. Consider the following problem:
\begin{align} 
    \max_{(\boldsymbol{Y}, \boldsymbol{R})} \quad &\Psi(\boldsymbol{Y}, \boldsymbol{R}, \boldsymbol{\mu}_k,\boldsymbol{\gamma}_k, \boldsymbol{s}_k) \label{prob:innersub} \\
    \text{s.t.} \quad &(\boldsymbol{Y}, \boldsymbol{R}) \in \mathcal{B}, \nonumber
\end{align}
where the values $\boldsymbol{\mu}_k,\boldsymbol{\gamma}_k, \boldsymbol{s}_k$ are given, and $\mathcal{B}$ is the box constraints set defined by \eqref{const:boxfory} and \eqref{const:boxforlambda}. Then the necessary optimality condition for Problem~\eqref{prob:innersub} is
\begin{equation} \label{equ:neccforinner}
    \| P\left( (\boldsymbol{Y},\boldsymbol{R})~,~ \nabla_{\boldsymbol{Y},\boldsymbol{R}} \Psi(\boldsymbol{Y}, \boldsymbol{R}, \boldsymbol{\mu}_k,\boldsymbol{\gamma}_k, \boldsymbol{s}_k) \right) \|  = 0,
\end{equation}
where $P(\boldsymbol{a}~,~\boldsymbol{b}) \triangleq \boldsymbol{a} - \Pi_{\mathcal{B}}(\boldsymbol{a} + \boldsymbol{b})$, and $\Pi_{\mathcal{B}}(\boldsymbol{a})$ is the projection of the vector $\boldsymbol{a}$ on  the set $\mathcal{B}$. At the $k$-th iteration, LBSB updates $\boldsymbol{Y}_k, \boldsymbol{R}_k$ by finding a point in $\mathcal{B}$, such that the following condition is satisfied:
\begin{equation}\label{equ:appforinner}
    \| P\left( (\boldsymbol{Y},\boldsymbol{R})~,~ \nabla_{\boldsymbol{Y},\boldsymbol{R}} \Psi(\boldsymbol{Y}, \boldsymbol{R}, \boldsymbol{\mu}_k,\boldsymbol{\gamma}_k, \boldsymbol{s}_k) \right) \| \leq \omega_k,
\end{equation}
where parameter $\omega_k\geq 0$ indicates the accuracy of the solution; when $\omega_k = 0$, the point $(\boldsymbol{Y}_k, \boldsymbol{R}_k)$ satisfies the necessary optimality condition \eqref{equ:neccforinner}. In general, this point can be found by iterative algorithms such as interior-point methods or projected gradient ascent. Here, we use the trust-region algorithm \cite{trustRegion} for simple box constraints\fullversion{}{, which we describe for completeness in Appendix~\ref{append:trustregion}}.

After updating $(\boldsymbol{Y}_k, \boldsymbol{R}_k)$, LBSB checks whether the solution is in a ``locally convergent regime'' (with tolerance $\delta_k$). If so, it updates  the Lagrange multiplier estimates. It also updates the accuracy parameter $\omega_{k+1}$, the tolerance parameter $\delta_{k+1}$ and the shifts $\boldsymbol{s}_{k+1}$; these updates differ depending on whether the algorithm is in a locally convergent regime or not. These iterations  continue until the algorithm converges; a high-level summary of LBSB is described in Alg.~\ref{alg:findKKT}. We refer the interested reader to Conn et al. \cite{conn1997globally} or Appendix~\ref{append:LBSB} for a detailed description of the algorithm.\fullversion{}{ The details include initial parameters, updates of the Lagrange multiplier estimates, shifts, accuracy and tolerance parameters, as well as a formal definition of the locally convergent regime.} Under relatively mild assumptions (see Lemma~\ref{lem:kkt}), the solution generated by LBSB converges to a KKT point and the Lagrange multiplier estimates converge to the Lagrange multipliers corresponding to that KKT point.

\begin{algorithm}[!t] 
    \caption{Summary of Lagrangian Barrier with Simple Bounds (LBSB)}
    \label{alg:findKKT}
    \begin{algorithmic}[1] 
        \State Set accuracy parameter $\omega_0$
        \State Set tolerance parameter for locally convergent regime $\delta_0$
        \State Set Lagrange multiplier estimates $\boldsymbol{\mu}_0$, $\boldsymbol{\gamma}_0$, and other initial parameters
        \State $k \gets -1$
        \Repeat
            \State $k \gets k + 1$
            \State Compute shifts $\boldsymbol{s}_k$
            \State Find $(\boldsymbol{Y}_k,\boldsymbol{R}_k) \in \mathcal{B}$ such that \label{line:inner}
            
            \scalebox{0.8}{$ || P( (\boldsymbol{Y}_k,\boldsymbol{R}_k)~,~ \nabla_{\boldsymbol{Y},\boldsymbol{R}} \Psi(\boldsymbol{Y}_k, \boldsymbol{R}_k, \boldsymbol{\mu}_k,\boldsymbol{\gamma}_k, \boldsymbol{s}_k) ) || \leq \omega_k$}.
            
            \If{ in locally convergent regime (with threshold $\delta_k$)} \label{line:localyconv}
                \State Update Lagrange multiplier estimates $\boldsymbol{\mu}_{k+1},\boldsymbol{\gamma}_{k+1}$
                \State Update $\omega_{k+1}$ using $\omega_k$
                \State Update $\delta_{k+1}$ using $\delta_k$ 
            \Else
                \State Update $\omega_{k+1}$ using initial parameters
                \State Update $\delta_{k+1}$ using initial parameters
            \EndIf
        \Until{convergence}
    \end{algorithmic}
\end{algorithm}

\noindent \textbf{Guarantees.} For general non-convex problems, the KKT point to which LBSB converges \emph{comes with no optimality guarantees}. Our main contribution is showing that due to DR-submodularity, applying LBSB to Problem~\eqref{cachingproblem1} yields a stronger result. We first need a few additional  assumptions.

\begin{defn}\label{def:logDR}
A  function $U_\requestindex: \mathbb{R}_+ \rightarrow \mathbb{R}$ has \textit{logarithmic diminishing return} if there exists a finite number $\theta_\requestindex \in \mathbb{R}_+$ such that $\lambda \frac{d U_\requestindex(\lambda)}{d\lambda} \leq \theta_\requestindex$ for all $\lambda \in[0, \infty).$
  
\end{defn}



\begin{assm} \label{assm:logreturn}
All utility functions $U_{\requestindex}$, $\requestindex \in \requestset$, have logarithmic diminishing return. 
\end{assm}

\begin{assm} \label{assm:unbounded}
At least one of the utility functions is unbounded from above.
\end{assm}

We want to stress that Assumptions~\ref{assm:logreturn} and~\ref{assm:unbounded} are relatively mild. For example, consider the well-known $\alpha-$fair utility functions \cite{warland2000fairwindow}\fullversion{.}{:
\begin{equation*}
    U^{\alpha}(x) =
\left\{
	\begin{array}{ll}
		\omega \cdot \frac{x^{1-\alpha}}{1-\alpha} & \mbox{if } \alpha >0, ~ \alpha \neq 1 \\
	    \omega \cdot log(x) & \mbox{if } \alpha = 1,
	\end{array}
\right.\
\end{equation*}
where $\omega\geq 0.$}
All $\alpha-$fair utility functions with $\alpha \geq 1$ have logarithmic diminishing return. For $\alpha = 1$, the utility function is unbounded from above. Therefore, for example, a problem instance with  $\alpha-$fair utility functions, where $\alpha\geq 1$ and at least one function has $\alpha = 1$, satisfies Asssumptions~\ref{assm:logreturn} and~\ref{assm:unbounded}.

\begin{defn}\emph{Regular point}:
If the gradients of the active inequality constraints at $(\boldsymbol{Y}, \boldsymbol{R})$
are linearly independent, then $(\boldsymbol{Y}, \boldsymbol{R})$ is called a regular point. \label{def:regular}
\end{defn}

Our main result is the following theorem,  characterizing the quality of \emph{regular} limit points of the sequence $\{ (\boldsymbol{Y}_{k}, \boldsymbol{R}_{k})\}$ generated by Alg.~\ref{alg:findKKT}. We note that the regularity of limit points is typically  considered in the analysis of other methods in constrained optimization literature as well \cite{bertsekas1982const, fletcher1981practical, bertsekas1999nonlinear}.

\begin{thm}\label{thm:assymptotic}
Consider a problem instance with link capacity vector $\boldsymbol{C} \in \mathbb{R}_+^{|\mathcal{E}|}$ and demand rate vector $\bar{\boldsymbol{\lambda}} \in \mathbb{R}_+^{|\requestset|}$.   Suppose Assumptions~\ref{assm:logreturn} and~\ref{assm:unbounded} hold, and $\{ (\boldsymbol{Y}_k, \boldsymbol{R}_k) \}$, $k \in \mathcal{K}$ is a sub-sequence generated by Alg.~\ref{alg:findKKT} which converges to a regular point $\big[\hat{\boldsymbol{Y}}(\boldsymbol{C},\bar{\boldsymbol{\lambda}}), \hat{\boldsymbol{R}}(\boldsymbol{C},\bar{\boldsymbol{\lambda}})\big]$.  Denote the optimal solution by $\big[\boldsymbol{Y}^*(\boldsymbol{C},\bar{\boldsymbol{\lambda}}), \boldsymbol{R}^*(\boldsymbol{C},\bar{\boldsymbol{\lambda}})\big]$. Then, we have $\lim_{m\to\infty} {F( \hat{\boldsymbol{R}}(m\boldsymbol{C},m\bar{\boldsymbol{\lambda}}) )}/{F( \boldsymbol{R}^* (m\boldsymbol{C},m\bar{\boldsymbol{\lambda}}) )} = 1$.
\end{thm}
Hence, the value of the objective at a regular limit point of Alg.~1 approaches the optimal objective value,  when  link capacities and demand rates grow to infinity by the same factor $m$. Note that increasing the link capacities \emph{does not} make the problem easier, since   demand rates increase proportionally.

The proof of Theorem~\ref{thm:assymptotic} follows from a sequence of lemmas, which we now outline.

\begin{lem} \label{lem:kkt}
Let $\{ (\boldsymbol{Y}_k, \boldsymbol{R}_k) \}$, $k \in \mathcal{K}$, be any subsequence generated by Alg.~\ref{alg:findKKT} which converges to a regular point $(\hat{\boldsymbol{Y}}, \hat{\boldsymbol{R}})$. Then $(\hat{\boldsymbol{Y}}, \hat{\boldsymbol{R}})$ is a KKT point for Problem~\eqref{cachingproblem1}. 
\end{lem}

\begin{proof}[Proof] Please see Appendix~\ref{append:proofOflemKKT}.  The lemma is proved by showing that the regularity assumption is equivalent to the assumption stated in Theorem 4.4 of Conn et. al. \cite{conn1997globally}.\end{proof}

\fullversion{The next key technical lemma characterizes the difference between the value of the objective at a KKT point and the global optimal value:

\begin{lem} \label{lem:diff} \small
Let $(\hat{\boldsymbol{Y}},\hat{\boldsymbol{R}})$ be a KKT point and $(\boldsymbol{Y}^*,\boldsymbol{R}^*)$ be the optimal point for Problem~\eqref{cachingproblem1}. Then $F(\hat{\boldsymbol{R}}) \geq F(\boldsymbol{R}^*) - \sum_{(b,a)\in \mathcal{E}}\hat{\mu}_{ba} \left(\sum_{(i,p) : (a,b) \in p} \bar{\lambda}_{(i,p)} - C_{ba}\right)$, where $\hat{\mu}_{ba}$ is the Lagrange multiplier corresponding to link $(b,a)$ in constraint \eqref{const:link}, for all $(b,a) \in \mathcal{E}$.
\end{lem}}{Our major contribution is to characterize the difference between the value of the objective at a KKT point and the global optimal value as shown in Lemma~\ref{lem:diff}. The key factor in proving Lemma~\ref{lem:diff} is the concavity of $F(\cdot)$ and the fact that $g_{ba}(\cdot)$ are monotone DR-submodular functions for all $(b,a) \in \mathcal{E}$.

\begin{lem} \label{lem:diff} \small
Let $(\hat{\boldsymbol{Y}},\hat{\boldsymbol{R}})$ be a KKT point and $(\boldsymbol{Y}^*,\boldsymbol{R}^*)$ be the optimal point for Problem~\eqref{cachingproblem1}. Then
\begin{equation}
F(\hat{\boldsymbol{R}}) \geq F(\boldsymbol{R}^*) - \sum_{(b,a)\in \mathcal{E}}\hat{\mu}_{ba} \left(\sum_{(i,p) : (a,b) \in p} \bar{\lambda}_{(i,p)} - C_{ba}\right).\label{eq:diff}
\end{equation}
\end{lem}}

\fullversion{
\begin{proof}[Proof] Please see Appendix~\ref{append:ProofOfLemDiff}.  Key elements in proving Lemma~\ref{lem:diff} are the concavity of $F(\cdot)$ and the fact that $g_{ba}(\cdot)$ are monotone DR-submodular functions for all $(b,a) \in \mathcal{E}$.\end{proof}}{Lemma~\ref{lem:diff} is proved in Appendix~\ref{append:ProofOfLemDiff}, and implies that the value of the objective in the KKT point is bounded away from the optimal value by an additive term $\sum_{(b,a) \in \mathcal{E}}\hat{\mu}_{ba} (\sum_{(i,p) : (a,b) \in p} \bar{\lambda}_{(i,p)} - C_{ba})$.}
\begin{lem} \label{lem:boundOnMu}
Under Assumption~\ref{assm:logreturn},
$$    \sum_{(b,a) \in \mathcal{E}}\!\!\hat{\mu}_{ba} (\!\!\!\!\!\!\!\!\!\sum_{(i,p) : (a,b) \in p}\!\!\!\!\!\!\!\! \bar{\lambda}_{(i,p)} - C_{ba})
    \leq \theta \!\!\!\!\sum_{(b,a) \in \mathcal{E}}\!\! \frac{n_{ab}}{C_{ba}}~~ (\!\!\!\!\!\!\!\underset{(i,p) : (a,b) \in p}{\sum}\!\!\!\!\!\!\!\! \bar{\lambda}_{(i,p)} - C_{ba}),$$
where $n_{ab}$ is the number of paths passing through $(a,b)$, and $\theta \triangleq \max_{\requestindex \in \requestset} \theta_\requestindex$ is the maximum logarithmic diminishing return parameter among utilities. 
\end{lem}

\begin{proof}[Proof] Please see Appendix~\ref{append:ProofOfLemBoundOnMu}.\end{proof} 

\begin{proof}[Proof of Theorem 1]
By Lemma~\ref{lem:kkt}, we know that $\big[\hat{\boldsymbol{Y}}(m\boldsymbol{C},m\bar{\boldsymbol{\lambda}}), \hat{\boldsymbol{R}}(m\boldsymbol{C},m\bar{\boldsymbol{\lambda}})\big]$ is a KKT point for all $m \in \mathbb{R}_+$. Thus, Lemma~\ref{lem:diff} and Lemma~\ref{lem:boundOnMu} imply that, for all $m \in \mathbb{R}_+,$ $F( \hat{\boldsymbol{R}}(m\boldsymbol{C},m\bar{\boldsymbol{\lambda}}))$ is bounded from below, i.e., 
\begin{align}
    &F( \hat{\boldsymbol{R}}(m\boldsymbol{C},m\bar{\boldsymbol{\lambda}}) ) \geq \nonumber \\
    &F( \boldsymbol{R}^* (m\boldsymbol{C},m\bar{\boldsymbol{\lambda}}) ) -  \theta\!\! \sum_{(b,a) \in \mathcal{E}} \!\frac{n_{ab}}{C_{ba}}~~ (\!\!\!\!\!\sum_{(i,p) : (a,b) \in p}\!\!\!\!\!\! \bar{\lambda}_{(i,p)}\! - \!C_{ba}).\!\! \!\!\! \label{equ:greaterwithm}
\end{align} 

According to \eqref{equ:newutility}, $F( \boldsymbol{R}^* (m\boldsymbol{C},m\bar{\boldsymbol{\lambda}}) ) = U( \boldsymbol{\lambda}^* (m\boldsymbol{C},m\bar{\boldsymbol{\lambda}}) )$. The rate vector $m\boldsymbol{\lambda}^* (\boldsymbol{C},\bar{\boldsymbol{\lambda}})$ is feasible in Problem~\eqref{cachingproblem1} with link capacity vector $m\boldsymbol{C}$ and demand rate vector $m\bar{\boldsymbol{\lambda}}$, and we have $U( \boldsymbol{\lambda}^* (m\boldsymbol{C},m\bar{\boldsymbol{\lambda}}) ) \geq U( m\boldsymbol{\lambda}^* (\boldsymbol{C},\bar{\boldsymbol{\lambda}}) )$.
Combining this and the fact that there exists an unbounded utility function $U_n(.)$ which grows without bound as the input rate goes to infinity (Assumption~\ref{assm:unbounded}), we have $\lim_{m \to \infty} U( \boldsymbol{\lambda}^* (m\boldsymbol{C},m\bar{\boldsymbol{\lambda}}) ) = \infty$, or equivalently $\lim_{m \to \infty} F( \boldsymbol{R}^* (m\boldsymbol{C},m\bar{\boldsymbol{\lambda}}) ) = \infty$.
This implies that there exists a $m_0 > 0$ such that $F( \boldsymbol{R}^* (m\boldsymbol{C},m\bar{\boldsymbol{\lambda}}) ) > 0$, for all $m \geq m_0.$
We conclude the proof by dividing both sides of \eqref{equ:greaterwithm} by $F( \boldsymbol{R}^* (m\boldsymbol{C},m\bar{\boldsymbol{\lambda}}) )$ for $m \geq m_0$, and letting $m \to \infty$.
\end{proof}

\noindent \textbf{Convergence Rate.} 
\fullversion{We briefly discuss the convergence rate studied by Conn et. al. \cite{conn1997globally} as applied to \textsc{UtilityMax}; to do so, we need the following additional assumption:}{Conn et. al. \cite{conn1997globally} have also studied the convergence rate of Alg.~\ref{alg:findKKT} under additional assumptions. We can also use this to characterize convergence rate of the algorithm as applied to \textsc{UtilityMax}.}
\begin{assm} \label{assm:lipsh}
 The function $F(\boldsymbol{R})$, its gradient, and elements of its Hessian are Lipschitz continuous.
\end{assm} 
\begin{prop} \label{prop:k0}
Suppose Assumption~\ref{assm:lipsh} holds, and iterates $\{ (\boldsymbol{Y}_k, \boldsymbol{R}_k) \}$ generated by Alg.~\ref{alg:findKKT} have a single limit point $(\hat{\boldsymbol{Y}}, \hat{\boldsymbol{R}})$ which is regular, and satisfies the second-order sufficiency condition \fullversion{(see Section 4.3 of Bertsekas \cite{bertsekas1999nonlinear})}{ (discussed in Appendix~\ref{append:constrainedopt})}. 
Then with proper choice of parameters, $\{ (\boldsymbol{Y}_k, \boldsymbol{R}_k) \}$ converges to $(\hat{\boldsymbol{Y}}, \hat{\boldsymbol{R}})$ at least \emph{R-linearly} for sufficiently large $k$, i.e., there exists $r \in (0,1)$, $P >0$, and $k_0$ such that $\| (\boldsymbol{Y}_k, \boldsymbol{R}_k) -  (\hat{\boldsymbol{Y}}, \hat{\boldsymbol{R}}) \| \leq P r^k$, for all $k \geq k_0$.
\end{prop}
\fullversion{We omit the proof for brevity; it follows by verifying that the assumptions in Proposition~\ref{prop:k0} imply all the assumptions for part (ii) of Theorem 5.3 and Corollary 5.7 of Conn et. al. \cite{conn1997globally}}{We prove this in Appendix~\ref{append:ProofOfPropk0}, by showing that the assumptions in Proposition~\ref{prop:k0} imply all the assumptions for part (ii) of Theorem 5.3 and Corollary 5.7 of Conn et. al. \cite{conn1997globally}}. 

\fullversion{}{We can combine the results of Thm.~\ref{thm:assymptotic} and Proposition~\ref{prop:k0} to characterize the quality of solutions obtained at the $k$-th iteration of Alg.~\ref{alg:findKKT}. 
\begin{corollary} \label{cor:combined}
If the assumptions of Thm.~\ref{thm:assymptotic} and Prop.~\ref{prop:k0} hold, there exist $r \in (0,1)$, $Q >0$, $\beta > 0$ and $k_0$ such that  $\frac{F(\boldsymbol{R}^{(k)})}{ F( \boldsymbol{R}^*) }$ $\geq 1 - \frac{\theta}{F( \boldsymbol{R}^*)}\sum_{(b,a) \in \mathcal{E}} n_{ab} \frac{\left(\sum_{(i,p) : (a,b) \in p} \bar{\lambda}_{(i,p)} - C_{ba}\right)}{C_{ba}} - \frac{Q r^k}{F( \boldsymbol{R}^*)}$,  for all $k \geq k_0$.
\end{corollary}
\begin{proof}
By Assumption~\ref{assm:lipsh}, $F(\cdot)$ is Lipschitz continuous. Hence, there exists an $L\in \mathbb{R}_+$ s.t.
\begin{align*}
    | F(\boldsymbol{R}_k) -  F(\hat{\boldsymbol{R}}) |&\overset{}{\leq} L \| (\boldsymbol{Y}_k, \boldsymbol{R}_k) -  (\hat{\boldsymbol{Y}}, \hat{\boldsymbol{R}}) \| \overset{\text{Prop.~ \ref{prop:k0}}}{\leq} L P r^k,
 \end{align*}
 where $L$ is the Lipschitz constant. 
 Let $Q\triangleq LP,$ we then have the following for all $k \geq k_0$
  \begin{align*}    
     F ( &\boldsymbol{R}_k ) \geq F( \hat{\boldsymbol{R}} ) - Q r^k \overset{\text{Lem.~\ref{lem:diff}},\text{Lem.~\ref{lem:boundOnMu}}}{\geq} F( \boldsymbol{R}^* ) \nonumber \\
     & - \theta \sum_{(b,a) \in \mathcal{E}} n_{ab} \frac{\left(\sum_{(i,p) : (a,b) \in p} \bar{\lambda}_{(i,p)} - C_{ba}\right)}{C_{ba}} - Q r^k. 
\end{align*}
Dividing both sides by $F( \boldsymbol{R}^*)$ concludes the proof. 
\end{proof}

Corollary~\ref{cor:combined} states that if the assumptions in Thm.~\ref{thm:assymptotic} and Proposition~\ref{prop:k0} hold and $k$ is large enough, the ratio of the objective value for $(\boldsymbol{Y}_k, \boldsymbol{R}_k)$ to the optimum is bounded away from 1 by two additive terms, i.e., a sum and $\frac{Q r^k}{F( \boldsymbol{R}^* )}$; the latter can be made arbitrary small (by increasing $k$), while the former also goes to zero when link capacities and demand rates are increased with the same factor (see Thm. \ref{thm:assymptotic}).
}

\subsection{Convex Relaxation of \textsc{UtilityMax}} \label{subsec:concaverelax}
An alternative approach for solving Problem~\eqref{cachingproblem1} is to come up with a convex relaxation of constraint set $\mathcal{D}_1$. This turns our problem into a convex optimization problem which can be solved efficiently. Similar to prior literature \fullversion{\cite{ioannidis2016addaptive, karimi2017stochastic}}{\cite{seeman2013adaptive,ioannidis2016addaptive, karimi2017stochastic}}, we construct concave upper and lower bounds for the non-convex and non-concave functions in constraints \eqref{const:link}, using the so-called Goemans and Williamson  inequality:
\begin{lem}[Goemans and Williamson \cite{goemans1994new}] \label{lem:concavebiconjugate}
For $\boldsymbol{Z} \in [0,1]^n$ define $A(\boldsymbol{Z}) \triangleq 1- \prod_{i=1}^n (1-z_i)$ and $B(\boldsymbol{Z}) \triangleq \min \{1, \sum_{i=1}^{n} z_i \}$. Then, $(1-1/e) B(\boldsymbol{Z}) \leq A(\boldsymbol{Z}) \leq B(\boldsymbol{Z})$.
\end{lem}
\fullversion{Using Lemma~\ref{lem:concavebiconjugate}, we obtain 
\begin{align} \label{eq:sandwich}
  &(1-1/e)\tilde{g}_{ba}(\boldsymbol{Y}, \boldsymbol{R}) 
  \leq  g_{ba}(\boldsymbol{Y}, \boldsymbol{R}) 
  \leq  \tilde{g}_{ba}(\boldsymbol{Y}, \boldsymbol{R}),
\end{align}
where $$\tilde{g}_{ba}(\boldsymbol{Y}, \boldsymbol{R}) \triangleq \sum_{(i,p)\in \requestset: (a,b) \in p} \bar{\lambda}_{(i,p)} \min\bigg\{ 1, \frac{r_{(i,p)}}{\bar{\lambda}_{(i,p)}} + \sum_{k=1}^{a} y_{p_k i}\bigg\}$$ are concave functions for all $(b,a) \in \mathcal{E}$.}{Applying Lemma~\ref{lem:concavebiconjugate} to the functions $g_{ba}(\cdot)$ yields the following corollary.
\begin{corollary}
\label{col:sandwich}
Functions $g_{ba}(\boldsymbol{Y}, \boldsymbol{R})$, $(b,a) \in \mathcal{E}$ , satisfy:
\begin{align*}
  &(1-1/e)\tilde{g}_{ba}(\boldsymbol{Y}, \boldsymbol{R}) 
  \leq  g_{ba}(\boldsymbol{Y}, \boldsymbol{R}) 
  \leq  \tilde{g}_{ba}(\boldsymbol{Y}, \boldsymbol{R})
\end{align*}
where
\begin{align}\label{eq:gtilde}
\tilde{g}_{ba}(\boldsymbol{Y}, \boldsymbol{R}) \triangleq \!\!\! \!\!\!\sum_{(i,p)\in \requestset: (a,b) \in p} \!\!\!\!\!\! \bar{\lambda}_{(i,p)} \min\{ 1, \frac{r_{(i,p)}}{\bar{\lambda}_{(i,p)}} + \sum_{k=1}^{a} y_{p_k i}\},
\end{align}
are concave functions, for all $(b,a)\in \mathcal{E}$.
\end{corollary} 
The proof of Corollary~\ref{col:sandwich} is presented in Appendix~\ref{append:proofcolsandwich}.} We use this to formulate the following convex problem:

\vspace{2mm}
{\hspace*{\stretch{1}} \textsc{ConvexUtilityMax} \hspace*{\stretch{1}} }
   \begin{subequations}\label{cachingproblem2}
    \begin{align}
    \textrm{maximize} \quad &F(\boldsymbol{R}) \\
    \textrm{subject to} \quad & (\boldsymbol{Y},\boldsymbol{R}) \in \mathcal{D}_2,
\end{align}
\end{subequations}
where $\mathcal{D}_2\subseteq \mathbb{R}^{|\mathcal{V}| \times |\itemcat|} \times \mathbb{R}^{|\requestset|}$ is the set of  $(\boldsymbol{Y},\boldsymbol{R}) $ satisfying:
\begin{subequations}
\begin{align*}
    &
    \tilde{g}_{ba}(\boldsymbol{Y},\boldsymbol{R}))\geq \frac{\sum_{(i,p) : (b,a) \in p} \bar{\lambda}_{(i,p)} - C_{ba}}{1-1/e},  &&\forall (b,a) \in \mathcal{E} \\
    & g_v(\boldsymbol{Y}) \leq c_v, & &\forall v \in \mathcal{V}\\
    & 0 \leq y_{vi} \leq 1,& &\forall v \in \mathcal{V}, \forall i \in \itemcat \\
    & 0 \leq r_{\requestindex} \leq \bar{\lambda}_{\requestindex}, &&\forall \requestindex \in \requestset. 
\end{align*} 
\end{subequations}
Although $\tilde{g}_{ba}(\cdot)$, for all $(b,a) \in \mathcal{E}$, are non-differentiable, the optimal solution can be found using sub-gradient methods as $\mathcal{D}_2$ is convex. The following theorem provides a bound on the optimal value of Problem~\eqref{cachingproblem2} with respect to Problem \eqref{cachingproblem1}.

\begin{thm} \label{thm:approximationBound}
Let  $(\boldsymbol{Y}^{**}_{\boldsymbol{C}}, \boldsymbol{R}^{**}_{\boldsymbol{C}})$ be the optimal solution of Problem~\eqref{cachingproblem2} with link capacity vector $\boldsymbol{C}$. Also let  $(\boldsymbol{Y}^{*}_{\boldsymbol{C}}, \boldsymbol{R}^{*}_{\boldsymbol{C}})$ and  $(\boldsymbol{Y}^{*}_{\boldsymbol{C}'}, \boldsymbol{R}^{*}_{\boldsymbol{C}'})$ be the optimal solutions of two instances of Problem~\eqref{cachingproblem1} with  link capacity vectors $\boldsymbol{C}$ and $\boldsymbol{C}'$, respectively, where for all $(b,a) \in \mathcal{E}$
\begin{equation}\small
    C_{ba}' = C_{ba} - \frac{1}{e-1} [\textstyle\sum_{(i,p) : (a,b) \in p} \bar{\lambda}_{(i,p)} - C_{ba}].\label{eq:cba_prime}
\end{equation}
Then, $(\boldsymbol{Y}^{**}_{\boldsymbol{C}}, \boldsymbol{R}^{**}_{\boldsymbol{C}})$ is a feasible solution to Problem~\eqref{cachingproblem1} with link capacity vector $\boldsymbol{C}$, and
$F(\boldsymbol{R}^{*}_{\boldsymbol{C}'}) \leq  F(\boldsymbol{R}^{**}_{\boldsymbol{C}})  \leq F(\boldsymbol{R}^{*}_{\boldsymbol{C}})$.
\end{thm}
\begin{proof}
Let $\mathcal{D}_3\subseteq  \mathbb{R}^{|\mathcal{V}||\itemcat|} \times \mathbb{R}^{|\requestset|}$ be the set of   $(\boldsymbol{Y},\boldsymbol{R}) $ satisfying:
\begin{subequations}\label{equ:Dprime}
\begin{align*}
    &g_{ba}(\boldsymbol{Y}, \boldsymbol{R}) \geq \frac{\sum_{(i,p) : (a,b) \in p} \bar{\lambda}_{(i,p)} - C_{ba}}{1-1/e} & \forall (b,a) \in \mathcal{E}\\
    &g_v(\boldsymbol{Y}) \leq c_v  & \forall v \in \mathcal{V} \\
    &0 \leq y_{vi} \leq 1 & \forall v \in \mathcal{V}, \forall i \in \itemcat\\
    &0 \leq r_{\requestindex} \leq \bar{\lambda}_{\requestindex} & \forall \requestindex \in \requestset 
\end{align*}
\end{subequations}
Observe that $\mathcal{D}_3$ is the constraint set for Problem~\eqref{cachingproblem1} with the link capacity vector $\boldsymbol{C}'$. By \fullversion{\eqref{eq:sandwich}}{Corollary~\ref{col:sandwich}}, we have $\mathcal{D}_3 \subseteq \mathcal{D}_2 \subseteq \mathcal{D}_1$. By definition, $F(\boldsymbol{R}^{*}_{\boldsymbol{C}}), F(\boldsymbol{R}^{**}_{\boldsymbol{C}}),$ and $F(\boldsymbol{R}^{*}_{\boldsymbol{C}'})$ are maximum values of $F(\boldsymbol{R})$ subject to $\mathcal{D}_1, \mathcal{D}_2,$ and $\mathcal{D}_3,$ respectively. As a result, we have $F(\boldsymbol{R}^{*}_{\boldsymbol{C}'}) \leq  F(\boldsymbol{R}^{**}_{\boldsymbol{C}})  \leq F(\boldsymbol{R}^{*}_{\boldsymbol{C}}).\qedhere$
\end{proof}

\fullversion{Thm.~\ref{thm:approximationBound} states that the solution to the convex problem is no worse than the optimum of an instance of the original problem, with link capacities $C_{ba}' = C_{ba} - \frac{1}{e-1} [\sum_{(i,p) : (a,b) \in p} \bar{\lambda}_{(i,p)} - C_{ba}]$. Note that $C_{ba}'$ can be negative. In that case, Problem \eqref{cachingproblem1} with negative link capacities has no feasible solutions, and the solution of Problem~\eqref{cachingproblem2} has no lower bound.}{Thm.~\ref{thm:approximationBound} implies that instead of Problem~\eqref{cachingproblem1}, we can solve Problem~\eqref{cachingproblem2}, which is a convex program with tighter constraint set. Since Problem~\eqref{cachingproblem2} has more restrictive constraints, its solution is naturally a feasible solution to \eqref{cachingproblem1} and is upper bounded by the optimum. On the other hand, Thm.~\ref{thm:approximationBound} states that this solution is no worse than the optimum of an instance of the original problem, with link capacities $C_{ba}' = C_{ba} - \frac{1}{e-1} [\sum_{(i,p) : (a,b) \in p} \bar{\lambda}_{(i,p)} - C_{ba}]$. Note that $C_{ba}'$ can be negative. In that case, Problem \eqref{cachingproblem1} with negative link capacities has no feasible solutions, and the solution of Problem~\eqref{cachingproblem2} has no lower bound.} 
\fullversion{}{The following corollary is provided in terms of rate of requests $\boldsymbol{\lambda}$ for further clarification. 
\begin{corollary}
If $\sum_{(i,p) : (a,b) \in p} \bar{\lambda}_{(i,p)} \leq \Delta C_{ba}$ for all $(b,a) \in \mathcal{E}$ and some $\Delta \in [1,e]$, then
\begin{equation*}
 U\left(\frac{e- \Delta}{e-1}\boldsymbol{\lambda}^{*}\right) \leq
 U(\boldsymbol{\lambda}^{**})\leq
    U(\boldsymbol{\lambda}^{*}),
\end{equation*}
where $\boldsymbol{\lambda}^{*} = \bar{\boldsymbol{\lambda}} -\boldsymbol{R}^* $ are the optimal rates for Problem~\eqref{cachingproblem1}, and $\boldsymbol{\lambda}^{**}=\bar{\boldsymbol{\lambda}} -\boldsymbol{R}^{**} $ are the optimal rates for Problem~\eqref{cachingproblem2}.
\end{corollary}
\begin{proof}
By \eqref{equ:newutility} 
 and Thm.~\ref{thm:approximationBound} we have that
\begin{equation*}
 U(\boldsymbol{\lambda}^{*}_{\boldsymbol{C}'}) \leq U(\boldsymbol{\lambda}^{**}_{\boldsymbol{C}}) \leq   U(\boldsymbol{\lambda}^{*}_{\boldsymbol{C}})  ,
\end{equation*}
where $\boldsymbol{\lambda}^{*}_{\boldsymbol{C}} = \bar{\boldsymbol{\lambda}} -\boldsymbol{R}^*_{\boldsymbol{C}}$ and  $\boldsymbol{\lambda}^{*}_{\boldsymbol{C}'} = \bar{\boldsymbol{\lambda}} -\boldsymbol{R}^*_{\boldsymbol{C}'}$ are the optimal rates of two instances of Problem~\eqref{cachingproblem1} with link capacity vectors $\boldsymbol{C}$ and $\boldsymbol{C}'$, respectively, and $\boldsymbol{\lambda}^{**}_{\boldsymbol{C}} =  \bar{\boldsymbol{\lambda}} -\boldsymbol{R}^{**}_{\boldsymbol{C}}$ is the optimal rate for Problem~\eqref{cachingproblem2} with link capacity vector $\boldsymbol{C}$. Since $\sum_{(i,p) : (a,b) \in p} \bar{\lambda}_{(i,p)} \leq \Delta C_{ba}$ for all $(b,a) \in \mathcal{E}$, by \eqref{eq:cba_prime} we have $ \frac{e-\Delta}{e-1} \boldsymbol{C} \leq \boldsymbol{C}' $. Therefore, the rate vector $\frac{e-\Delta}{e-1} \boldsymbol{\lambda}^{*}_{\boldsymbol{C}}$ is feasible in Problem \eqref{cachingproblem1} with link capacity vector $\boldsymbol{C}'$. Hence, $ U\left(\frac{e-\Delta}{e-1} \boldsymbol{\lambda}^{*}_{\boldsymbol{C}}\right) \leq U(\boldsymbol{\lambda}^{*}_{\boldsymbol{C}'}) $ and we can write
\begin{align*}
U\left(\frac{e-\Delta}{e-1} \boldsymbol{\lambda}^{*}_{\boldsymbol{C}}\right) \leq 
U(\boldsymbol{\lambda}^{*}_{\boldsymbol{C}'}) \leq
U(\boldsymbol{\lambda}^{**}_{\boldsymbol{C}})  \leq 
     U(\boldsymbol{\lambda}^{*}_{\boldsymbol{C}}).
\end{align*}
Therefore, we have $U\left(\frac{e-\Delta}{e-1} \boldsymbol{\lambda}^{*}_{\boldsymbol{C}}\right) \leq 
U(\boldsymbol{\lambda}^{**}_{\boldsymbol{C}})  \leq 
U(\boldsymbol{\lambda}^{*}_{\boldsymbol{C}})$, and this concludes the proof. 
\end{proof}
}

Comparing this to LBSB, the convex relaxation is a simpler problem, as it requires solving a convex program. On the other hand,  LBSB provides better optimality guarantees, especially when the demand rate of requests exceeds  the capacity of the links.  
In practice, as we see in the numerical evaluations (Section~\ref{sec:simul}), the Lagrangian barrier outperforms the convex relaxation method for a wide range of network topologies and  parameter settings.   

\section{Numerical Evaluation} \label{sec:simul}
\begin{table}[!t]
\caption{Graph Topologies and Experiment Parameters.}\label{networks}
\begin{scriptsize}
\begin{tabular}{lp{1em}p{1em}p{1em}p{1em}p{1em}p{1em}p{1em}p{1em}p{1em}}

Graph & $|\mathcal{V}|$ & $|\mathcal{E}|$ & $|\itemcat|$ & $|\requestset|$ &  $|Q|$ & $c_v'$ &$\hat{F}_{\text{(loose)}}$ &$\hat{F}_{\text{(tight)}}$ & \#Vars   \\

\hline
\hline
\texttt{cycle}  & 30 & 60 & 10 & 100 & 10  & 2   & 9.53 &   9.53  & 344  \\
 \texttt{lollipop} & 30 & 240 & 10 & 100 & 10 & 2 &  9.53 & 9.53 & 274\\
  \texttt{geant} &22 & 66& 10&100 & 10 & 2 & 9.53& 9.53 & 228\\
 \texttt{abilene} & 9 & 26 & 10 &40 & 4 & 2  & 3.81& 3.81 & 85\\
 \texttt{dtelekom} & 68 & 546 & 15 & 125 & 15 & 3 & 11.91& 11.91 & 301\\
\texttt{balanced-tree} & 63 & 124  & 30 & 450 & 15 & 3& 42.88& 28.45 & 1434\\
\texttt{grid-2d} &64 & 224  & 30 & 450 & 15 & 3&42.88 & 37.08 & 1665\\
\texttt{hypercube}&64 & 384  & 15 & 450 & 15 & 3& 42.88&35.59 &1189\\
 \texttt{small-world} &64 & 308  & 30 & 450 & 15 & 3 & 42.88 & 37.90&1349 \\
  \texttt{erdos-renyi} &64 & 378  & 30 & 450 & 15 & 3 & 42.88& 35.06 & 1191 \\
\end{tabular}
\end{scriptsize}
\end{table}

\fullversion
{
\begin{figure}[!t]
\vspace{-3mm}
	\centering
	\subfloat[$\kappa=0.95$]{
		\includegraphics[width=1\columnwidth]{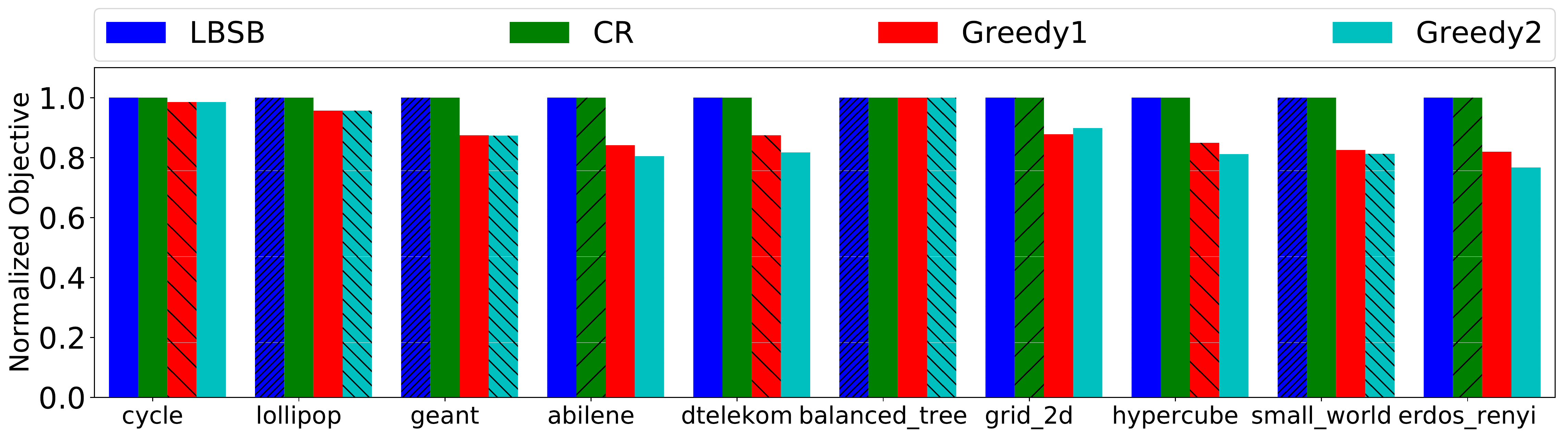}\label{bar95}
	}
	\vspace{-3mm}
	
	\subfloat[$\kappa=0.85$]{
		\includegraphics[width=1\columnwidth]{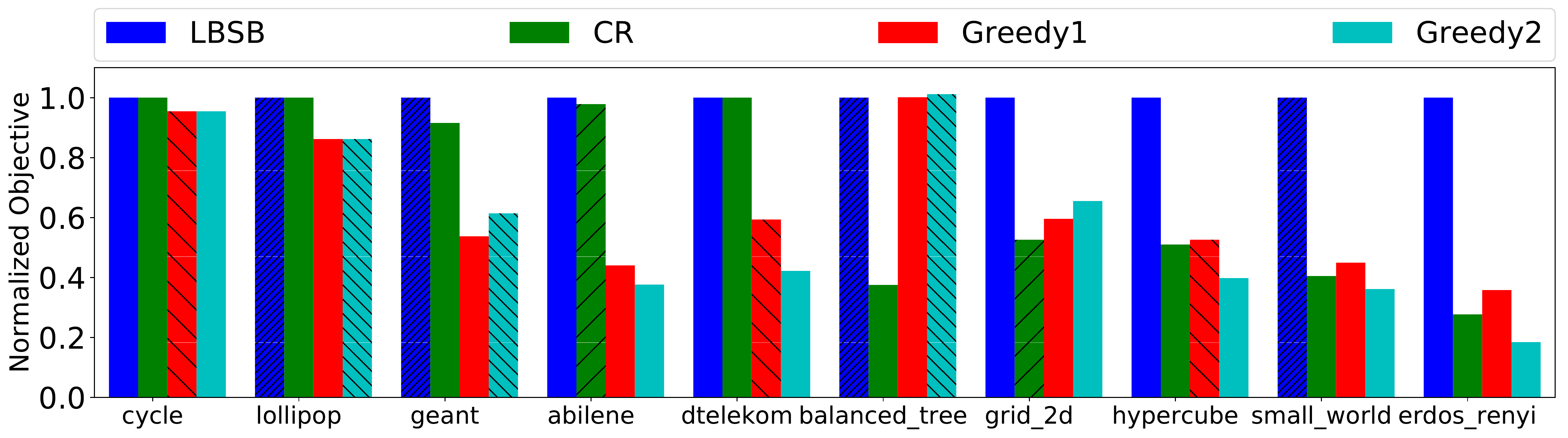}\label{bar8}
	}
\caption{Objectives performance. The figure shows the normalized objective obtained by different algorithms across different topologies and in two settings, i.e., the loose setting, with $\kappa=0.95$ (Fig~\ref{bar95}), and the tight setting with $\kappa=0.85$ (Fig~\ref{bar8}). Note that solutions for all 4 algorithms are feasible in all cases.}\label{utility}
\vspace{-5mm}
\end{figure}
}
{
\begin{figure*}[!t]
\vspace{-3mm}
	\centering
	\subfloat[$\kappa=0.95$]{
		\includegraphics[width=1\textwidth]{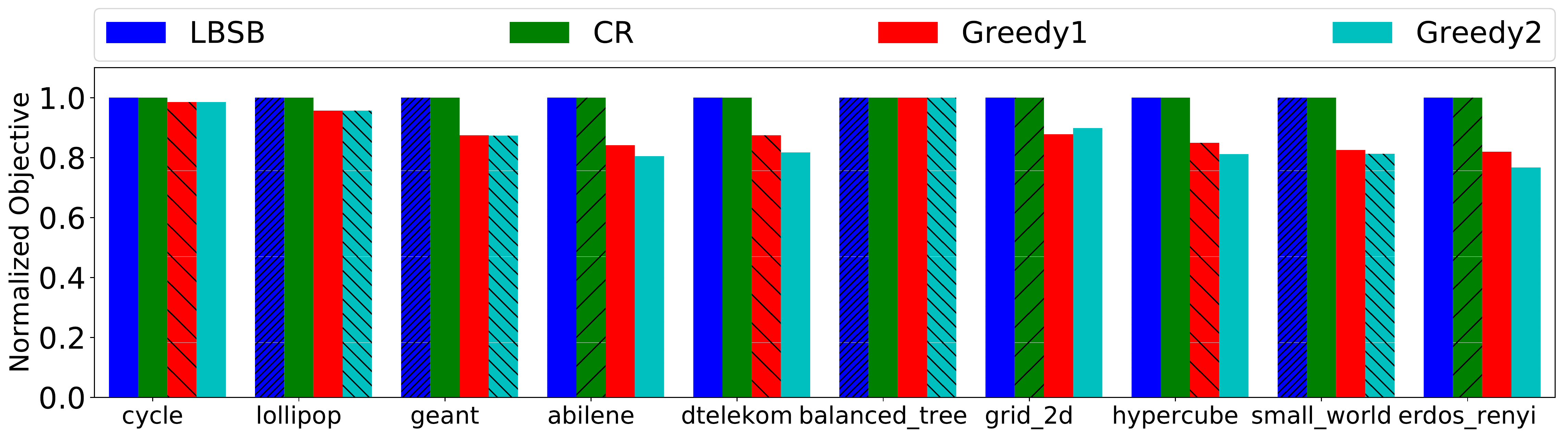}\label{bar95}
	}

	\subfloat[$\kappa=0.85$]{
		\includegraphics[width=1\textwidth]{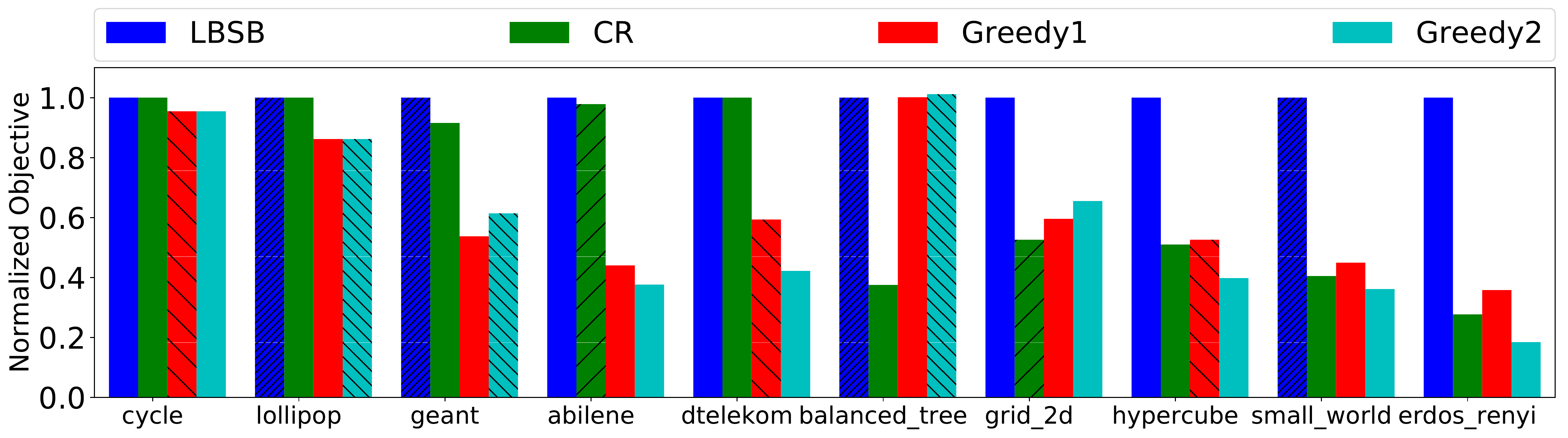}\label{bar8}
	}
\caption{Objectives performance. The figure shows the normalized objective obtained by different algorithms across different topologies and in two settings, i.e., the loose setting, with $\kappa=0.95$ (Fig~\ref{bar95}), and the tight setting with $\kappa=0.85$ (Fig~\ref{bar8}). We see that \texttt{LBSB} outperforms other algorithms. Also, \texttt{CR} performance almost matches \texttt{LBSB} in the loose setting (Fig.~\ref{bar95}); however, its performance significantly deteriorates in the tight setting (Fig.~\ref{bar8}). Note that solutions for all 4 algorithms are feasible in all cases.}\label{utility}
\end{figure*}
}

\fullversion
{
\begin{figure}[!t]
\vspace{-1mm}
	\centering
	\subfloat[$\kappa=0.95$]{
		\includegraphics[width=0.48\columnwidth]{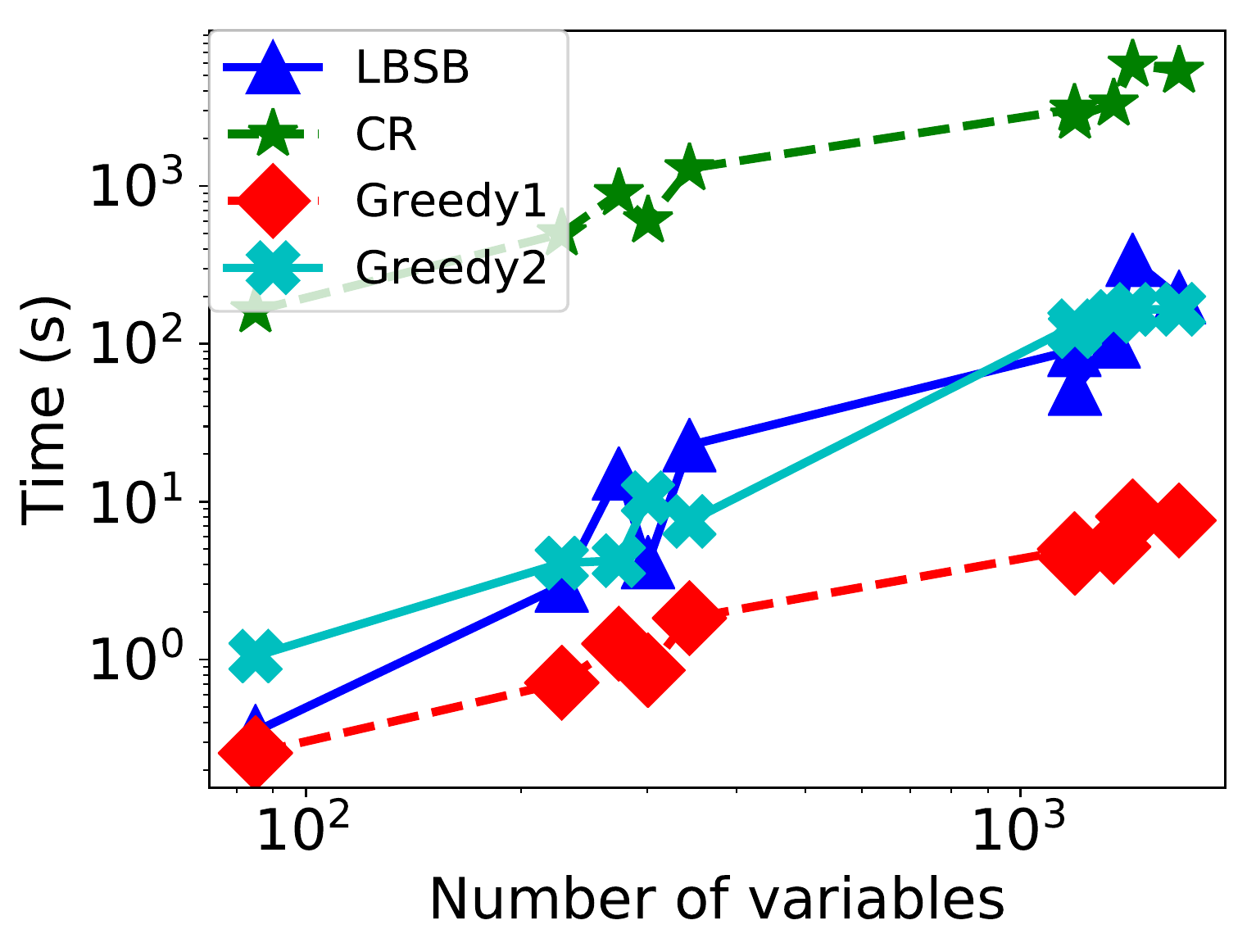}\label{time95}
	}
		\subfloat[$\kappa=0.85$]{
		\includegraphics[width=0.48\columnwidth,scale=0.5]{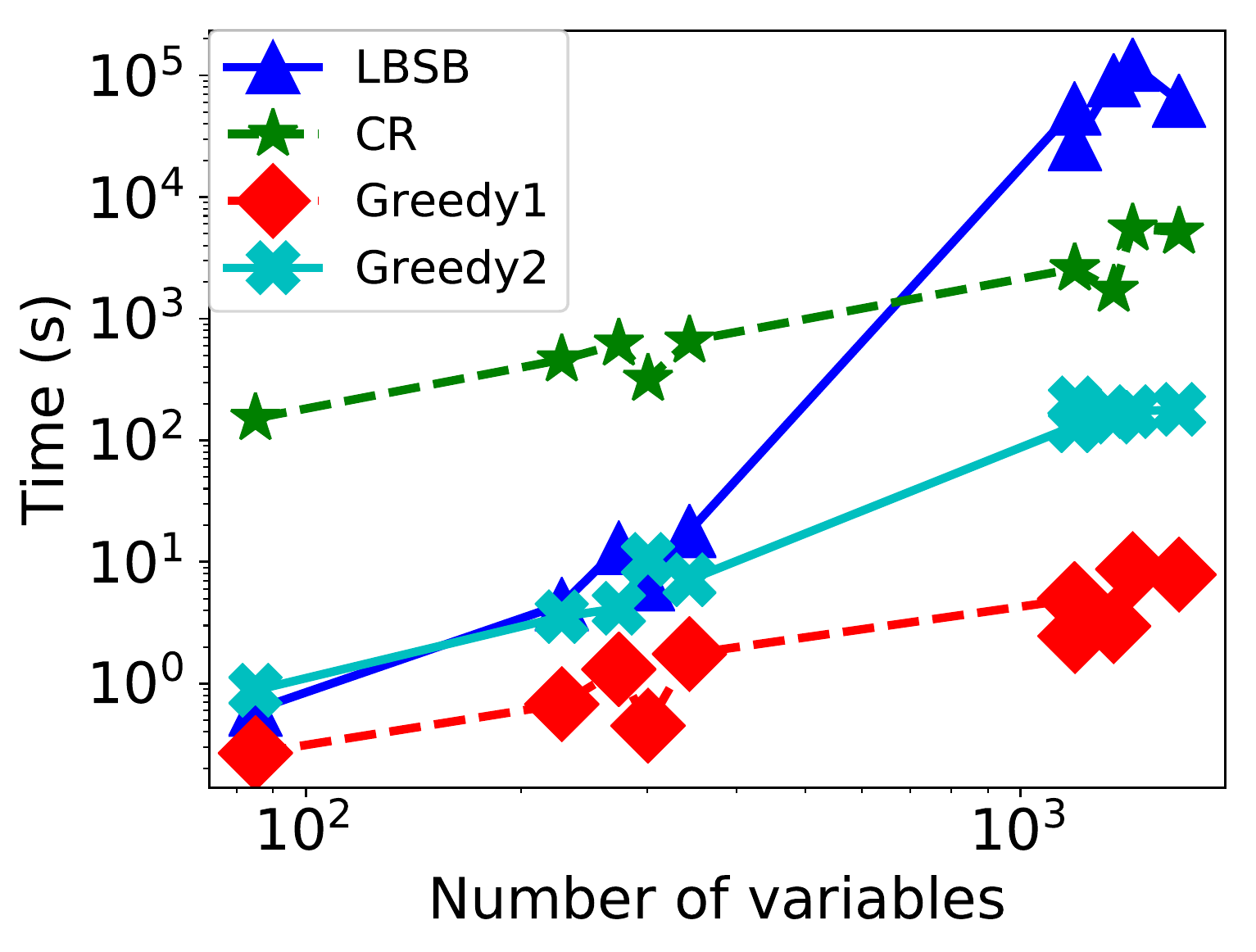}\label{time8}
	}
	\caption{Execution Times.  Figures~\ref{time95} and \ref{time8} show the execution times for algorithms w.r.t. the number of variables, for the loose and tight settings, respectively.}\label{fig:time}
	\vspace{-5mm}	
\end{figure}
}
{
\begin{figure*}[!t]
\vspace{-1mm}
	\centering
	\subfloat[$\kappa=0.95$]{
		\includegraphics[height=0.23\textwidth]{plots/95_times.pdf}\label{time95}
	}
	\subfloat[$\kappa=0.85$]{
		\includegraphics[height=0.23\textwidth]{plots/85_times.pdf}\label{time8}
	}
	\subfloat[\texttt{LBSB} trajectory for \texttt{grid-2d}]{
		\includegraphics[height=0.23\textwidth]{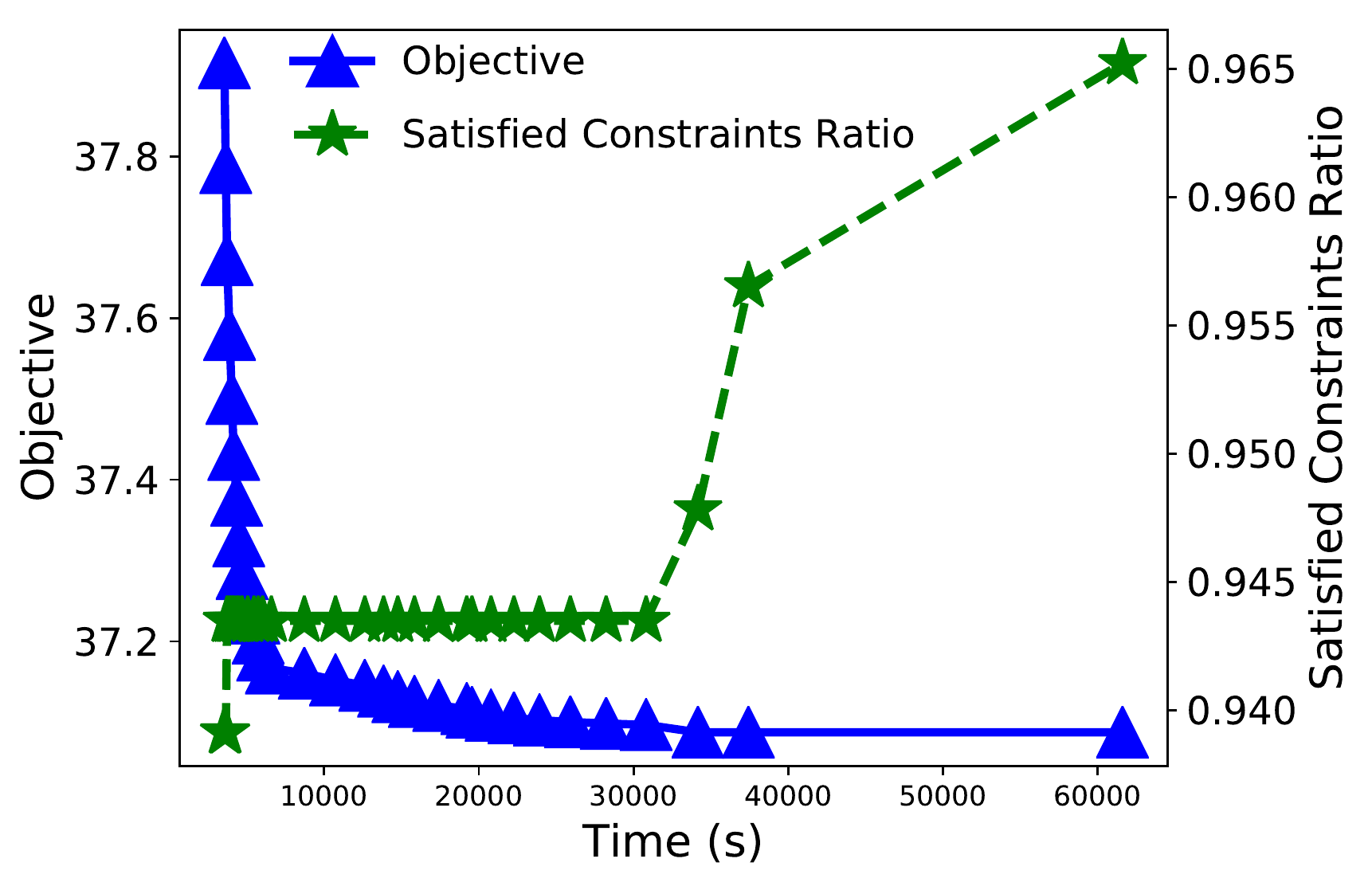}\label{fig:trajgrid}
	}
\caption{Execution Times and Convergence.  Figures~\ref{time95} and \ref{time8} show the execution times for algorithms w.r.t. the number of variables, for the loose and tight settings, respectively. We observe that in the loose setting the execution times for \texttt{LBSB} are comparable with greedy algorithms; however, in the tight setting \texttt{LBSB} is much slower, as the the trust-region algorithm that we use at each iteration of \texttt{LBSB} requires more iterations to satisfy \eqref{equ:appforinner}. Fig.~\ref{fig:trajgrid} shows the objective and feasibility trajectories for the iterations of \texttt{LBSB}; as iterations progress, feasibility improves and thusly the objective value decreases.}\label{fig:time}
\end{figure*}}

\fullversion{
\begin{figure}[!t]
\vspace{-1mm}
	\centering
	\subfloat[\texttt{abilene}]{
		\includegraphics[width=0.48\columnwidth]{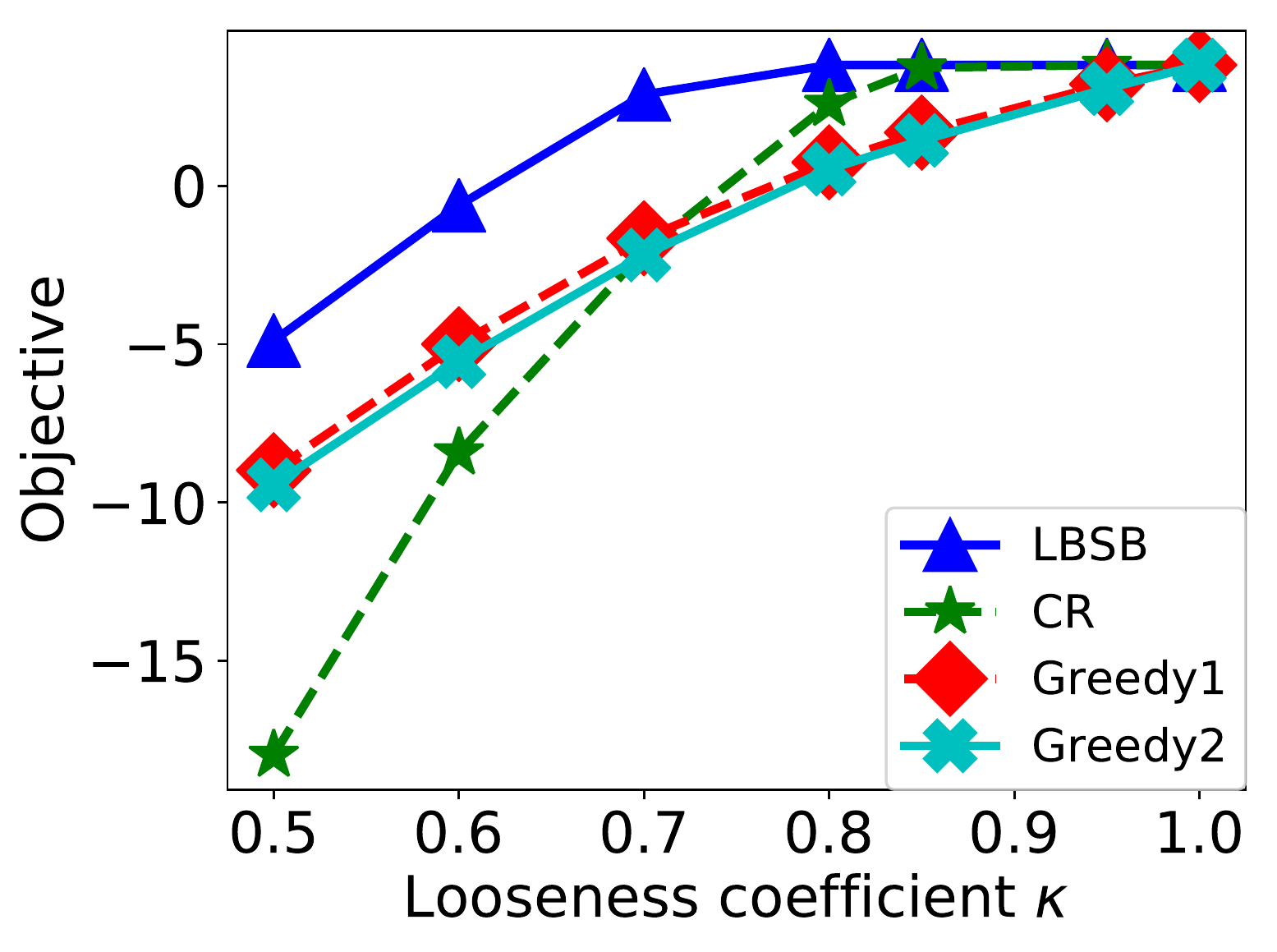}\label{fig:abilene}
	}
		\subfloat[\texttt{geant}]{
		\includegraphics[width=0.48\columnwidth]{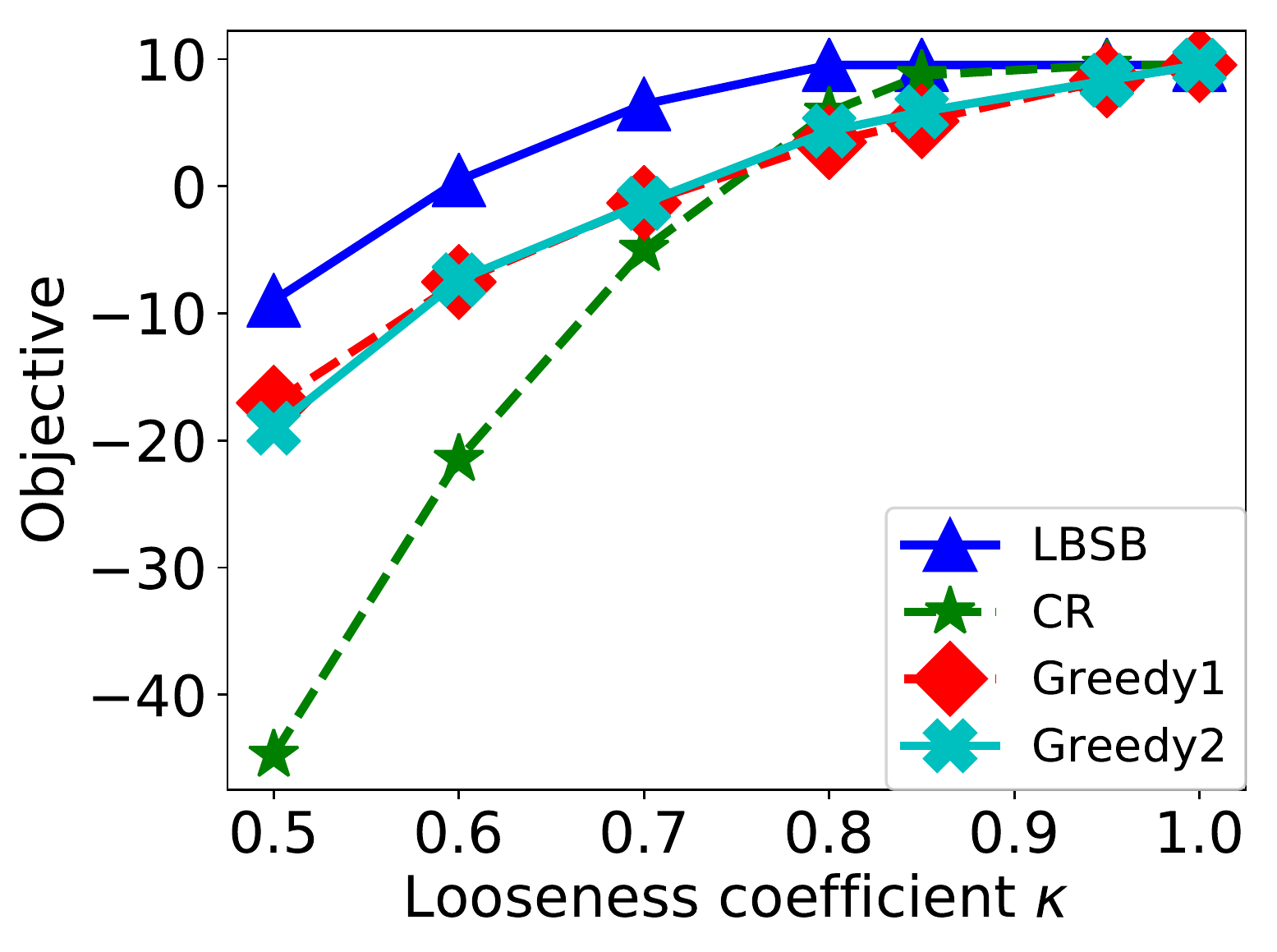}\label{fig:geant}
	}
\caption{Effects of tightening the constraints. The figure shows the objective values w.r.t.  the looseness coefficient $\kappa$, for two topologies \texttt{abilene}   and \texttt{geant}. Note that solutions for all 4 algorithms are feasible in all cases.}\label{congestion}
\end{figure}
}
{
\begin{figure*}[!t]
\vspace{-1mm}
	\centering
	\subfloat[\texttt{abilene}]{
		\includegraphics[width=0.33\textwidth]{plots/abilene_congestion.pdf}\label{fig:abilene}
	}
	\subfloat[\texttt{geant}]{
		\includegraphics[width=0.33\textwidth]{plots/geant_congestion.pdf}\label{fig:geant}
	}
	\subfloat[\texttt{cycle}]{
		\includegraphics[width=0.33\textwidth]{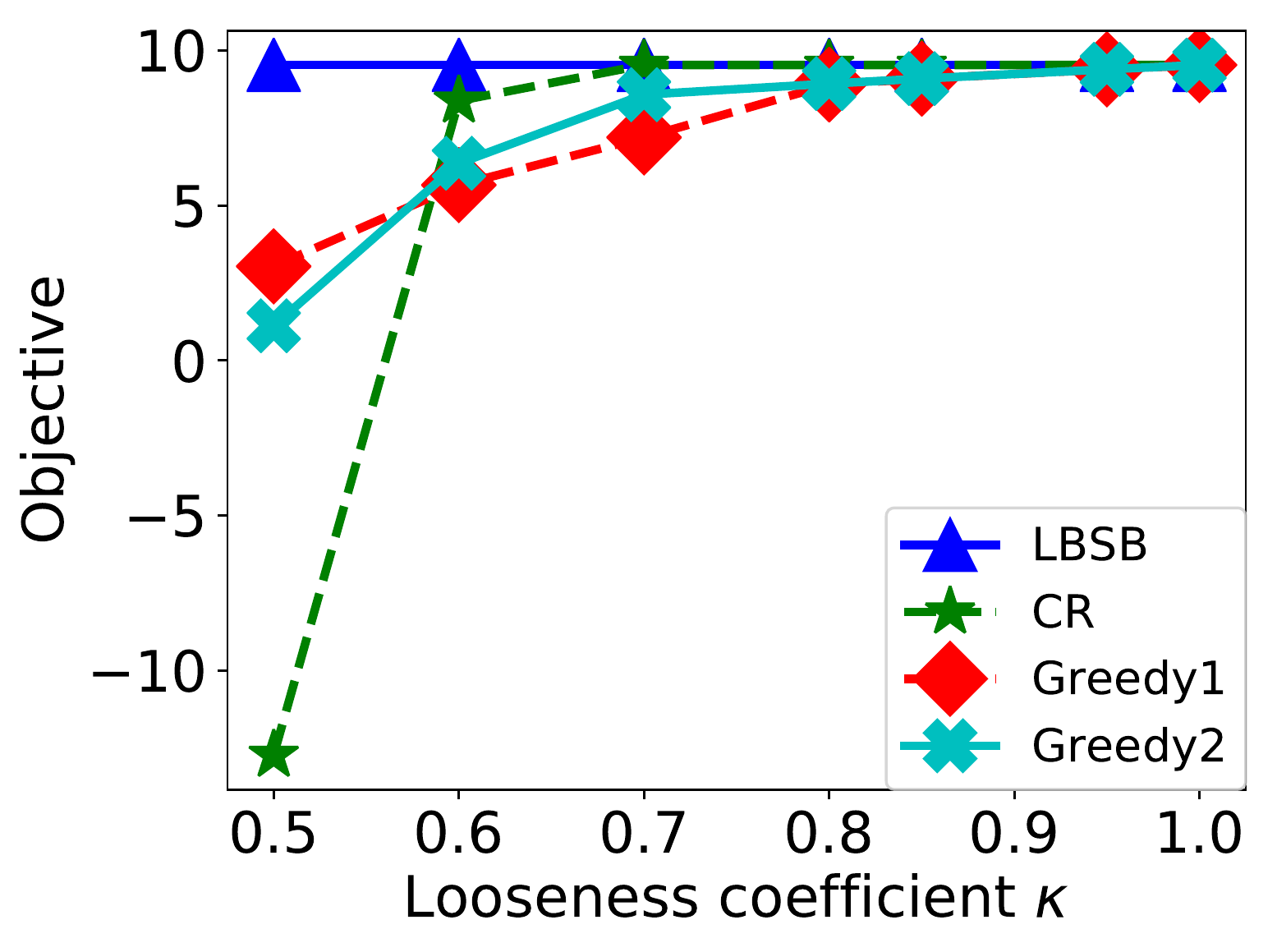}\label{fig:cycle}
	}
\caption{Effects of tightening the constraints. The figure shows the objective values w.r.t.  the looseness coefficient $\kappa$, for three topologies \texttt{abilene},  \texttt{geant}, and \texttt{cycle}. We see that as the constraints are tightened (i.e., $\kappa$ decreases) the gap between the objective values obtained by \texttt{LBSB} and other algorithms increases significantly, where \texttt{LBSB} delivers a superb performance. Moreover, for tighter settings, \texttt{CR} performance gets poor, e.g., for $\kappa=0.5$ \texttt{Greedy1} and \texttt{Greedy2} achieve higher objectives, in all three cases. Note that solutions for all 4 algorithms are feasible in all cases.}\label{congestion}
\end{figure*}
}

In this section, we evaluate the performance of the proposed methods, i.e., LBSB introduced in Section~\ref{subsec:LagrangianBarrierForUtilityMax} and the convex relaxation introduced in Section~\ref{subsec:concaverelax}. We also implement two greedy algorithms and compare the performance of the proposed methods against these greedy algorithms. \fullversion{}{As discussed below, we observe that our methods demonstrate impressive performance; for example, in Fig.~\ref{utility} we observed that out of 20 scenarios, the  Lagrangian barrier and the convex relaxation methods attain higher objective values than the greedy algorithms in 19 and 15 scenarios, respectively. }

\fullversion{}{\noindent\textbf{Topologies.} The networks we consider are summarized in Table~\ref{networks}.  Graph \texttt{cycle} is a simple cyclic graph, and \texttt{lollipop} is a clique (i.e., complete graph), connected to a path graph of equal size. The next 3 graphs represent the Deutche Telekom, Abilene, and GEANT backbone networks \cite{rossi2011caching}. Graph \texttt{grid-2d} is a two-dimensional square grid, \texttt{balanced-tree} is a complete binary tree of depth 5, and \texttt{hypercube} is a 6-dimensional hypercube. Finally, the last two graphs are random; \texttt{small-world} is the graph by Kleinberg \cite{kleinberg2000small}, which comprises a grid with additional long range links, and \texttt{erdos-renyi} is an Erd\H{o}s-R\'enyi graph with parameter $p=0.1$.}

\noindent\textbf{Experiment Setup.} \fullversion{We evaluate our algorithms on 10 graphs summarized in Table~\ref{networks}.}{We evaluate our algorithms on \texttt{cycle}, \texttt{lollipop} (which is a complete graph connected to a path graph of equal size), Deutche Telekom, Abilene, and GEANT backbone networks \cite{rossi2011caching}, \texttt{grid-2d}, \texttt{balanced-tree}, \texttt{hypercube}, \texttt{small-world}, and \texttt{erdos-renyi} graph with parameter $p=0.1$.}
Given a graph $G(\mathcal{V},\mathcal{E})$, we generate a catalog $\itemcat$, and assign a cache to each node in the graph. For every item $i\in \itemcat$, we designate a source node selected uniformly at random (u.a.r.)~from $\mathcal{V}$. We set the capacity $c_v$ of every node $v$ so that $c_v'=c_v-|\{i:v\in S_i\}|$ is constant among all nodes in $\mathcal{V}$.  We then generate a set of requests $\requestset$ as follows. First, we select a set $Q$ nodes in $\mathcal{V}$ selected u.a.r., that we refer to as \emph{query nodes}: these are the only nodes that generate requests. 
More specifically, for each query node $v\in Q,$ we generate $\approx |\requestset|/|Q|$ requests according to a Zipf distribution with parameter $1.2$ and without replacement  from the catalog $\itemcat.$ 
 Each request is then routed over the shortest path between the query node and the designated source for the requested item. We assign a demand rate $\bar{\lambda}_{(i,p)}=1$ to every request $n \in \requestset$. 
  The values of $|\itemcat|$, $|\requestset|$, $|Q|$, and $c_v$ for each topology are given in Table~\ref{networks}. Our process also makes sure that each item $i\in \itemcat$ is requested at least once.
  
  We determine the link capacities $C_{ba},$ $(b,a)\in \mathcal{E},$ as follows. First, note that the maximum possible load on each link $(b,a)\in \mathcal{E}$ is $\lambda_{ba}^{(\text{max})}\triangleq\sum_{(i,p) : (a,b) \in p} \bar{\lambda}_{(i,p)}$.  We  set the link capacities as $C_{ba} = \kappa \lambda^{(\text{max})}_{ba},$ where  $\kappa\in(0,1]$ is a \emph{looseness coefficient}: the higher $\kappa$ is, the easier it becomes to satisfy the demand. Note that for every link $(b,a)$, if $C_{ba}\geq \lambda_{ba}^{(\text{max})}$ (or equivalently $\kappa\geq 1$), then the link constraint corresponding to $(b,a)$ in \eqref{const:link} is trivially satisfied.
  In contrast, as $\kappa$ decreases below $1$, the link constraints are tightened, and finding optimal rate and cache allocation strategies becomes non-trivial. For each topology in Table~\ref{networks}, we study two settings, i.e., (1) a \emph{loose} setting, where $\kappa=0.95$ and (2) a \emph{tight} setting, where $\kappa=0.85.$ 


\noindent\textbf{Algorithms.} We implement\footnote{Our code is publicly available at \href{https://github.com/neu-spiral/UtilityMaximizationProbCaching}{https://github.com/neu-spiral/UtilityMaximizationProbCaching}.} Alg.~\ref{alg:findKKT} and refer to it as \texttt{LBSB}. We run this algorithm until the convergence criterion in \cite{conn1997globally} is met (with $\delta_*=10^{-4}, \omega_*=10^{-4}$).  We also solve the convex relaxation in  \eqref{cachingproblem2} via a sub-gradient method, as described in Section~7.5 of Bertsekas \cite{bertsekas1999nonlinear}; we refer to this algorithm by \texttt{CR}, for convex relaxation. We run \texttt{CR} for a fixed number of iterations (500 iterations). In addition, we implement two greedy algorithms, i.e., \texttt{Greedy1} and \texttt{Greedy2}.  We describe each of these algorithms below; we stress that, as we are the first to consider problem \textsc{UtilityMax}, there are no prior art algorithms to compare with. 
\begin{itemize}
\item \texttt{Greedy1} consists of three steps; in Step 1, we initialize $\boldsymbol{Y}=0$ and update $\boldsymbol{R}$ by solving \eqref{cachingproblem1} only w.r.t. $\boldsymbol{R}.$ This is a convex optimization problem. 
In Step 2, \texttt{Greedy1} keeps $\boldsymbol{R}$ fixed, as computed by Step 1, and updates $\boldsymbol{Y}$ by maximizing the sum $\sum_{(b,a)\in \mathcal{E}}g_{ba}(\boldsymbol{Y},\boldsymbol{R}),$ subject to the constraints~\eqref{const:cache} and \eqref{const:boxfory} and only w.r.t. $\boldsymbol{Y}.$ This is equivalent to minimizing the total long-term time-average item load over the links of the graph (c.f. Section~\ref{sect:congestioncontrol}). Note that Step 2 is a monotone DR-submodular maximization subject to a polytope, which we solve via the Frank-Wolfe algorithm proposed by Bian et. al. \cite{bian2017guaranteed}. Finally in Step 3, \texttt{Greedy1} updates $\boldsymbol{R}$ by solving \eqref{cachingproblem1}  w.r.t. $\boldsymbol{R},$ one more time, while  $\boldsymbol{Y}$ is fixed to the value computed by Step 2. 

\item 
\texttt{Greedy2} \fullversion{initializes $\boldsymbol{Y}=0$, and then alternatively updates $\boldsymbol{Y}$ and  $\boldsymbol{R}$}{is an alternating optimization algorithm. We initialize $\boldsymbol{Y}=0$, and then alternatively update $\boldsymbol{Y}$ and  $\boldsymbol{R}$, while keeping the other variable  fixed at a time}; we refer to the former step as the cache allocation step and to the latter as the rate allocation step. 

In the cache allocation step, \texttt{Greedy2} ``greedily'' places one item to a cache: it changes one zero variable ($y_{vi}=0$) to 1, where $(v, i)$ is the feasible pair with the largest marginal gain in load reduction. Formally, (a) node $v$ has not fully used its cache capacity ($g_v(\boldsymbol{Y})<c_v'$) and (b) changing $y_{vi}$ from 0 to 1 has the highest increase in the total sum $\sum_{(b,a)\in \mathcal{E}} g_{ba}(\boldsymbol{Y},\boldsymbol{R})$ (or equivalently the highest decrease in the aggregate item load). 
In the rate allocation step, \texttt{Greedy2} keeps $\boldsymbol{Y}$ constant from the previous step, and updates $\boldsymbol{R}$ by solving \eqref{cachingproblem1} w.r.t.~$\boldsymbol{R}$. 

Finally, \texttt{Greedy2} terminates once node storage capacities are depleted, i.e., there is no pair $(v, i)$ left, s.t., $y_{vi}=0$ and $y_v < c_v'.$  
\end{itemize}

\noindent\textbf{Metrics.} Throughout the experiments we report the objective function $F(\boldsymbol{R})$ obtained by different algorithms for the choice of the logarithmic utility functions $U_{\requestindex}(\lambda)= \log (\lambda + 0.1),$ for all $\requestindex \in \requestset.$ Note that these utility functions satisfy Assumption~\ref{assm:logreturn} and Assumption~\ref{assm:unbounded}. \fullversion{We observe that all algorithms  generate    feasible solutions for all cases.}{To assess feasibility during the progression of an algorithm, we also report the ratio of satisfied constraints, i.e., the fraction of constraints in \eqref{const:link} and \eqref{const:cache} that are satisfied. 
}

\noindent\textbf{Objective Performance.}
Fig.~\ref{utility} shows  objectives attained by different algorithms,  normalized by the objective under \texttt{LBSB}; the latter is given in the $\hat{F}_{\text{(loose)}}$ and $\hat{F}_{\text{(tight)}}$ columns of Table~\ref{networks}, for the loose $(\kappa = 0.95)$ and  tight $(\kappa = 0.85)$ settings, respectively. We observe that \texttt{LBSB} outperforms all its competitors across all
topologies in both settings, except for one case (\texttt{balanced-tree}
with $\kappa=0.85.$) In Fig.~\ref{bar95}, we see that in the loose setting, \texttt{CR} performance almost matches \texttt{LBSB} and achieves better objective values in comparison with \texttt{Greedy1} and \texttt{Greedy2}. However, in Fig.~\ref{bar8} we see that as the constraint set is tightened, the performance of  \texttt{CR} deteriorates. Moreover, by comparing Fig.~\ref{bar95} and  Fig.~\ref{bar8} we see that for some topologies (e.g., \texttt{grid-2d}, \texttt{hypercube}, \texttt{small-world}, and \texttt{erdos-renyi}), the gap between the objective values obtained by \texttt{LBSB} and other algorithms is significantly higher in the tight constraints regime. 

\noindent\textbf{Execution Time.}
In Fig.~\ref{fig:time}, we plot the execution times of all algorithms for each scenario as a function of the number of variables in the corresponding instance of problem \textsc{UtilityMax} (reported in the last column of Table~\ref{networks}). Figures~\ref{time95} and \ref{time8} correspond to Figures~\ref{bar95} and \ref{bar8} (the loose and tight settings), respectively. In particular, in the loose setting (Fig.~\ref{time95}) we see that the execution times for \texttt{LBSB} almost match the execution times for \texttt{Greedy2}, and  \texttt{LBSB} is much faster than \texttt{CR}.  In the tight setting (Fig.~\ref{time8}), however, we see that the execution time of \texttt{LBSB} is  higher, particularly when the number of variables is large (corresponding to larger topologies, i.e., \texttt{grid-2d}, \texttt{balanced-tree}, \texttt{hypercube}, \texttt{small-world}, and \texttt{erdos-renyi}). The main reason is that, in the tight setting, the trust-region algorithm used as a subroutine at each iteration requires a higher number of iterations to satisfy  \eqref{equ:appforinner}; as a result, the execution time for \texttt{LBSB} increases. Nonetheless, as we observed in Fig.~\ref{bar8}, \texttt{LBSB} achieves significantly improved objective performance compared to other algorithms.

\fullversion{}{
\noindent\textbf{Convergence.} To obtain further insight into the convergence of \texttt{LBSB}, we plot in Fig.~\ref{fig:trajgrid}  the objective and feasibility (i.e., the ratio of satisfied constraints \eqref{const:link} and \eqref{const:cache}) trajectories as a function of time. Each marker in these plots corresponds to an outer iteration of \texttt{LBSB}. We observe that, initially, the solutions are infeasible and the objective value is high, as the algorithm converges, the feasibility improves and consequently the objective decreases. We stress here that although the constraint satisfaction ratio in Fig.~\ref{fig:trajgrid} does not reach 1, the highest constraint violation at the last displayed iteration is in the order $10^{-8}$ (and, hence, our convergence criterion was met).   For brevity, we show only the trajectory for \texttt{grid-2d} and $\kappa=0.85$, for which \texttt{LBSB} is slower than other algorithms. 
}

\noindent\textbf{Effect of Tightening Constraints.}
Motivated by our observation regarding the superior performance of \texttt{LBSB} in the tighter setting, we study the effects of further decreasing the looseness coefficient $\kappa$. In Fig.~\ref{congestion}, we plot the objective values achieved by different algorithms for looseness coefficients $\kappa=0.5, 0.6, 0.7, 0.8, 0.85, 0.95$, and $1$. For brevity, we report these results only for \fullversion{\texttt{abilene} and  \texttt{geant}.}{\texttt{abilene},  \texttt{geant}, and \texttt{cycle}.} From Fig.~\ref{congestion}, we observe that for  \fullversion{both}{all three} topologies, when $\kappa=1$, all algorithms achieve the optimal objective value; this is expected, because as explained, when $C_{ba}\geq \lambda_{ba}^{(\text{max})},$ the non-convex constraints \eqref{const:link} are trivially satisfied. As we tighten the constraints by decreasing $\kappa$, we observe that all algorithms obtain smaller objectives, which is also expected. Crucially, \texttt{LBSB} significantly outperforms other algorithms and remains quite resilient to tightening of the constraints; for example, for  \fullversion{both}{all three} topologies, when $\kappa\geq0.8$ \texttt{LBSB} still obtains the maximum objective value, i.e., the same value as  with $\kappa=1.$  In fact, for \texttt{cycle}, the performance of  \texttt{LBSB} remains practically invariant, while all other algorithms deteriorate.

Moreover,  we also see in Fig.~\ref{congestion} that for moderate tightness of constraints (e.g., $\kappa\geq 0.8,$ \texttt{CR}) shows decent performance and outperforms \texttt{Greedy1} and \texttt{Greedy2}; however, tightening the constraints further, e.g., for $\kappa\leq 0.7,$ grossly affects the performance of  \texttt{CR}: the objective values decrease significantly, falling below that of the greedy algorithms. In fact, this is expected from Thm.~\ref{thm:approximationBound}. To see this, note that when $\kappa <  1/e \approx 0.36,$ for the capacities in \eqref{eq:cba_prime} we have $C_{ba}' < 0$, for all $(a,b) \in \mathcal{E}$. As a result, based on Thm.~\ref{thm:approximationBound}, for $\kappa < 1/e,$ the lower bound on the optimal objective of \texttt{CR} is non-existent, as the constraint set $\mathcal{D}_3$ (see \eqref{equ:Dprime}) is
an empty set. In other words, in this regime, \texttt{CR} comes with no guaranteed lower bound.

\section{Conclusion} \label{sec:conc}
We studied a new class of non-convex optimization problems for joint content placement and rate allocation in cache networks, and proposed solutions with optimality guarantees. Our solutions establish a foundation for several  possible future investigations. First, in the spirit of Kelly et al.~\cite{kelly1998rate}, studying distributed algorithms that converge to a KKT point, and providing similar guarantees as Thm.~\ref{thm:assymptotic}, is an important open question. \fullversion{}{Ideally, such algorithms would amount to protocols that adjust both rates as well as caching decisions in a way that leads to optimality guarantees.}  
Second, in both the centralized/offline setting we study here, as well as a distributed/adaptive setting, designing new rounding techniques for deterministic content placement is another open question. \fullversion{}{Finally, both providing lower bounds for \textsc{UtilityMax}, or devising algorithms with better/tighter optimality guarantees, are additional open questions in the context of our problem.}


\section*{Acknowledgment}

The authors gratefully acknowledge support from National Science Foundation grants NeTS-1718355 and CCF-1750539, and a research grant from American Tower Corp.

\appendices

\fullversion{}{\section{Probabilistic Content Placement Algorithm} \label{append:probcachealg}
Here we describe a distributed and random content placement algorithm \cite{geocaching2015, ioannidis2016addaptive}. Each node $v \in \mathcal{V}$ has access to its cache allocation strategy $[\hat{y}_{vi}]_{i \in \itemcat}$ obtained by solving Problem~\eqref{cachingproblem1}. Our goal at the beginning of the $t$-th  time period, is to use $[\hat{y}_{vi}]_{i \in \itemcat}$ as marginal cache probabilities at node $v$, and provide a probabilistic content placement (i.e., mapping of content items to the cache) $\boldsymbol{\hat{X}}(t) \triangleq [\hat{x}_{vi}(t)]_{v \in \mathcal{V}, i \in \itemcat}$. The probabilistic content placement must ensure $\hat{y}_{vi} = Pr \{ \hat{x}_{vi}(t) = 1 \}$, and that cache capacity constraint is satisfied exactly, i.e.,
\begin{equation*}
    \sum_{i} \hat{x}_{vi}(t) \leq c_v \quad \forall t>0.
\end{equation*} 

If we naively pick $\hat{x}_{vi}$ independently using a Bernoulli distribution with marginal probabilities $[\hat{y}_{vi}]_{i \in \itemcat}$, the capacity constraint is only satisfied in expectation, and at each time, the cache may store fewer or more items than its capacity. To construct a desirable placement,
consider a rectangle box of area $c_v \times 1$. For each $i \in \itemcat$, place a rectangle of length $\hat{y}_{vi}$ and height $1$ inside the box,
starting from the top left corner. If a rectangle does not fit in a row, cut it, and place the remainder in the row immediately below, starting again from the left. As $\sum_{i\in \itemcat} \hat{y}_{vi} \leq c_v$, this space-filling method is contained in the $c_v \times 1$ box. In order to randomly choose a
set of items, at the beginning of the $t$-th time period we pick uniformly at random a number $\tau^{(t)} \in [0, 1]$ and draw a vertical line located at that number which intersects the box area by no more than $c_v$ distinct items. The items are distinct because $y_{vi} \in [0,1]$. Moreover, the probability of appearance of item $i$ in a memory of size $c_v$ is exactly equal to $\hat{y}_{vi}$. A graphical explanation of this algorithm is presented in Fig.~\ref{fig:ProbCaching}. 

\begin{figure} [!t]
  \centerline{\includegraphics[scale = 0.4]{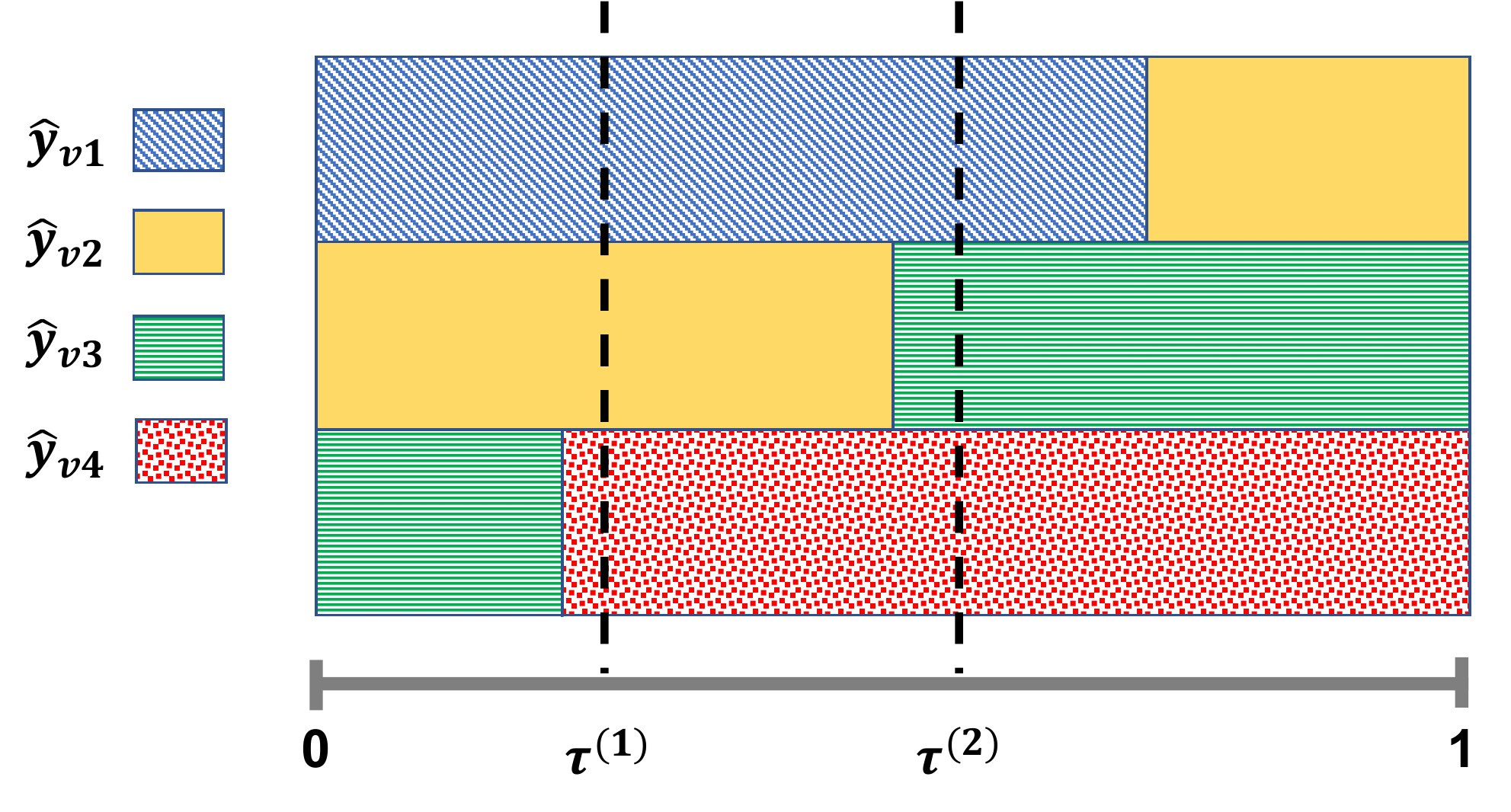}}
  \caption{A content placement satisfying cache capacity constraints \eqref{equ:cachecv}, using $[y_{vi}]_{i \in \itemcat}$ as marginal distribution. Here, $c_v = 3$ and $\itemcat = \{1, 2, 3, 4\}$. At the beginning of the first time period, a random number $\tau^{(1)} \in [0,1]$ is chosen. The corresponding vertical line intersects rectangles of item 1, item 2, and item 4. Thus, the set of $\{1,2,4\}$ is stored in the cache for the duration of this time period. At the beginning of the second time period, another random number $\tau^{(2)} \in [0,1]$ is chosen, and the set $\{ 1,3,4 \}$ is stored in the cache for the duration of this time period.}
  \label{fig:ProbCaching}
\end{figure}}

\section{Proof of Lemma \ref{lem:monotonedr}} \label{append:proofoflemmonotonedr}
By Eq.~\eqref{eq:gba}, $\frac{\partial^2 g_{ba}(\boldsymbol{Y},\boldsymbol{R})}{\partial y_{vi} \partial y_{v'i'}} \leq 0$, $\frac{\partial^2 g_{ba}(\boldsymbol{Y},\boldsymbol{R})}{\partial y_{vi} \partial r_n} \leq 0$, and $\frac{\partial^2 g_{ba}(\boldsymbol{Y},\boldsymbol{R})}{\partial r_{n} \partial r_{n'}} \leq 0$, for all $v, v' \in \mathcal{V}$, $i, i' \in \itemcat$, and $n, n' \in \requestset$. This proves the DR-submodularity of $g_{ba}(\cdot)$, for all $(b,a) \in \mathcal{E}$ (See \fullversion{\cite{bian2017guaranteed}}{Section~\ref{subsec:drsub}}). In addition, since $\frac{\partial g_{ba}(\boldsymbol{Y},\boldsymbol{R})}{\partial y_{vi}} \geq 0$, and $\frac{\partial g_{ba}(\boldsymbol{Y},\boldsymbol{R})}{\partial r_{n}} \geq 0$, for all $v \in \mathcal{V}, i \in \itemcat, \requestindex \in \requestset$, $g_{ba}(\cdot)$ are monotone, for all $(b,a) \in \mathcal{E}$. \qed

\fullversion{}{\section{A Trust-Region Algorithm for Simple Box Bounds} \label{append:trustregion}
We briefly describe a trust-region algorithm proposed by Conn et. al. \cite{trustRegion}. This algorithm is used to find a point satisfying the necessary optimality condition for the following problem:
\begin{align}  \label{prob:appendtrustpsi}
    \max_{\boldsymbol{x}} \quad &\Psi(\boldsymbol{x}) \\
    \text{s.t.} \quad &\boldsymbol{x} \in \mathcal{B}, \nonumber
\end{align}
where $\mathcal{B}$ is a region consisting of simple box constraint. The necessary optimality condition for Problem \eqref{prob:appendtrustpsi} can be written as
\begin{equation*}
    || P\big( \boldsymbol{x}~,~ \nabla_{\boldsymbol{x}} \Psi(\boldsymbol{x}) \big) || = 0.
\end{equation*}
where $P(\boldsymbol{a}~,~\boldsymbol{b}) \triangleq \boldsymbol{a} - \Pi_{\mathcal{B}}(\boldsymbol{a} + \boldsymbol{b})$, and $\Pi_{\mathcal{B}}(\boldsymbol{a})$ is the projection of vector $\boldsymbol{a}$ on to the box region $\mathcal{B}$. We use this algorithm to find a point satisfying the condition at line \ref{line:inner} of Alg.~\ref{alg:findKKT}. 

The trust-region algorithm starts from an initial point inside the box constraint $\mathcal{B}$. At the $k$-th iteration, it performs the projected gradient ascent; $s(t)$ is the projected gradient ascent direction, parameterized by the step size $t$. This step size $t$ is chosen in such that (a) it is within a \emph{trust region} defined with $\Delta_k$ and (b) it is the smallest local maximum of $Q_{\Psi}^{(k)}$, i.e., the second-order approximation of the objective $\Psi(\cdot)$ around the current point $\boldsymbol{x}_{k}$. Next, the algorithm assesses the improvement of the objective along the direction of $s_k$. If the improvement is above some threshold $\mu$, it accepts the direction $s_k$ and updates the solution  
, and enlarges the trust region. Otherwise, the algorithm rejects the direction and does not update the solution, and shrinks the trust region. These steps are outlined in Alg.~\ref{alg:trustRegion}

\begin{algorithm}[!t] 
    \caption{Trust-Region Algorithm for Simple Box Constraints}
    \label{alg:trustRegion}
    \begin{algorithmic}[1] 
        \State Set parameters $\mu \in (0,1), \eta \in (\mu, 1), 0 < \gamma_0\leq \gamma_1\leq 1 \leq \gamma_2$
        \State Set the initial point $\boldsymbol{x}_{-1} \in \mathcal{B}$
        \State $k \gets -1$
        \While{Not converged}
        \State $s(t) \triangleq \Pi_{\mathcal{B}}\left(\boldsymbol{x}_k + t  \nabla \Psi(\boldsymbol{x}_k) \right) - \boldsymbol{x}_k$ \label{stp:s_t}
        \State $t_k$ $\gets$ smallest local maximum of $Q_{\Psi}^{(k)}\big(s(t) + \boldsymbol{x}_{k}  \big)$
        
        subject to $s(t) \leq \Delta_k$ \label{stp:t_ck}
        \State $s_k\gets s(t_k)$
        \State $\rho_k\gets \frac{\Psi( \boldsymbol{x}_k + s_k)- \Psi( \boldsymbol{x}_k) }{  Q_{\Psi}^{(k)}( \boldsymbol{x}_{k} + s_k )-\Psi( \boldsymbol{x}_k)}$ \label{stp:improvement}
        \If{$\rho_k \geq \mu$}
        \State $\boldsymbol{x}_{k+1} \gets \boldsymbol{x}_k + s_k$\label{stp:updateSol}
        \If{$\rho_k \geq \eta$}
        \State $\Delta_{k+1} \gets \Delta \in  [\Delta_k, \gamma_2\Delta_k ]$\label{stp:delta1}
        \Else
        \State $\Delta_{k+1} \gets \Delta \in  [\gamma_1\Delta_k, \Delta_k ]$\label{stp:delta2}
       \EndIf
       \Else
       \State $\boldsymbol{x}_{k+1} \gets \boldsymbol{x}_k $\label{stp:NotupdateSol}
        \State $\Delta_{k+1} \gets \Delta \in  [\gamma_0\Delta_k, \gamma_1\Delta_k ]$\label{stp:delta3}
       
        \EndIf
        
        \EndWhile
    \end{algorithmic}
\end{algorithm}}

\fullversion{}{\section{Lagrangian Barrier With Simple Bounds}\label{append:LBSB}
We describe the details of the Lagrangian Barrier with Simple Bounds (in short LBSB) by Conn et. al. \cite{conn1997globally}. For simplicity and avoid repetition, let us define an index set for inequality constraints:
\begin{equation*}
    \lbsbconstset \triangleq \{ (b,a) \in \mathcal{E} \} \cup \{ v \in \mathcal{V} \}.
\end{equation*}
Thus, we can denote all inequality constraint \eqref{const:link}, \eqref{const:cache} with
\begin{align*}
    & c_\lbsbconstindex(\boldsymbol{Y},\boldsymbol{R}) \triangleq \\
    &\begin{cases} 
      g_{ba}(\boldsymbol{Y},\boldsymbol{R}) - \sum_{(i,p): (a,b) \in p} \bar{\lambda}_{(i,p)} + C_{ba} & \lbsbconstindex = (b,a) \in \mathcal{E} \\
      c_v - g_v(\boldsymbol{Y}) & \lbsbconstindex = v \in \mathcal{V}, 
   \end{cases} 
\end{align*}
for all $\lbsbconstindex \in \lbsbconstset$. 
In addition, we denote the box constraint set define by \eqref{const:boxfory}, \eqref{const:boxforlambda} with $\mathcal{B}$. 
Problem \eqref{cachingproblem1} is then written as
\begin{equation} \label{connstyleproblem}
\begin{aligned}
\max \quad & F(\boldsymbol{R}) \\
\textrm{s.t.} \quad & c_\lbsbconstindex(\boldsymbol{Y},\boldsymbol{R}) \geq 0 \quad \forall \lbsbconstindex \in \lbsbconstset\\
  &(\boldsymbol{Y},\boldsymbol{R}) \in \mathcal{B}. 
\end{aligned}
\end{equation}

LBSB defines the Lagrangian barrier function
\begin{equation*} 
    \Psi(\boldsymbol{Y}, \boldsymbol{R},\boldsymbol{\Sigma},\boldsymbol{s}) \triangleq F(\boldsymbol{R}) + \sum_{\lbsbconstindex \in \lbsbconstset} \sigma_\lbsbconstindex s_\lbsbconstindex \log(c_\lbsbconstindex(\boldsymbol{Y},\boldsymbol{R}) + s_\lbsbconstindex)
\end{equation*}
where the elements of the vector $\boldsymbol{\Sigma} \triangleq [ \sigma_\lbsbconstindex]_{\lbsbconstindex \in \lbsbconstset} \in \mathbb{R}_+^{|\lbsbconstset|}$ are the positive \emph{Lagrange multipliers estimates} associated with inequality constraints $c_\lbsbconstindex(\cdot) \geq 0 ~ \forall \lbsbconstindex \in \lbsbconstset$ in \eqref{connstyleproblem}. The elements of the vector $\boldsymbol{s} \triangleq [s_\lbsbconstindex]_{\lbsbconstindex \in \lbsbconstset} \in \mathbb{R}_+^{|\lbsbconstset|}$ are positive \emph{shifts}.

Writing the gradient of the Lagrangian barrier function, we have:
\begin{align} 
    &\nabla_{\boldsymbol{Y},\boldsymbol{R}} \Psi(\boldsymbol{Y}, \boldsymbol{R},\boldsymbol{\Sigma},\boldsymbol{s}) \nonumber \\
    & = \nabla_R F(\boldsymbol{R}) + \sum_{\lbsbconstindex \in \lbsbconstset} \frac{\sigma_\lbsbconstindex s_\lbsbconstindex}{c_\lbsbconstindex(\boldsymbol{Y},\boldsymbol{R}) + s_\lbsbconstindex}  \nabla_{\boldsymbol{Y},\boldsymbol{R}} c_\lbsbconstindex(\boldsymbol{Y},\boldsymbol{R}) \label{equ:firstorderestimates}
\end{align}
The values multiplied by the gradient of the constraints in \eqref{equ:firstorderestimates} are called \emph{first-order Lagrange multiplier approximations}:
\begin{equation*}
    \bar{\sigma}_\lbsbconstindex(\boldsymbol{Y}, \boldsymbol{R},\boldsymbol{\Sigma},\boldsymbol{s}) \triangleq \frac{\sigma_\lbsbconstindex s_\lbsbconstindex}{c_\lbsbconstindex(\boldsymbol{Y},\boldsymbol{R}) + s_\lbsbconstindex} 
\end{equation*}
The vector of first-order Lagrangian multiplier approximations is denoted by $\bar{\Sigma}(\boldsymbol{Y}, \boldsymbol{R},\boldsymbol{\Sigma},\boldsymbol{s}) \triangleq [ \bar{\sigma}_\lbsbconstindex(\boldsymbol{Y}, \boldsymbol{R},\boldsymbol{\Sigma},\boldsymbol{s})]_{\lbsbconstindex \in \lbsbconstset} \in \mathbb{R}^{|\lbsbconstset|}$, and is used in updating the Lagrangina multiplier estimates in LBSB. 

Consider the following problem:
\begin{align} 
    \max_{(\boldsymbol{Y}, \boldsymbol{R})} \quad &\Psi(\boldsymbol{Y}, \boldsymbol{R}, \boldsymbol{\Sigma}_k, \boldsymbol{s}_k) \label{prob:innersubappend} \\
    \text{s.t.} \quad &(\boldsymbol{Y}, \boldsymbol{R}) \in \mathcal{B}, \nonumber
\end{align}
where the values $\boldsymbol{\Sigma}_k$ and $\boldsymbol{s}_k$ are given. Then the following condition is the necessary optimality condition for Prob. \eqref{prob:innersubappend}:
\begin{equation*}
    \bigg|\bigg| P\bigg( (\boldsymbol{Y},\boldsymbol{R})~,~ \nabla_{\boldsymbol{Y},\boldsymbol{R}} \Psi(\boldsymbol{Y}, \boldsymbol{R}, \boldsymbol{\Sigma}_k, \boldsymbol{s}_k) \bigg) \bigg|\bigg| = 0,
\end{equation*}
where $P(\boldsymbol{a}~,~\boldsymbol{b}) \triangleq \boldsymbol{a} - \Pi_{\mathcal{B}}(\boldsymbol{a} + \boldsymbol{b})$, and $\Pi_{\mathcal{B}}(\boldsymbol{a})$ is the projection of vector $\boldsymbol{a}$ on to the box region $\mathcal{B}$. LBSB uses this condition as explained below.

After setting initial parameters, LBSB proceeds as follows. At the $k-$th iteration, it updates $(\boldsymbol{Y}_k, \boldsymbol{R}_k)$ by finding a point in $\mathcal{B}$, such that the following condition is satisfied:
\begin{equation*}
    \bigg|\bigg| P\bigg( (\boldsymbol{Y},\boldsymbol{R})~,~ \nabla_{\boldsymbol{Y},\boldsymbol{R}} \Psi(\boldsymbol{Y}, \boldsymbol{R}, \boldsymbol{\Sigma}_k, \boldsymbol{s}_k) \bigg) \bigg|\bigg| \leq \omega_k,
\end{equation*}
where parameter $\omega_k$ indicates the accuracy of the solution; when $\omega_k = 0$, the point $(\boldsymbol{Y}_k, \boldsymbol{R}_k)$ satisfies the necessary optimality conditions. 
LBSB is designed to be \emph{locally convergent} to a KKT point if the penalty parameter $\epsilon_k$ is fixed at a sufficiently
small value, and the Lagrange multipliers estimates are updated using their first order approximations. The following condition is to detect whether we are able to move from a globally convergent to a locally convergent regime, using a tolerance parameter $\delta_k$:
\begin{equation} \label{equ:localconvregappend}
\bigg|\bigg| \big[\frac{c_\lbsbconstindex(\boldsymbol{Y}_k,\boldsymbol{R}_k) \bar{\sigma}_\lbsbconstindex (\boldsymbol{Y}_k, \boldsymbol{R}_k,\boldsymbol{\Sigma}_k,\boldsymbol{s}_k)}{\sigma_{k,\lbsbconstindex}^{\alpha_\sigma}}  \big]_{\lbsbconstindex \in \lbsbconstset} \bigg|\bigg| \leq \delta_k    
\end{equation}

After updating $(\boldsymbol{Y}_k, \boldsymbol{R}_k)$, if the condition at \eqref{equ:localconvregappend} is satisfied, Lagrangian multiplier estimates are updated using their first order approximations $\bar{\Sigma}(\boldsymbol{Y}_k, \boldsymbol{R}_k,\boldsymbol{\Sigma}_k,\boldsymbol{s}_k)$, penalty parameter $\epsilon_{k+1}$ stays the same, $\omega_{k+1}$ and $\delta_{k+1}$ are updated using their previous value $\omega_k$ and $\delta_k$. Otherwise, it means that the penalty parameter is not small enough. Thus, the Lagrangian multipliers are not changed, penalty parameter is decreased, and $\omega_{k+1}$ and $\delta_{k+1}$ are updated using the initial parameters $\omega_s$ and $\delta_s$. The iterations stop if both conditions below are satisfied:
\begin{subequations}\label{eq:convcret12}
\begin{align}
    || P\big( (\boldsymbol{Y}_k,\boldsymbol{R}_k), \nabla_{ \boldsymbol{Y},\boldsymbol{R}} \Psi(\boldsymbol{Y}_k, \boldsymbol{R}_k,\boldsymbol{\Sigma}_k,\boldsymbol{s}_k) \big) || \leq \omega_* \label{equ:convcret1}\\
    || \big[ c_\lbsbconstindex(\boldsymbol{Y}_k,\boldsymbol{R}_k) \bar{\sigma}_\lbsbconstindex (\boldsymbol{Y}_k, \boldsymbol{R}_k,\boldsymbol{\Sigma}_k,\boldsymbol{s}_k) \big]_{\lbsbconstindex \in \lbsbconstset} || \leq \delta_*. \label{equ:convcret2}
\end{align}
\end{subequations}
If the convergence criteria in \eqref{equ:convcret1} and \eqref{equ:convcret2} is not met, the algorithm proceeds to the next iteration. At the beginning of the next iteration, LBSB computes the shifts using the obtained parameters. The algorithm can guarantee that the penalty parameter is sufficiently small by driving it to zero, while at the same time ensuring that the Lagrange multiplier estimates converge to the Lagrange multipliers of the KKT point. The procedure is fully described in Alg.~\ref{alg:LBSB}. For more details refer to the paper by Conn et. al. \cite{conn1997globally}.

\begin{algorithm}[!t] 
    \caption{Lagrangian Barrier with Simple Bounds (LBSB)}
    \label{alg:LBSB}
    \begin{algorithmic}[1] 
        \State Set strictly positive parameters $\delta_s$, $\omega_s$, $\alpha_\omega$, $\beta_\omega$, $\alpha_\delta$, $\beta_\delta$, $\alpha_\sigma \leq 1$, $\tau < 1$, $\omega_* \ll 1$, $\delta_* \ll 1$, such that $\alpha_\delta + (1 + \alpha_\sigma)^{-1}$
        \State Set penalty parameter $\epsilon_0 < 1$
        \State Set initial solution $(\boldsymbol{Y}_{-1},\boldsymbol{R}_{-1}) \in \mathcal{B}$
        \State Set initial vector of Lagrangian multiplier estimates $\boldsymbol{\Sigma}_0 = [\sigma_{\lbsbconstindex,0}]_{\lbsbconstindex \in \lbsbconstset}$, such that $c_\lbsbconstindex(\boldsymbol{Y}_{-1},\boldsymbol{R}_{-1}) + \epsilon_0 \sigma_{\lbsbconstindex,0}^{\alpha_\epsilon} >0$ for all $\lbsbconstindex \in \lbsbconstset$
        \State Set accuracy parameter $\omega_0 \gets \omega_s \epsilon_0^{\alpha_\omega}$ 
        \State Set tolerance parameter $\delta_0 \gets \delta_s \epsilon_0^{\alpha_\delta}$
        \State $k \gets -1$
        \Repeat
            \State $k \gets k + 1$
            \State Compute shifts $s_{k,\lbsbconstindex} = \epsilon_k \sigma_{k,\lbsbconstindex}^{\alpha_\sigma}$ for all $\lbsbconstindex \in \lbsbconstset$
            \State Find $(\boldsymbol{Y}_k,\boldsymbol{R}_k) \in \mathcal{B}$ such that:
            
            $ || P\big( (\boldsymbol{Y}_k,\boldsymbol{R}_k), \nabla_{\boldsymbol{Y}_k,\boldsymbol{R}_k} \Psi(\boldsymbol{Y}, \boldsymbol{R}, \boldsymbol{\Sigma}_k, \boldsymbol{s}_k) \big) || \leq \omega_k$.
            
            \If{$|| \big[\frac{c_\lbsbconstindex(\boldsymbol{Y}_k,\boldsymbol{R}_k) \bar{\sigma}_\lbsbconstindex (\boldsymbol{Y}_k, \boldsymbol{R}_k,\boldsymbol{\Sigma}_k,\boldsymbol{s}_k)}{\sigma_{k,\lbsbconstindex}^{\alpha_\sigma}}  \big]_{\lbsbconstindex \in \lbsbconstset} || \leq \delta_k$}
                \State $\Sigma_{k+1} \gets \bar{\Sigma}(\boldsymbol{Y}_k, \boldsymbol{R}_k,\boldsymbol{\Sigma}_k,\boldsymbol{s}_k)$
                \State $\epsilon_{k+1} \gets \epsilon_k$
                \State $\omega_{k+1} \gets \omega_k \epsilon_{k+1}^{\beta_{\omega}}$
                \State $\delta_{k+1} \gets \delta_k \epsilon_{k+1}^{\beta_{\delta}}$
            \Else
                \State $\Sigma_{k+1} \gets \Sigma_k$
                \State $\epsilon_{k+1} \gets \tau \epsilon_k$
                \State $\omega_{k+1} \gets \omega_s \epsilon_{k+1}^{\alpha_{\omega}}$
                \State $\delta_{k+1} \gets \delta_s \epsilon_{k+1}^{\alpha_{\delta}}$
            \EndIf
        \Until{$|| P\big( (\boldsymbol{Y}_k,\boldsymbol{R}_k), \nabla_{ \boldsymbol{Y},\boldsymbol{R}} \Psi(\boldsymbol{Y}_k, \boldsymbol{R}_k,\boldsymbol{\Sigma}_k,\boldsymbol{s}_k) \big) || \leq \omega_*$ and $|| \big[ c_\lbsbconstindex(\boldsymbol{Y}_k,\boldsymbol{R}_k) \bar{\sigma}_\lbsbconstindex (\boldsymbol{Y}_k, \boldsymbol{R}_k,\boldsymbol{\Sigma}_k,\boldsymbol{s}_k) \big]_{\lbsbconstindex \in \lbsbconstset} || \leq \delta_*$}
    \end{algorithmic}
\end{algorithm}}

\fullversion{}{\section{Constrained Optimization and Optimality Conditions}\label{append:constrainedopt}
Consider a general optimization problem of the following form:
\begin{subequations}\label{eq:general}
\begin{align}
    \text{Maximize} \quad &f(\boldsymbol{x})\\
    \text{Subj. to} \quad & g_j(\boldsymbol{x}) \geq 0 \quad j = 1, \dots, r \\
    & h_i(\boldsymbol{x}) = 0 \quad i = 1, \dots, m,
\end{align}
\end{subequations}
where $f : \mathbb{R}^n \rightarrow \mathbb{R}$, $g_j: \mathbb{R}^n \rightarrow \mathbb{R}, ~ j = 1, \dots, r$, and $h_i : \mathbb{R}^n \rightarrow \mathbb{R}, ~ i=1,\dots,m$ are continuously differentiable functions. Here we provide a statement of the most common first-order necessary optimality condition known as \emph{KKT} condition. First let us define a regular point:
\begin{defn}\emph{Regular point}:
If the gradient of equality constraints and active inequality constraints are linearly independent at $\boldsymbol{x}^*$, then $\boldsymbol{x}^*$ is called a regular point. 
\end{defn}
Now, let us formally define \emph{Karush-Kuhn-Tucker} (KKT) points which we use extensively throughout this paper and in stating optimality conditions: 
\begin{defn}\label{def:KKT}
A point $\boldsymbol{x}^*\in \mathbb{R}^n$ is called a KKT point for Problem~\eqref{eq:general} if there exist \emph{Lagrangian variables} $\boldsymbol{\nu}^* \in \mathbb{R}^m$ and $\boldsymbol{\mu}^* \in \mathbb{R}^r$, such that:
\begin{subequations}\label{KKTCondition}
\begin{align*}
    & \nabla_{x}L(\boldsymbol{x}^*, \boldsymbol{\mu}^*, \boldsymbol{\nu}^*) = 0 \\
    &h_i(\boldsymbol{x}^*) = 0 \quad \forall i \in \{1, \dots, m\} \\
    &g_j(\boldsymbol{x}^*) \geq 0 \quad \forall j \in \{1, \dots, r\} \\
    &\mu_j^* \geq 0 \quad \forall j \in \{1, \dots, r\} \\
    &\mu_j^* g_j(\boldsymbol{x}^*) = 0 \quad \forall j \in\{1, \dots, r\}.
\end{align*}
\end{subequations}
where $L(\boldsymbol{x}, \boldsymbol{\mu}, \boldsymbol{\nu}) \triangleq f(\boldsymbol{x}) + \sum_i \nu_i h_i(\boldsymbol{x}) + \sum_j \mu_j g_j(\boldsymbol{x})$ is called the \emph{Lagrangian function}.
\end{defn}

\begin{prop} (First-order KKT Necessary Conditions)
Let $\boldsymbol{x}^*$ be a local minimum of Problem~\eqref{eq:general}, and assume $\boldsymbol{x}^*$ is regular. Then $\boldsymbol{x}^*$ is a KKT point.
\end{prop}
Using the second derivatives of the Lagrangian function, we can state the sufficient condition for optimality.
\begin{prop} (Second-order Sufficiency Conditions)
Assume $f$, $g_j, ~\forall j=1, \dots, r$, and $h_i,~\forall i=1, \dots,m$ are twice continuously differentiable, and $\boldsymbol{x}^*$ is a KKT point with corresponding Lagrange variables $\boldsymbol{\mu}^*$ and $\boldsymbol{\nu}^*$. In addition let 
\begin{equation*}
    \boldsymbol{V}^T \nabla^2_{xx}L(\boldsymbol{x}^*, \boldsymbol{\mu}^*, \boldsymbol{\nu}^*) \boldsymbol{V} <0,
\end{equation*}
for all $\boldsymbol{V} \neq 0$ such that
\begin{align*}
    \nabla h_i(\boldsymbol{x}^*)^T \boldsymbol{V} = 0, ~ \forall i = 1, \dots, m, \\
    \nabla g_j(\boldsymbol{x}^*)^T \boldsymbol{V} = 0, ~ \forall j \in A(\boldsymbol{x}^*).
\end{align*}
In addition assume we have strict complementary slackness, i.e.,
\begin{equation*}
    \mu_j > 0 \quad \forall j \in A(\boldsymbol{x}^*)
\end{equation*}
where $A(\boldsymbol{x}^*)$ is set of active inequality constraints in $\boldsymbol{x}^*$, i.e.,
\begin{equation*}
    A(\boldsymbol{x}^*) \triangleq \{j~|~g_j(\boldsymbol{x}^*) = 0 \}
\end{equation*}
Then $\boldsymbol{x}^*$ is a strict local maximum of Problem~\eqref{eq:general}.
\end{prop}

For further information on other forms of necessary and sufficient conditions for optimality, refer to the book by Bertsekas \cite{bertsekas1999nonlinear}.}

\section{Proof of Lemma~\ref{lem:kkt}} \label{append:proofOflemKKT}
Here we show that the regularity of a point $(\hat{\boldsymbol{Y}},\hat{\boldsymbol{R}})$ is equivalent to the assumption stated in Theorem 4.4 of Conn et. al. \cite{conn1997globally}. We divide the variables $(\hat{\boldsymbol{Y}}, \hat{\boldsymbol{R}})$ into two distinct class $\mathcal{F}$ and $\mathcal{F}'$, such that the variables in $\mathcal{F}'$ are equal to their upper or lower bound and the variables in $\mathcal{F}$ are strictly between their upper or lower bound. As a result, we have exactly $|\mathcal{F}'|$ active bounding constraints. We denote by $\mathcal{A}$ the set of active link and cache constraints \eqref{const:link}, \eqref{const:cache}. Thus, we can decompose the Jacobian matrix for the active constraints at point $(\hat{\boldsymbol{Y}}, \hat{\boldsymbol{R}})$ as $
J = \left[\begin{smallmatrix}
J_{[\mathcal{A} ~ \mathcal{F}]} & J_{[\mathcal{A} ~ \mathcal{F}']}\\
\boldsymbol{0}_{|\mathcal{F}'| \times |\mathcal{F}|} & Q_{|\mathcal{F}'| \times |\mathcal{F}'|}
\end{smallmatrix}\right]$.
We denote by $M_{[s_1 ~ s_2]}$ is a sub-matrix of the matrix $M$ where rows are picked according to set $s_1$ and columns are picked according to set $s_2$. The last $|\mathcal{F}'|$ rows correspond to active box constraints \eqref{const:boxfory}, \eqref{const:boxforlambda}. It can be easily seen that $Q_{|\mathcal{F}'| \times |\mathcal{F}'|}$ is a diagonal matrix with elements of diagonal being $+1$ or $-1$ depending on whether the variable is at its upper bound or lower bound. If $(\hat{\boldsymbol{Y}}, \hat{\boldsymbol{R}})$ is a regular point, then $J$ is full rank at $(\hat{\boldsymbol{Y}}, \hat{\boldsymbol{R}})$. Hence, it can be seen that $J_{[\mathcal{A} ~ \mathcal{F}]}$ is also full rank, which is exactly the assumption in Theorem 4.4 of Conn et. al. \cite{conn1997globally}. \qed

\fullversion{\section{Proof of Lemma~\ref{lem:diff}} \label{append:ProofOfLemDiff}
By definition, the KKT necessary conditions for optimality hold at $(\hat{\boldsymbol{Y}},\hat{\boldsymbol{R}})$. Hence, there exist Lagrange multipliers $[\hat{\mu}_{ba}]_{(b,a) \in \mathcal{E}}$ associated with \eqref{const:link}, $[\hat{\gamma}_v]_{v \in \mathcal{V}}$ associated with \eqref{const:cache}, $[\hat{\xi}_{vi}]_{v \in \mathcal{V}, i \in \itemcat}$, $[\hat{\xi}_{vi}']_{v \in \mathcal{V}, i \in \itemcat}$ associated with \eqref{const:boxfory}, $[\hat{\eta}_{(i,p)}]_{(i,p) \in \requestset}$,  $[\hat{\eta}_{(i,p)}']_{(i,p) \in \requestset}$ associated with \eqref{const:boxforlambda}. 
Due to the concavity of $F$, we can write
$F(\hat{\boldsymbol{R}}) \geq F(\boldsymbol{R}^*) - \nabla _{\boldsymbol{Y}, \boldsymbol{R}} F(\hat{\boldsymbol{R}})^T
\big[(\boldsymbol{Y}^*,\boldsymbol{R}^*)- (\hat{\boldsymbol{Y}},\hat{\boldsymbol{R}}) \big]$. After applying the first order necessary condition and complementary slackness, we can write
$F(\hat{\boldsymbol{R}}) \geq  F(\boldsymbol{R}^*) + \sum_{(b,a) \in \mathcal{E}} \hat{\mu}_{ba} \nabla_{\boldsymbol{Y},\boldsymbol{R}} g_{ba}(\hat{\boldsymbol{Y}}, \hat{\boldsymbol{R}})^T \big[(\boldsymbol{Y}^*,\boldsymbol{R}^*)- (\hat{\boldsymbol{Y}},\hat{\boldsymbol{R}}) \big]$. As stated in Lemma \ref{lem:monotonedr}, functions $g_{ba}(\cdot)$, for all $(b,a) \in \mathcal{E}$ are monotone DR-submodular. Therefore, we can use the following Lemma from Bian et al. \cite{bian2017dr}:
\begin{lem}[Bian et al. \cite{bian2017dr}] \label{lem:drlemproof}
For any differentiable DR-submodular function $f : \mathcal{X} \rightarrow \mathbb{R}$ and any two points $\boldsymbol{a}, \boldsymbol{b}$ in $\mathcal{X}$, we have
\begin{equation*}
    (\boldsymbol{b}-\boldsymbol{a})^T \nabla f(\boldsymbol{a}) \geq f(\boldsymbol{a} \vee \boldsymbol{b}) + f( \boldsymbol{a} \wedge \boldsymbol{b}) - 2f(\boldsymbol{a}),
\end{equation*}
where $\vee$ and $\wedge$ are coordinate-wise maximum and minimum operations, respectively.
\end{lem}
By Lemma~\ref{lem:drlemproof}, we have $F(\hat{\boldsymbol{R}}) \geq F(\boldsymbol{R}^*) + \sum_{(b,a) \in \mathcal{E}} \hat{\mu}_{ba} \big(g_{ba} (\hat{\boldsymbol{Y}} \vee \boldsymbol{Y}^*, \hat{\boldsymbol{R}} \vee \boldsymbol{R}^* ) +  g_{ba}(\hat{\boldsymbol{Y}} \wedge \boldsymbol{Y}^*, \hat{\boldsymbol{R}} \wedge \boldsymbol{R}^* ) - 2g_{ba}(\hat{\boldsymbol{Y}}, \hat{\boldsymbol{R}} ) \big)$.
By Lemma \ref{lem:monotonedr}, $g_{ba}(\cdot)$ are monotone and we know that they are positive for all $(b,a) \in \mathcal{E}$. Hence, $g_{ba} (\hat{\boldsymbol{Y}} \vee \boldsymbol{Y}^*, \hat{\boldsymbol{R}} \vee \boldsymbol{R}^* ) \geq g_{ba} (\boldsymbol{Y}^*, \boldsymbol{R}^* )$ and $g_{ba}(\hat{\boldsymbol{Y}} \wedge \boldsymbol{Y}^*, \hat{\boldsymbol{R}} \wedge \boldsymbol{R}^* ) \geq 0$ for all $(b,a) \in \mathcal{E}$. Therefore, $F(\hat{\boldsymbol{R}}) \geq F(\boldsymbol{R}^*) + \sum_{(b,a) \in \mathcal{E}} \hat{\mu}_{ba} \big(g_{ba} (\boldsymbol{Y}^*, \boldsymbol{R}^* ) - 2g_{ba}(\hat{\boldsymbol{Y}}, \hat{\boldsymbol{R}} ) \big) \overset{(**)}{\geq} F(\boldsymbol{R}^*) - \sum_{(b,a)}\hat{\mu}_{ba} (\sum_{(i,p) : (a,b) \in p} \bar{\lambda}_{(i,p)} - C_{ba}),$ where $(**)$ is due to complementary slackness. \qed}{\section{Proof of Lemma~\ref{lem:diff}} \label{append:ProofOfLemDiff}
Since $(\hat{\boldsymbol{Y}},\hat{\boldsymbol{R}})$ is a KKT point, KKT necessary conditions for optimality hold at $(\hat{\boldsymbol{Y}},\hat{\boldsymbol{R}})$. Hence, there exists Lagrange multipliers $[\hat{\mu}_{ba}]_{(b,a) \in \mathcal{E}}$ associated with \eqref{const:link}, $[\hat{\gamma}_v]_{v \in \mathcal{V}}$ associated with \eqref{const:cache}, $[\hat{\xi}_{vi}]_{v \in \mathcal{V}, i \in \itemcat}$, $[\hat{\xi}_{vi}']_{v \in \mathcal{V}, i \in \itemcat}$ associated with \eqref{const:boxfory}, $[\hat{\eta}_{(i,p)}]_{(i,p) \in \requestset}$,  $[\hat{\eta}_{(i,p)}']_{(i,p) \in \requestset}$ associated with \eqref{const:boxforlambda}. To be concise and avoid repetitions, we put our variables and Lagrange multipliers corresponding to box constraints \eqref{const:boxfory}, \eqref{const:boxforlambda} in 1-column vectors as follows:
\begin{align*}
&\begin{bmatrix}
\hat{\boldsymbol{Y}}\\
\hat{\boldsymbol{R}}
\end{bmatrix}_{(|\mathcal{V}||\itemcat| + |\requestset|) \times 1}, \quad
\begin{bmatrix}
\boldsymbol{Y}^*\\
\boldsymbol{R}^*
\end{bmatrix}_{(|\mathcal{V}||\itemcat| + |\requestset|) \times 1}, \\
&\hat{\boldsymbol{\eta}} \triangleq \begin{bmatrix}
\boldsymbol{0} _{|\mathcal{V}||\itemcat| \times 1} \\
[\hat{\eta}_{\requestindex}]_{\requestindex \in \requestset}
\end{bmatrix}, \quad
\hat{\boldsymbol{\eta}}' \triangleq \begin{bmatrix}
\boldsymbol{0} _{|\mathcal{V}||\itemcat| \times 1} \\
[\hat{\eta}_{\requestindex}']_{\requestindex \in \requestset}
\end{bmatrix}, \\
&\hat{\boldsymbol{\xi}} \triangleq \begin{bmatrix}
[\hat{\xi}_{vi}] _{v \in \mathcal{V}, i \in \itemcat} \\
\boldsymbol{0}_{|\requestset| \times 1}
\end{bmatrix}, \quad
\hat{\boldsymbol{\xi}}' \triangleq \begin{bmatrix}
[\hat{\xi}_{vi}'] _{v \in \mathcal{V}, i \in \itemcat} \\
\boldsymbol{0}_{|\requestset| \times 1}
\end{bmatrix}. 
\end{align*}
Here the first $|\mathcal{V}| |\itemcat|$ columns correspond to the variables $\hat{\boldsymbol{Y}}$ and the next $|\requestset|$ correspond to the variables $\hat{\boldsymbol{R}}$. Now we write the KKT conditions as:
\begin{subequations}
\begin{align}
    & \nabla_{\boldsymbol{Y},\boldsymbol{R}} F(\boldsymbol{R}) +\sum_{(b,a) \in \mathcal{E}}{\hat{\mu}_{ba} \nabla_{\boldsymbol{Y},\boldsymbol{R}} g_{ba}(\hat{\boldsymbol{Y}}, \hat{\boldsymbol{R}})} - \sum_{v \in \mathcal{V}} \hat{\gamma}_v \nabla_{\boldsymbol{Y},\boldsymbol{R}} g_{v}(\hat{\boldsymbol{Y}}) \nonumber \\
    & + \hat{\boldsymbol{\eta}} - \hat{\boldsymbol{\eta}'}+ \hat{\boldsymbol{\xi}} - \hat{\boldsymbol{\xi}'}  = \boldsymbol{0} \label{eq:gradxr}\\
    &\hat{\mu}_{ba} \big(g_{ba}(\hat{\boldsymbol{Y}}, \hat{\boldsymbol{R}} ) -  (\sum_{(i,p) : (a,b) \in p} \bar{\lambda}_{(i,p)} - C_{ba}) \big) = 0, \quad \hat{\mu}_{ba} \geq 0\label{eq:complslack1} \\
    &\hat{\gamma}_v (g_v(\hat{\boldsymbol{Y}}) - c_v) = 0 , \quad \hat{\gamma}_v \geq 0, \quad \forall v \in \mathcal{V} \label{eq:complslack2}\\
    &\hat{\xi}_{vi} \hat{y}_{vi} = 0, \quad \hat{\xi}_{vi} \geq 0 \quad \forall v \in \mathcal{V}, \forall i \in \itemcat \label{eq:complslack3}\\
    &\hat{\xi}_{vi}' (1-\hat{y}_{vi}) = 0, \quad \hat{\xi}_{vi}' \geq 0, \quad \forall v \in \mathcal{V}, \forall i \in \itemcat \label{eq:complslack4}\\
    &\hat{\eta}_{\requestindex} \hat{r}_{\requestindex} = 0, \quad \hat{\eta}_\requestindex \geq 0, \quad \forall n \in \requestset \label{eq:complslack5}\\
    &\hat{\eta}_{\requestindex}' (\bar{\lambda}_{\requestindex}-\hat{r}_{\requestindex}) = 0, \quad \hat{\eta}_{\requestindex}' \geq 0, \quad \forall n \in \requestset. \label{eq:complslack6}
\end{align}
\end{subequations}
Due to the concavity of $F$, we can write
\begin{align*}
& \nabla _{\boldsymbol{Y}, \boldsymbol{R}} F(\hat{\boldsymbol{R}})^T
\begin{bmatrix}
      \boldsymbol{Y}^* - \hat{\boldsymbol{Y}} \\ \boldsymbol{R}^* - \hat{\boldsymbol{R}}
    \end{bmatrix} \\
 &= \nabla_{\boldsymbol{R}} F(\hat{\boldsymbol{R}})^T(\boldsymbol{R}^*-\hat{\boldsymbol{R}}) \geq F(\boldsymbol{R}^*) - F(\hat{\boldsymbol{R}}). 
\end{align*}
Hence,
\begin{align*}
F(\hat{\boldsymbol{R}}) &\geq F(\boldsymbol{R}^*) - \nabla _{\boldsymbol{Y}, \boldsymbol{R}} F(\hat{\boldsymbol{R}})^T
\begin{bmatrix}
      \boldsymbol{Y}^* - \hat{\boldsymbol{Y}} \\ \boldsymbol{R}^* - \hat{\boldsymbol{R}}
    \end{bmatrix} \\
& \overset{\eqref{eq:gradxr}}{=} F(\boldsymbol{R}^*) + \sum_{(b,a) \in \mathcal{E}} \hat{\mu}_{ba} \nabla_{\boldsymbol{Y},\boldsymbol{R}} g_{ba}(\hat{\boldsymbol{Y}}, \hat{\boldsymbol{R}})^T
\begin{bmatrix}
    \boldsymbol{Y}^* - \hat{\boldsymbol{Y}} \\ \boldsymbol{R}^* - \hat{\boldsymbol{R}}
\end{bmatrix} \\
 & -\sum_{v \in \mathcal{V}} \hat{\gamma}_v \nabla_{\boldsymbol{Y},\boldsymbol{R}} g_{v}(\hat{\boldsymbol{Y}})^T
 \begin{bmatrix}
    \boldsymbol{Y}^* - \hat{\boldsymbol{Y}} \\ \boldsymbol{R}^* - \hat{\boldsymbol{R}}
\end{bmatrix} \\
& + \hat{\boldsymbol{\xi}}\begin{bmatrix}
      \boldsymbol{Y}^* - \hat{\boldsymbol{Y}} \\ \boldsymbol{R}^* - \hat{\boldsymbol{R}}
    \end{bmatrix} - \hat{\boldsymbol{\xi}'}\begin{bmatrix}
      \boldsymbol{Y}^* - \hat{\boldsymbol{Y}} \\ \boldsymbol{R}^* - \hat{\boldsymbol{R}}
    \end{bmatrix} + \hat{\boldsymbol{\eta}}\begin{bmatrix}
      \boldsymbol{Y}^* - \hat{\boldsymbol{Y}} \\ \boldsymbol{R}^* - \hat{\boldsymbol{R}}
    \end{bmatrix} \\
&- \hat{\boldsymbol{\eta}'}\begin{bmatrix}
      \boldsymbol{Y}^* - \hat{\boldsymbol{Y}} \\ \boldsymbol{R}^* - \hat{\boldsymbol{R}}
    \end{bmatrix}. 
\end{align*}

Due to Lemma \ref{lem:monotonedr}, functions $g_{ba}(\cdot)$, for all $(b,a) \in \mathcal{E}$ are monotone DR-submodular. Therefore, we can use the following Lemma from Bian et. al \cite{bian2017dr}:
\begin{lem} \label{lem:drlemproof}
For any differntiable DR-submodular function $f : \mathcal{X} \rightarrow \mathbb{R}$ and any two points $\boldsymbol{a}, \boldsymbol{b}$ in $\mathcal{X}$, we have
\begin{equation*}
    (\boldsymbol{b}-\boldsymbol{a})^T \nabla f(\boldsymbol{a}) \geq f(\boldsymbol{a} \vee \boldsymbol{b}) + f( \boldsymbol{a} \wedge \boldsymbol{b}) - 2f(\boldsymbol{a}),
\end{equation*}
where $\vee$ and $\wedge$ are coordinate-wise maximum and minimum operations, respectively.
\end{lem}
Hence, we can write
\begin{align*}
    F(\hat{\boldsymbol{R}}) & \overset{(\text{Lemma~\ref{lem:drlemproof}})}{\geq} F(\boldsymbol{R}^*) + \sum_{(b,a) \in \mathcal{E}} \hat{\mu}_{ba} \big(g_{ba} (\hat{\boldsymbol{Y}} \vee \boldsymbol{Y}^*, \hat{\boldsymbol{R}} \vee \boldsymbol{R}^* ) + \\
& g_{ba}(\hat{\boldsymbol{Y}} \wedge \boldsymbol{Y}^*, \hat{\boldsymbol{R}} \wedge \boldsymbol{R}^* ) - 2g_{ba}(\hat{\boldsymbol{Y}}, \hat{\boldsymbol{R}} ) \big) \\
& -\sum_{v \in \mathcal{V}} \hat{\gamma}_v \nabla_{\boldsymbol{Y},\boldsymbol{R}} g_{v}(\hat{\boldsymbol{Y}})^T
    \begin{bmatrix}
      \boldsymbol{Y}^* - \hat{\boldsymbol{Y}} \\ \boldsymbol{R}^* - \hat{\boldsymbol{R}}
    \end{bmatrix} + \hat{\boldsymbol{\xi}}
    \begin{bmatrix}
      \boldsymbol{Y}^* - \hat{\boldsymbol{Y}} \\ \boldsymbol{R}^* - \hat{\boldsymbol{R}}
    \end{bmatrix} 
    \\
&- \hat{\boldsymbol{\xi}'}
    \begin{bmatrix}
      \boldsymbol{Y}^* - \hat{\boldsymbol{Y}} \\ \boldsymbol{R}^* - \hat{\boldsymbol{R}}
    \end{bmatrix}
 + \hat{\boldsymbol{\eta}}\begin{bmatrix}
      \boldsymbol{Y}^* - \hat{\boldsymbol{Y}} \\ \boldsymbol{R}^* - \hat{\boldsymbol{R}}
    \end{bmatrix} - \hat{\boldsymbol{\eta}'}\begin{bmatrix}
      \boldsymbol{Y}^* - \hat{\boldsymbol{Y}} \\ \boldsymbol{R}^* - \hat{\boldsymbol{R}}
    \end{bmatrix}
\end{align*}
By Lemma \ref{lem:monotonedr}, $g_{ba}(\cdot)$ are monotone and we know that they are all positive of for all $(b,a) \in \mathcal{E}$. Hence, $g_{ba} (\hat{\boldsymbol{Y}} \vee \boldsymbol{Y}^*, \hat{\boldsymbol{R}} \vee \boldsymbol{R}^* ) \geq g_{ba} (\boldsymbol{Y}^*, \boldsymbol{R}^* )$ and $g_{ba}(\hat{\boldsymbol{Y}} \wedge \boldsymbol{Y}^*, \hat{\boldsymbol{R}} \wedge \boldsymbol{R}^* ) \geq 0$ for all $(b,a) \in \mathcal{E}$. Therefore,
\begin{align*}
\small
F(\hat{\boldsymbol{R}}) \geq &F(\boldsymbol{R}^*) + \sum_{(b,a) \in \mathcal{E}} \hat{\mu}_{ba} \big(g_{ba} (\boldsymbol{Y}^*, \boldsymbol{R}^* ) - 2g_{ba}(\hat{\boldsymbol{Y}}, \hat{\boldsymbol{R}} ) \big) \\
& -\sum_{v \in \mathcal{V}} \hat{\gamma}_v \nabla_{\boldsymbol{Y},\boldsymbol{R}} g_{v}(\hat{\boldsymbol{Y}})^T
    \begin{bmatrix}
      \boldsymbol{Y}^* - \hat{\boldsymbol{Y}} \\ \boldsymbol{R}^* - \hat{\boldsymbol{R}}
    \end{bmatrix} + \hat{\boldsymbol{\xi}}\begin{bmatrix}
      \boldsymbol{Y}^* - \hat{\boldsymbol{Y}} \\ \boldsymbol{R}^* - \hat{\boldsymbol{R}}
    \end{bmatrix} \\
    &- \hat{\boldsymbol{\xi}'}\begin{bmatrix}
      \boldsymbol{Y}^* - \hat{\boldsymbol{Y}} \\ \boldsymbol{R}^* - \hat{\boldsymbol{R}}
    \end{bmatrix} + \hat{\boldsymbol{\eta}}\begin{bmatrix}
      \boldsymbol{Y}^* - \hat{\boldsymbol{Y}} \\ \boldsymbol{R}^* - \hat{\boldsymbol{R}}
    \end{bmatrix} - \hat{\boldsymbol{\eta}'}\begin{bmatrix}
      \boldsymbol{Y}^* - \hat{\boldsymbol{Y}} \\ \boldsymbol{R}^* - \hat{\boldsymbol{R}}
    \end{bmatrix}
\end{align*}
\begin{align*}
&\overset{(**)}{=} F(\boldsymbol{R}^*) \\
& + \sum_{(b,a) \in \mathcal{E}} \hat{\mu}_{ba} \big(g_{ba} (\boldsymbol{Y}^*, \boldsymbol{R}^* ) - 2(\sum_{(i,p) : (a,b) \in p} \bar{\lambda}_{(i,p)} - C_{ba}) \big)\\
& - \sum_{v \in \mathcal{V}} \hat{\gamma}_v( \sum_{i \in \itemcat} y_{vi}^* - \sum_{i \in \itemcat} \hat{y}_{vi} ) + \sum_{v \in \mathcal{V},i \in \itemcat} \hat{\xi}_{vi} y_{vi}^* \\
&- \sum_{v \in \mathcal{V},i \in \itemcat} \hat{\xi'}_{vi}
      (y_{vi}^* - 1) + \sum_{\requestindex \in \requestset} \hat{\eta}_{\requestindex} r_{\requestindex}^*  - \sum_{\requestindex \in \requestset} \hat{\eta'}_{\requestindex} (r_{\requestindex}^* - \bar{\lambda}_{\requestindex}) \\
 & \geq F(\boldsymbol{R}^*) - \sum_{(b,a)}\hat{\mu}_{ba} (\sum_{(i,p) : (a,b) \in p} \bar{\lambda}_{(i,p)} - C_{ba}),
 \end{align*}
 where $(**)$ is due to \eqref{eq:complslack1},\eqref{eq:complslack2},\eqref{eq:complslack3}, \eqref{eq:complslack4}, \eqref{eq:complslack5}, and \eqref{eq:complslack6}. Therefore, we have
 \begin{equation*} 
F(\hat{\boldsymbol{R}}) \geq F(\boldsymbol{R}^*) - \sum_{(b,a)}\hat{\mu}_{ba} (\sum_{(i,p) : (a,b) \in p} \bar{\lambda}_{(i,p)} - C_{ba}).
 \end{equation*}}

\section{Proof of Lemma~\ref{lem:boundOnMu}} \label{append:ProofOfLemBoundOnMu}
\fullversion{We call a link $(b,a)$ an \emph{active link} if $\sum_{(i,p) : (a,b) \in p} \lambda_{(i,p)} \prod_{v=p_1}^{a} (1-y_{vi}) = C_{ba}$.}{First we define an \emph{active link} as follows.
\begin{defn}
A link for which the constraint \eqref{const:link} is satisfied exactly is called an \emph{active link}. In other words, for an active link $(b,a)$ we have:
\begin{equation*}
    \sum_{(i,p) : (a,b) \in p} \lambda_{(i,p)} \prod_{v=p_1}^{a} (1-y_{vi}) = C_{ba}
\end{equation*}
\end{defn}}
\fullversion{Since $C_{ba}$ is positive, for an active link $(b,a)$ there exists a set
\begin{align*} \small
     \hat{\requestset}_{(b,a)}^{\text{\scalebox{0.8}{active}}} \triangleq &\big\{ (i,p) : (a,b) \in p,~ (\bar{\lambda}_{(i,p)} - \hat{r}_{(i,p)}) \prod_{v=p_1}^{a} (1-\hat{y}_{vi}) > 0, \nonumber \\
     &\sum_{(i,p) \in \hat{\requestset}_{(b,a)}^{\text{\scalebox{0.8}{active}}}} (\bar{\lambda}_{(i,p)} - \hat{r}_{(i,p)}) \prod_{v=p_1}^{a} (1-\hat{y}_{vi}) = C_{ba}\big\} 
\end{align*} }{\begin{lem}\label{lem:Ractive}
For an active link $(b,a)$ there exists a set $\hat{\requestset}_{(b,a)}^{active} \subseteq \requestset$ such that:
\begin{align*} 
     &\hat{\requestset}_{(b,a)}^{active} = \\
     &\{ (i,p) : (a,b) \in p,~ (\bar{\lambda}_{(i,p)} - \hat{r}_{(i,p)}) \prod_{v=p_1}^{a} (1-\hat{y}_{vi}) > 0, \nonumber \\
     &\sum_{(i,p) \in \hat{\requestset}_{(b,a)}^{active}} (\bar{\lambda}_{(i,p)} - \hat{r}_{(i,p)}) \prod_{v=p_1}^{a} (1-\hat{y}_{vi}) = C_{ba}\} 
\end{align*}
\end{lem}
The proof of Lemma~\ref{lem:Ractive} is followed by the definition of \emph{active link} and $C_{ba}$ being positive.} Suppose $(b,a)$ is an active link. By \fullversion{definition}{Lemma~\ref{lem:Ractive}}, we have $\hat{r}_{(i,p)} < \bar{\lambda}_{(i,p)},~\forall (i,p) \in \hat{\requestset}_{(b,a)}^{\text{\scalebox{0.8}{active}}}$. By writing the KKT conditions with respect to $\hat{r}_{(i,p)}$ for $(i,p) \in \hat{\requestset}_{(b,a)}^{\text{\scalebox{0.8}{active}}}$, we have
$\frac{dU_{(i,p)}(\bar{\lambda}_{(i,p)} - \hat{r}_{(i,p)})}{d \lambda} = \sum_{(c,d): (c,d) \in p} \hat{\mu}_{dc} \prod_{v=p_1}^{c} (1-\hat{y}_{vi}) + \hat{\eta}_{(i,p)}$. This implies
$\frac{dU_{(i,p)}(\bar{\lambda}_{(i,p)} - \hat{r}_{(i,p)})}{d \lambda}  \geq \hat{\mu}_{ba}\prod_{v=p_1}^{a} (1-\hat{y}_{vi})$.
After multiplying both sides by $(\bar{\lambda}_{(i,p)} - \hat{r}_{(i,p)})$, \fullversion{and using Assumption~\ref{assm:logreturn} and the fact that $\theta$ is the maximum among logarithmic diminishing return parameters, we can write $\theta \geq \hat{\mu}_{ba}(\bar{\lambda}_{(i,p)} - \hat{r}_{(i,p)})\prod_{v=p_1}^{a} (1-\hat{y}_{vi})$. By summing over all $(i,p) \in \hat{\requestset}_{(b,a)}^{\scalebox{0.8}{\text{active}}}$, we have $\theta | \hat{\requestset}_{(b,a)}^{\scalebox{0.8}{\text{active}}} | \geq  \hat{\mu}_{ba} C_{ba}$, or equivalently $\hat{\mu}_{ba} \leq \theta \frac{| \hat{\requestset}_{(b,a)}^{active} |}{C_{ba}}$.}{we have
\begin{align} \small
    & (\bar{\lambda}_{(i,p)} - \hat{r}_{(i,p)}) \frac{dU_{(i,p)}(\lambda)}{d \lambda} \bigg |_{\bar{\lambda}_{(i,p)} - \hat{r}_{(i,p)}} \nonumber  \\ 
    \geq &\hat{\mu}_{ba}(\bar{\lambda}_{(i,p)} - \hat{r}_{(i,p)})\prod_{v=p_1}^{a} (1-\hat{y}_{vi}).  \label{equ:multipliedreadytouse}
\end{align}
Using Assumption~\ref{assm:logreturn} and the fact that $\theta$ is the maximum among logarithmic diminishing return parameters, we can write \eqref{equ:multipliedreadytouse} as
\begin{equation} \label{equ:logreturnuse}
    \theta \geq \hat{\mu}_{ba}(\bar{\lambda}_{(i,p)} - \hat{r}_{(i,p)})\prod_{v=p_1}^{a} (1-\hat{y}_{vi}).
\end{equation}
By summing \eqref{equ:logreturnuse} over all $(i,p) \in \hat{\requestset}_{(b,a)}^{\scalebox{0.8}{\text{active}}}$, we have $\theta | \hat{\requestset}_{(b,a)}^{\scalebox{0.8}{\text{active}}} | \geq  \hat{\mu}_{ba} C_{ba} \Rightarrow \hat{\mu}_{ba} \leq \theta \frac{| \hat{\requestset}_{(b,a)}^{active} |}{C_{ba}}$.} If $(b,a)$ is not an active link, $\hat{\mu}_{ba} = 0$. As a result, we have $\sum_{(b,a) \in \mathcal{E}}\hat{\mu}_{ba} (\sum_{(i,p) : (a,b) \in p} \bar{\lambda}_{(i,p)} - C_{ba}) \leq \theta \sum_{(b,a) \in \mathcal{E}} n_{ab} \big(\sum_{(i,p) : (a,b) \in p} \bar{\lambda}_{(i,p)} - C_{ba}\big)/C_{ba},$
where the last inequality is due to  $| \hat{\requestset}_{(b,a)}^{\scalebox{0.8}{active}} | \leq n_{ab}$.       \qed


\fullversion{}{\section{Proof of Proposition~\ref{prop:k0}} \label{append:ProofOfPropk0}
We show that if the assumptions of Proposition~\ref{prop:k0} hold, then assumptions of part (ii) of
Theorem 5.3 and Corollary 5.7 of Conn et. al. \cite{conn1997globally} hold. 
\begin{lem} \label{lem:ass5ofconn}
Suppose $(\hat{\boldsymbol{Y}}, \hat{\boldsymbol{R}})$ is regular and satisfies the second-order sufficiency condition. Then Assumption 5 of Conn et. al.  \cite{conn1997globally} hold.
\end{lem}
\begin{proof}
Let us decompose the Jacobian matrix for the active constraints at point $(\hat{\boldsymbol{Y}},\hat{\boldsymbol{R}})$ similar to the way mentioned in Appendix \eqref{append:proofOflemKKT}:
\begin{align*}
J = \begin{bmatrix}
J_{[\mathcal{A} ~ \mathcal{F}]} & J_{[\mathcal{A} ~ \mathcal{F}']}\\
\boldsymbol{0}_{|\mathcal{F}'| \times |\mathcal{F}|} & Q_{|\mathcal{F}'| \times |\mathcal{F}'|}
\end{bmatrix}.
\end{align*}
Similarly, we denote by $M_{[s_1 ~ s_2]}$ is a sub-matrix of matrix $M$ where rows are picked according to set $s_1$ and columns are picked according to set $s_2$. The sets $\mathcal{A}$, $\mathcal{F}$, $\mathcal{F}'$, and matrix $Q$ are introduced in Appendix~\ref{append:proofOflemKKT}. The Lagrangian function is written as
\begin{align*}
&L(\boldsymbol{Y}, \boldsymbol{R}, \boldsymbol{\mu}, \boldsymbol{\gamma}, \boldsymbol{\eta}, \boldsymbol{\xi}) \triangleq \\
&F(\boldsymbol{R}) + \sum_{(b,a) \in \mathcal{E}}{\mu_{ba} g_{ba}(\boldsymbol{Y}, \boldsymbol{R})} - \sum_{v} \gamma_v g_{v \in \mathcal{V}}(\boldsymbol{Y})  \\
& + \sum_{\requestindex \in \requestset} \eta_{\requestindex} r_{\requestindex} - \sum_{\requestindex \in \requestset} \eta'_{\requestindex} (r_{\requestindex} - \bar{
\lambda}_{\requestindex})  \\
& + \sum_{v \in \mathcal{V},i \in \itemcat} \xi_{vi}y_{vi} - \sum_{v \in \mathcal{V},i \in \itemcat} \xi'_{vi} (y_{vi} - 1)    
\end{align*}

Due to the second order sufficiency conditions, we have
\begin{equation} \label{equ:vnablav}
    \boldsymbol{V}^T \nabla^2_{xx}L(\hat{\boldsymbol{Y}}, \hat{\boldsymbol{R}}, \hat{\boldsymbol{\mu}}, \hat{\boldsymbol{\gamma}}, \hat{\boldsymbol{\eta}}, \hat{\boldsymbol{\xi}}) \boldsymbol{V} < 0,
\end{equation}
for all $\boldsymbol{V} \neq 0$ such that
\begin{align} \label{equ:jacobv}
    J \times \boldsymbol{V} = \boldsymbol{0}.
\end{align}
We decompose $\boldsymbol{V}$ and $\nabla^2_{xx}L(\hat{\boldsymbol{Y}}, \hat{\boldsymbol{R}}, \hat{\boldsymbol{\mu}}, \hat{\boldsymbol{\gamma}}, \hat{\boldsymbol{\eta}}, \hat{\boldsymbol{\xi}})$ into variables corresponding to class $\mathcal{F}$ and $\mathcal{F}'$:
\begin{align*}
\boldsymbol{V} = \begin{bmatrix}
\boldsymbol{V}_\mathcal{F} \\
\boldsymbol{V}_{\mathcal{F}'}
\end{bmatrix}, \quad \nabla^2_{xx}L(\hat{\boldsymbol{Y}}, \hat{\boldsymbol{R}}, \hat{\boldsymbol{\mu}}, \hat{\boldsymbol{\gamma}}, \hat{\boldsymbol{\eta}}, \hat{\boldsymbol{\xi}}) = \begin{bmatrix}
H_{[\mathcal{F} ~ \mathcal{F}]} & H_{[\mathcal{F} ~ \mathcal{F}']}\\
H_{[\mathcal{F}' ~ \mathcal{F}]} & H_{[\mathcal{F}' ~ \mathcal{F}']}
\end{bmatrix}
\end{align*}
For all $\boldsymbol{V} \neq 0$ that satisfies \eqref{equ:jacobv} we can write
\begin{align} \label{equ:vfprimezero}
&\begin{bmatrix}
J_{\mathcal{A} ~ \mathcal{F}]} & J_{[\mathcal{A} ~ \mathcal{F}']}\\
\boldsymbol{0}_{|\mathcal{F}'| \times |\mathcal{F}|} & Q_{|\mathcal{F}'| \times |\mathcal{F}'|}
\end{bmatrix} \times \begin{bmatrix}
\boldsymbol{V}_{\mathcal{F}} \\
\boldsymbol{V}_{\mathcal{F}'}
\end{bmatrix} = \boldsymbol{0} \Rightarrow \boldsymbol{V}_{\mathcal{F}'} = \boldsymbol{0}.
\end{align}
Based on \eqref{equ:vnablav}, \eqref{equ:jacobv}, and \eqref{equ:vfprimezero} we can write 
\begin{align} \label{equ:vfhffvf}
&\boldsymbol{V}_{\mathcal{F}}^T H_{[\mathcal{F} ~ \mathcal{F}]} \boldsymbol{V}_\mathcal{F} < 0,
\end{align}
for all $V_\mathcal{F} \neq 0$ such that
\begin{equation*}
    J_{[\mathcal{A} ~ \mathcal{F}]} V_\mathcal{F} = 0.
\end{equation*}

Now consider the following matrix:
\begin{align} \label{equ:matrixone}
\begin{bmatrix}
H_{[\mathcal{F} ~ \mathcal{F}]} & J_{[\mathcal{A} ~ \mathcal{F}]}^T\\
J_{[\mathcal{A} ~ \mathcal{F}]} & \boldsymbol{0}
\end{bmatrix}
\end{align}
We claim that matrix defined in \eqref{equ:matrixone} is non-singular. Suppose 
\begin{equation*}
 \begin{bmatrix}
H_{[\mathcal{F} ~ \mathcal{F}]} & J_{[\mathcal{A} ~ \mathcal{F}]}^T\\
J_{[\mathcal{A} ~ \mathcal{F}]} & \boldsymbol{0}
\end{bmatrix} \times 
\begin{bmatrix}
X_\mathcal{F}\\
X_\mathcal{A} 
\end{bmatrix} = \boldsymbol{0}.
\end{equation*}
Then,
\begin{align}
& J_{[\mathcal{A} ~ \mathcal{F}]} X_\mathcal{F} = 0, \nonumber \\
& H_{[\mathcal{F} ~ \mathcal{F}]} X_\mathcal{F} + J_{[\mathcal{A} ~ \mathcal{F}]}^T X_\mathcal{A} = 0 \label{equ:HXJX}
\end{align}
Since $J_{[\mathcal{A} ~ \mathcal{F}]} X_\mathcal{F} = 0$, if $X_\mathcal{F} \neq \boldsymbol{0}$ we must have $X_\mathcal{F}^T H_{\mathcal{F} \times \mathcal{F}} X_\mathcal{F} < 0$ according to \eqref{equ:vfhffvf}. We multiply both sides of \eqref{equ:HXJX} to obtain 
\begin{equation*}
    X_\mathcal{F}^T H_{\mathcal{F} \times \mathcal{F}} X_\mathcal{F} + X_\mathcal{F}^T J_{[\mathcal{A} ~ \mathcal{F}]}^T X_\mathcal{A} = 0 \Rightarrow  X_\mathcal{F}^T H_{\mathcal{F} \times \mathcal{F}} X_\mathcal{F} = 0,
\end{equation*}
which is not possible. As a result $X_\mathcal{F} = \boldsymbol{0}$ and $J_{[\mathcal{A} ~ \mathcal{F}]}^T X_\mathcal{A} = \boldsymbol{0}$. If $X_\mathcal{A} \neq \boldsymbol{0}$, it violates regularity of $(\hat{\boldsymbol{Y}}, \hat{\boldsymbol{R}})$. So $X_\mathcal{F} = \boldsymbol{0}$ and $X_\mathcal{A} = \boldsymbol{0}$. As a result the matrix in \eqref{equ:matrixone} is non-singular. Similar to Conn et. al. \cite{conn1997globally} we define 
\begin{align*}
    &g_l(\boldsymbol{R}, \boldsymbol{Y}, \boldsymbol{\mu}, \boldsymbol{\gamma}) \triangleq \\ & F(\boldsymbol{R}) +\sum_{(b,a) \in \mathcal{E}}{\mu_{ba}  g_{ba}(\boldsymbol{Y}, \boldsymbol{R})} - \sum_{v \in \mathcal{V}} \gamma_v  g_{v}(\boldsymbol{Y})
\end{align*}
Due to KKT conditions we have
\begin{align*}
    & \frac{\partial g_l(\boldsymbol{R}, \boldsymbol{Y}, \boldsymbol{\mu}, \boldsymbol{\gamma})}{\partial r_{\requestindex}} + \hat{\eta}_\requestindex - \hat{\eta'}_{\requestindex} = 0 \\
    & \frac{\partial g_l(\boldsymbol{R}, \boldsymbol{Y}, \boldsymbol{\mu}, \boldsymbol{\gamma})}{\partial y_{vi}} + \hat{\xi}_{vi} - \hat{\xi'}_{vi} = 0 
\end{align*}
if $\frac{\partial g_l(\boldsymbol{R}, \boldsymbol{Y}, \boldsymbol{\mu}, \boldsymbol{\gamma})}{\partial r_{\requestindex}} = 0$, we have $\hat{\eta}_{\requestindex} - \hat{\eta'}_{\requestindex} = 0$. Lagrange multipliers $\hat{\eta}_{\requestindex}$ and $\hat{\eta'}_{\requestindex}$ correspond to lower bound and upper bound constraints on the variable $r_{\requestindex}$ respectively. Due to complementary slackness, both of them cannot be positive since a variable cannot be equal to its lower bound and its upper bound at the same time. Therefore, both are zero. If they are both zero, it means neither of upper bound and lower bound constraints are active,  due to the strict complementary slackness in the second-order sufficient conditions. Hence, if $\frac{\partial g_l(\boldsymbol{R}, \boldsymbol{Y}, \boldsymbol{\mu}, \boldsymbol{\gamma})}{\partial r_{\requestindex}} = 0$, then 
\begin{equation*}
     \bar{\lambda}_{\requestindex} > r_{\requestindex} > 0.
\end{equation*}
The same is true for cache allocation variables $y_{vi}$. Therefore, the set $\mathcal{J}$ in Assumption 5 of Conn et. al. \cite{conn1997globally} is exactly the set $\mathcal{F}$ of variables which are not on their bounds. As a result, the matrix defined in Assumption 5 of Conn et. al. \cite{conn1997globally} is equivalent to the matrix \eqref{equ:matrixone}. Hence, if the second order sufficient condition and the regularity assumption hold, Assumption 5 of Conn et. al. \cite{conn1997globally} holds automatically. 
\end{proof}

According to Lemma~\ref{lem:ass5ofconn},  if the second order sufficient condition and the regularity assumption hold, Assumption 5 of Conn et. al. \cite{conn1997globally} holds automatically. In addition, Assumption~\ref{assm:lipsh} is exactly the Assumption 4 of Conn et. al. \cite{conn1997globally}, and the single limit point assumption is the same as Assumption 6 of Conn et. al. \cite{conn1997globally}. The strict complementary slackness is the same assumption made in part (ii) of Theorem 5.3 of Conn et. al. \cite{conn1997globally}. The proper choice of parameters is defined by Conn et. al. \cite{conn1997globally} as follows:
\begin{equation*}
    \alpha \triangleq \min(1,\alpha_w) \text{ and } \beta \triangleq \min (1, \beta_w)
\end{equation*}
and whenever $\alpha_\eta$ and $\beta_\eta$ satisfy the conditions
\begin{align*}
    \alpha_\eta &< \min(1, \alpha_w) \\
    \beta_\eta &< \beta \\
    \alpha_\eta + \beta_\eta &< \alpha + 1,
\end{align*}
Thus, assumptions of Proposition~\ref{prop:k0} implies all the assumptions for part (ii) of Theorem 5.3 and Corollary 5.7 of Conn et. al. \cite{conn1997globally}. Hence, the R-linear convergence results stated in Corollary 5.7 of Conn et. al. \cite{conn1997globally} hold here. }

\fullversion{}{\section{Proof of Corollary~\ref{col:sandwich}} \label{append:proofcolsandwich}
Based on Lemma~\ref{lem:concavebiconjugate} we have
\begin{align*}
  & (1-1/e) \min\{ 1, \frac{r_{(i,p)}}{\bar{\lambda}_{(i,p)}} + \sum_{k=1}^{a} y_{p_k i}\} \\ \leq & 1- (1-\frac{r_{(i,p)}}{\bar{\lambda}_{(i,p)}}) \prod_{v=p_1}^{a} (1-y_{vi}) \leq \min\{ 1, \frac{r_{(i,p)}}{\bar{\lambda}_{(i,p)}} + \sum_{k=1}^{a} y_{p_k i}\}.
\end{align*}
multiplying both sides by $\bar{\lambda}_{(i,p)}$, we have
\begin{align*}
   & (1-1/e) \bar{\lambda}_{(i,p)} \min\{ 1, \frac{r_{(i,p)}}{\bar{\lambda}_{(i,p)}} + \sum_{k=1}^{a} y_{p_k i}\} \\
  \leq & \bar{\lambda}_{(i,p)} - (\bar{\lambda} _{(i,p)}-r_{(i,p)}) \prod_{v=p_1}^{a} (1-y_{vi}) \\
  \leq & \bar{\lambda}_{(i,p)} \min\{ 1, \frac{r_{(i,p)}}{\bar{\lambda}_{(i,p)}} + \sum_{k=1}^{a} y_{p_k i}\}.
 \end{align*}
 Summing over all $(i,p) : (a,b) \in p$, we have
 \begin{align*}
   &  (1-1/e) \sum_{(i,p): (a,b) \in p} \bar{\lambda}_{(i,p)} \min\{ 1, \frac{r_{(i,p)}}{\bar{\lambda}_{(i,p)}} + \sum_{k=1}^{a} y_{p_k i}\}  \\
  \leq & \sum_{(i,p): (a,b) \in p} \bar{\lambda}_{(i,p)} - (\bar{\lambda}_{(i,p)}-r_{(i,p)}) \prod_{v=p_1}^{a} (1-y_{vi}) \\
  \leq & \sum_{(i,p): (a,b) \in p} \bar{\lambda}_{(i,p)} \min\{ 1, \frac{r_{(i,p)}}{\bar{\lambda}_{(i,p)}} + \sum_{k=1}^{a} y_{p_k i}\},
\end{align*}
and this concludes the proof.}

\bibliographystyle{IEEEtran}
\bibliography{bibliography}

\end{document}